\documentclass[a4paper,11pt]{article}
\usepackage{fullpage}

\def\llncs{0}											
\def\LNCSpreview{0}								
\def\papertype{0}								   
\def\fullpage{1}									 
\def\pagelimit{}									 
\def\anonymous{0}								 
\def\acknowledgments{0}						 
\def\overflow{0}									 
\def\showlabels{0}								   
\def\authnotes{0}									
\def\dieordollar{0}									
\def\notxfont{0}
\def\nomathpazo{0}
\def\stuffedtitlepage{0}						 
\def\abbrevref{0}									
\def\dropargs{0}									
\def\rmvtheoremspace{0}						 
\def\allowbreaks{0}								  
\def\choosebibstyle{1}							
\def\changefont{0}								  
\def\submission{0}
\def\cameraready{0}

\ifnum\submission=1
\def\llncs{1}
\def\anonymous{1}
\def\authnotes{0}
\def\choosebibstyle{0}
\def\nomathpazo{1}
\else
\fi

\ifnum\cameraready=1
\def\submission{1}
\def\llncs{1}
\def\authnotes{0}
\def\choosebibstyle{0}
\def\acknowledgments{1}
\def\nomathpazo{1}
\else
\fi

\ifnum\authnotes=0  
\newcommand{\fuyuki}[1]{}
\newcommand{\ryo}[1]{}
\newcommand{\takashi}[1]{}
\else
\newcommand{\fuyuki}[1]{$\ll$\textsf{\color{blue} Fuyuki: { #1}}$\gg$}
\newcommand{\ryo}[1]{$\ll$\textsf{\color{red} Ryo: { #1}}$\gg$}
\newcommand{\takashi}[1]{$\ll$\textsf{\color{orange} Takashi: { #1}}$\gg$}
\fi

\def\confvers{0}										   

\ifnum\llncs=1
\def\titletext{											
Secure Software Leasing from Standard Assumptions
}
\def\runningtitle{									
}
\else
\def\titletext{											
Secure Software Leasing from Standard Assumptions
}
\def\runningtitle{									
}
\fi

\date{}										  

\def\choosepubinfo{5}						   

\def\pubinfoYEAR{2016}								
\def\pubinfoSUBMISSIONDATE{}		  
\def\pubinfoDOI{ }								 
\def\pubinfoBIBDATA{ }						 
\def\pubinfoCONFERENCE{CRYPTO}				

\def\pubinfoindividual{

}

\makeatletter
\expandafter\newcommand{\createauthor}[5]{%
	\@namedef{#1name}{#2}%
	\@namedef{#1running}{#3}%
	\@namedef{#1institute}{#4}%
	\@namedef{#1thanks}{#5}%
}

\expandafter\newcommand{\createinstitute}[4]{%
	\@namedef{#1instname}{#2}%
	\@namedef{#1mail}{#3}%
	\@namedef{#1number}{#4}%
}
\makeatother

\newcounter{authorcount}
\newcommand{\newauthor}[4]{
	\stepcounter{authorcount}
	\createauthor{\theauthorcount}{#1}{#2}{#3}{#4}
}

\newcounter{institutecount}
\newcommand{\newinstitute}[3]{
	\stepcounter{institutecount}
	\createinstitute{\theinstitutecount}{#1}{#2}{#3}
}

\newauthor{Fuyuki~Kitagawa}{F.~Kitagawa}{1}{}
\newauthor{Ryo~Nishimaki}{R.~Nishimaki}{1}{}
\newauthor{Takashi~Yamakawa}{T.~Yamakawaf}{1}{}

\newinstitute{NTT Corporation, Tokyo, Japan}{\{fuyuki.kitagawa.yh,ryo.nishimaki.zk,takashi.yamakawa.ga\}@hco.ntt.co.jp}{1}

\def\contactmail{}

\def\keywords{
secure software leasing, learning with errors, classical communication
}

\def\authorsofpdf{
\ifcsname 5name\endcsname
	\csname 1name\endcsname~et~al.
\else
	\ifcsname 1name\endcsname
		\csname 1name\endcsname
		\ifcsname 3name\endcsname
			,\ \csname 2name\endcsname
			\ifcsname 4name\endcsname
				,\ \csname 3name\endcsname\ and \csname 4name\endcsname
			\else%
				\ and\ \csname 3name\endcsname
			\fi
		\else
			\ifcsname 2name\endcsname
				\csname 2name\endcsname
			\fi
		\fi
	\fi
\fi
}
\ifnum\LNCSpreview=1
	\def\llncs{1}
\fi

\let\accentvec\vec

\ifnum\llncs=1
\else
	\ifnum\papertype=0
	\else
	\fi

	\ifnum\fullpage=1
		\usepackage{fullpage}
	\fi
\fi

\let\vec\accentvec

\usepackage[dvipsnames]{xcolor}
\definecolor{darkblue}{rgb}{0,0,0.6}
\definecolor{darkgreen}{rgb}{0,0.5,0}
\definecolor{darkviolet}{RGB}{130,95,141}

\usepackage
[pdfpagelabels=true,
linktocpage=false,
bookmarks=true,bookmarksnumbered=false,bookmarkstype=toc,
pagebackref,
colorlinks=true,linkcolor=darkblue,urlcolor=darkblue,citecolor=darkgreen]
{hyperref}

\usepackage{amscd,amsmath,amssymb,amsfonts}

\usepackage{amsthm}

\ifnum\abbrevref=0
	\usepackage[capitalise,noabbrev]{cleveref}
\else
	\usepackage[capitalise]{cleveref}
\fi
\usepackage[absolute]{textpos}
\usepackage[final]{microtype}
\usepackage[absolute]{textpos}
\usepackage{everypage}

\ifnum\dieordollar=2
	\usepackage{tikz}
\else
	\ifnum\dieordollar=1
		\usepackage{tikz}
	\fi
\fi

\ifnum\LNCSpreview=1
	\usepackage[paperwidth=152mm,paperheight=235mm,textwidth=122mm,textheight=193mm]{geometry}
	\AtBeginDocument{
		\begin{textblock}{5.25}[-0.59,0](0,12)
			\centering\Large
			\textcolor{red}{\textbf{LNCS Preview mode active.}}
		\end{textblock}
	}
	\def\choosepubinfo{0}
\fi

\ifnum\overflow=1
\fi

\ifnum\llncs=0
	\renewcommand*{\backref}[1]{(Cited on page~#1.)}
\fi

\ifnum\allowbreaks=1
	\allowdisplaybreaks
\fi

\ifnum\showlabels=1
	\usepackage{showkeys}
\fi

\ifnum\anonymous=1
	\def\authnotes{0}
\fi

\ifnum\changefont=1
	\usepackage{times}
\fi

\ifnum\LNCSpreview=0
	\ifx\pagelimit\empty
	\else
		\newcounter{pagewarning}
		\AddEverypageHook{
			\ifnum\thepage>\thepagewarning
				\begin{textblock}{1}[0.5,0](0,-13)
					\centering\Large
					\textcolor{red}{\textbf{This page is exceeding the page limit.}}
				\end{textblock}
			\fi
	}
	\fi
\fi

\newcommand{\authnote}[2]{\ifnum\authnotes=1 \begin{center}\fbox{\begin{minipage}{.98\textwidth}
\textbf{#1 says:} #2\end{minipage}}\end{center} \fi}

\newlength{\strutdepth}%
\newcommand{\notes}[3]{
\ifnum\authnotes=1
	\noindent{\bfseries
	\color{#1}{#3}\color{#1}}%
	\strut\vadjust{\kern-\strutdepth%
		\vtop to \strutdepth{%
			\baselineskip\strutdepth%
			\vss\llap{{\large\color{#1}\textbf{#2}\quad\color{black}}}\null%
		}%
	}%
\fi
}

\ifnum\LNCSpreview=0

\else

\fi

\mathchardef\hyphen="2D

\newtheoremstyle{thicktheorem}%
{\ifnum\rmvtheoremspace=1
	0.4
\fi
\topsep}
{\ifnum\rmvtheoremspace=1
	0.4
\fi
\topsep}
{\itshape}{}%
{\bfseries}%
{.}
{ }%
{\thmname{#1}\thmnumber{ #2}%
	\ifnum\dropargs=0
		\thmnote{ (#3)}%
	\fi
}

\newtheoremstyle{remark}
{\ifnum\rmvtheoremspace=1
	0.4
\fi
\topsep}
{\ifnum\rmvtheoremspace=1
	0.4
\fi
\topsep}
	{}
	{}
	{}
	{.}
	{ }
	{\textit{\thmname{#1}}\thmnumber{ #2}
		\ifnum\dropargs=0
			\thmnote{ (#3)}%
		\fi
	}

\ifnum\llncs=0
	\theoremstyle{thicktheorem}
	\newtheorem{theorem}{Theorem}[section]
	\newtheorem{lemma}[theorem]{Lemma}
	\newtheorem{corollary}[theorem]{Corollary}
	\newtheorem{proposition}[theorem]{Proposition}
	\newtheorem{definition}[theorem]{Definition}

	\theoremstyle{remark}
	\newtheorem{claim}[theorem]{Claim}
	\newtheorem{remark}[theorem]{Remark}

	\newtheorem{conjecture}[theorem]{Conjecture}
\fi

\theoremstyle{remark}
\newtheorem{assumption}[theorem]{Assumption}
\newtheorem{observation}[theorem]{Observation}
\newtheorem{fact}[theorem]{Fact}
\newtheorem{experiment}{Experiment}
\newtheorem{construction}[theorem]{Construction}
\newtheorem{counterexample}[theorem]{Counterexample}

\ifnum\abbrevref=0
	\crefname{assumption}{Assumption}{Assumptions}
	\crefname{construction}{Construction}{Constructions}
	\crefname{corollary}{Corollary}{Corollaries}
	\crefname{conjecture}{Conjecture}{Conjectures}
	\crefname{definition}{Definition}{Definitions}
	\crefname{exmaple}{Example}{Examples}
	\crefname{experiment}{Experiment}{Experiments}
	\crefname{counterexample}{Counterexample}{Counterexamples}
	\crefname{lemma}{Lemma}{Lemmata}
	\crefname{observation}{Observation}{Observations}
	\crefname{proposition}{Proposition}{Propositions}
	\crefname{remark}{Remark}{Remarks}
	\crefname{theorem}{Theorem}{Theorems}
\else
	\crefname{assumption}{Ass.}{Ass.}
	\crefname{construction}{Constr.}{Constr.}
	\crefname{corollary}{Cor.}{Cor.}
	\crefname{conjecture}{Conj.}{Conj.}
	\crefname{definition}{Def.}{Def.}
	\crefname{exmaple}{Ex.}{Ex.}
	\crefname{experiment}{Exp.}{Exp.}
	\crefname{counterexample}{Counterex.}{Counterex.}
	\crefname{lemma}{Lem.}{Lem.}
	\crefname{observation}{Obs.}{Obs.}
	\crefname{proposition}{Prop.}{Prop.}
	\crefname{remark}{Rem.}{Rem.}
	\crefname{theorem}{Thm.}{Thms.}
\fi

\crefname{claim}{Claim}{Claims}
\crefname{fact}{Fact}{Facts}
\crefname{note}{Note}{Notes}

\def\YYYSMcoin{\mbox{\begin{tikzpicture}[scale=0.0125]
\definecolor{coinbrown}{HTML}{D89E36}\definecolor{coindarkyellow}{HTML}{F8D81E}\definecolor{coinyellow}{HTML}{F8F800}\fill[coinyellow] (3,-1) rectangle (9,9);\fill(0,0) rectangle (1,8);\fill(1,8) rectangle (2,10);\fill(2,10) rectangle (4,11);\fill(4,11) rectangle (8,12);\fill(8,11) rectangle (10,10);\fill(10,10) rectangle (11,8);\fill(11,8) rectangle (12,0);\fill(10,-2) rectangle (11,0);\fill(8,-3) rectangle (10,-2);\fill(4,-4) rectangle (8,-3);\fill(2,-3) rectangle (4,-2);\fill(1,0) rectangle (2,-2);\fill (5,-1) rectangle (7,0);\fill (7,0) rectangle (8,8);\fill[coinbrown] (9,8) rectangle (10,10);\fill[coinbrown] (10,0) rectangle (11,8);\fill[coinbrown] (9,-2) rectangle (10,0);\fill[coinbrown] (8,-2) rectangle (9,-1);\fill[coinbrown] (4,-3) rectangle (8,-2);\fill[coindarkyellow] (2,-2) rectangle (3,8);\fill[coindarkyellow] (3,-2) rectangle (8,-1);\fill[coindarkyellow] (8,-1) rectangle (9,0);\fill[coindarkyellow] (9,0) rectangle (10,8);\fill[coindarkyellow] (8,8) rectangle (9,10);\fill[coindarkyellow] (4,9) rectangle (8,10);\fill[coindarkyellow] (3,8) rectangle (4,9);\fill[coindarkyellow] (5,0) rectangle (7,2);\fill[coindarkyellow] (6,2) rectangle (7,8);\fill[white] (4,0) rectangle (5,8);\fill[white] (5,8) rectangle (7,9);
\end{tikzpicture}}}

\def\YYYdie{\mbox{\begin{tikzpicture}[scale=0.85,x=1em,y=1em,radius=0.09]
\draw[rounded corners=1,line width=.25pt] (0,0) rectangle (1,1);\fill (0.275,0.275) circle;\fill (0.725,0.725) circle;\fill (0.5,0.5) circle;
\end{tikzpicture}}}

\newcommand{\getsr}{
	\ifnum\dieordollar=0
		\mathrel{\vbox{\offinterlineskip\ialign{
			\hfil##\hfil\cr
			\hspace{0.1em}$\scriptscriptstyle\$$\cr
			$\leftarrow$\cr
		}}}
	\fi
	\ifnum\dieordollar=1
		\mathrel{\vbox{\offinterlineskip\ialign{
			\hfil##\hfil\cr
			{\scalebox{0.5}{\hspace{0.4em}\YYYdie}}\cr
			\noalign{\kern0.05ex}
			$\leftarrow$\cr
		}}}
	\fi
	\ifnum\dieordollar=2
		\mathrel{\vbox{\offinterlineskip\ialign{
			\hfil##\hfil\cr
			\hspace{0.1em}$\YYYSMcoin$\cr
			$\leftarrow$\cr
		}}}
	\fi
}

\def\makeuppercase#1{
\expandafter\newcommand\csname sf#1\endcsname{\mathsf{#1}}
\expandafter\newcommand\csname frak#1\endcsname{\mathfrak{#1}}
\expandafter\newcommand\csname bb#1\endcsname{\mathbb{#1}}
\expandafter\newcommand\csname bf#1\endcsname{\textbf{#1}}
}

\def\makelowercase#1{
\expandafter\newcommand\csname frak#1\endcsname{\mathfrak{#1}}
\expandafter\newcommand\csname bf#1\endcsname{\textbf{#1}}
}

\newcounter{char}
\loop
	\edef\letter{\alph{char}}
	\edef\Letter{\Alph{char}}
	\expandafter\makelowercase\letter
	\expandafter\makeuppercase\Letter
	\stepcounter{char}
	\unless\ifnum\thechar>26
\repeat

\usepackage{lmodern}
\usepackage[T1]{fontenc}
\usepackage[utf8]{inputenc}
\usepackage{etex}
\usepackage{graphicx,pstricks}
\usepackage{tikz}
\usetikzlibrary{matrix,arrows}
\usetikzlibrary{positioning}
\usetikzlibrary{decorations,decorations.text}
\usetikzlibrary{decorations.pathmorphing}
\usetikzlibrary{shapes}
\usepackage{centernot}
\usepackage{xspace}
\usepackage{pifont}
\usepackage{paralist}
\usepackage{booktabs}
\usepackage{dashbox}
\usepackage{multirow}
\usepackage{lscape}
\usepackage{amsmath}
\usepackage{fancybox}
\usepackage{braket}
\usepackage{physics}
\usepackage{mathtools}
\mathtoolsset{showonlyrefs=true}

\ifnum\submission=1
\renewcommand*{\backref}[1]{}
\def\notxfont{1}
\else
\fi

\usepackage{dsfont}
\DeclareMathAlphabet{\mathpzc}{OT1}{pzc}{m}{it}

\ifnum\llncs=1
\renewcommand{\subparagraph}{\paragraph}
\else
\ifnum\notxfont=1
\else
\usepackage{newtxtext}
\fi
\fi
\ifnum\nomathpazo=1
\else
\usepackage{mathpazo}
\usepackage{helvet}
\fi
\ifnum\showlabels=1
\usepackage{showkeys}
\else
\fi

\makeatletter
\newcounter{game}
\def\newGames#1#2#3{%
  \xdef\gameNS{#1}\xdef\gamePrefix{#2}\setcounter{game}{#3}\addtocounter{game}{-1}%
  \immediate\write\@auxout{\string\expandafter\string\xdef\noexpand\csname game-prefix-#1\string\endcsname{#2}}%
}
\def\newGame#1{%
  \xdef\prevGame{\gamePrefix\arabic{game}}\stepcounter{game}\xdef\thisGame{\gamePrefix\arabic{game}}%
  \immediate\write\@auxout{\string\expandafter\string\xdef\noexpand\csname game-\gameNS-#1\string\endcsname{\arabic{game}}}%
}
\makeatother
\def\safecsname#1{\expandafter\ifx\csname#1\endcsname\relax\else\csname#1\endcsname\fi}

\renewcommand{\Game}[1][]{\mathcmd{\textrm{Game\if!#1!\else~\ensuremath{#1}\fi}}}

\newcommand{\td}{\mathsf{td}}

\usepackage[font=small,
format=plain,
labelformat=simple, 
singlelinecheck=false,
labelfont=bf,
up
]{caption}
\DeclareCaptionFormat{myformat}{#1#2#3\hrulefill}
\newcommand{\sslgen}{\algo{Gen}}
\newcommand{\lessor}{\qalgo{Lessor}}
\newcommand{\lessee}{\qalgo{Lessee}}
\newcommand{\run}{\qalgo{Run}}
\newcommand{\sslcheck}{\qalgo{Check}}

\newcommand{\sft}{\qstate{sft}}
\newcommand{\ssl}{\mathsf{ssl}}

\newcommand{\adjmat}[1]{{#1}^{\dagger}}

\newcommand{\trdist}[1]{{\norm{#1}_{\tr}}}

\newcommand{\qb}{\qstate{b}}

\newcommand{\csft}{\keys{sft}}
\newcommand{\obligation}{\mathsf{obligation}}
\newcommand{\answer}{\mathsf{answer}}
\newcommand{\statelessee}{\qstate{st}_{\lessee}}
\newcommand{\clessor}{\mathsf{Lessor}}
\newcommand{\sslcert}{\qalgo{SSLCert}}

\newcommand{\trapgen}{\algo{TrapGen}}

\usepackage{mathtools}

\newcommand{\chosen}{\leftarrow}
\newcommand{\sample}{\leftarrow}
\renewcommand{\gets}{\leftarrow}
\newcommand{\out}{=}

\newcommand{\la}{\leftarrow}
\newcommand{\ra}{\rightarrow}

\newcommand{\seteq}{\coloneqq}

\newcommand{\tensor}{\otimes}
\newcommand{\concat}{\|}

\newcommand{\setbracket}[1]{\{#1\}}

\newcommand{\cA}{\mathcal{A}}

\newcommand{\cC}{\mathcal{C}}
\newcommand{\cD}{\mathcal{D}}

\newcommand{\cF}{\mathcal{F}}

\newcommand{\cH}{\mathcal{H}}

\newcommand{\cK}{\mathcal{K}}

\newcommand{\cM}{\mathcal{M}}

\newcommand{\cR}{\mathcal{R}}
\newcommand{\cS}{\mathcal{S}}

\newcommand{\cU}{\mathcal{U}}

\newcommand{\cX}{\mathcal{X}}
\newcommand{\cY}{\mathcal{Y}}

\def\tlC{\widetilde{C}}

\newcommand{\N}{\mathbb{N}}
\newcommand{\Z}{\mathbb{Z}}

\newcommand{\Zq}{\mathbb{Z}_q}

\newcommand{\M}{\cM}

\newcommand{\Params}{\algo{\Params}}

\newcommand{\sep}{\lambda}

\newcommand{\secp}{\lambda}

\newcommand{\crs}{\mathsf{crs}}
\newcommand{\aux}{\mathsf{z}}

\newcommand{\A}{\entity{A}}

\newcommand{\LWE}{\textrm{LWE}}
\newcommand{\SIS}{\textrm{SIS}}

\newcommand{\adva}[2]{\mathsf{Adv}_{#1}^{\mathsf{#2}}}

\newcommand{\expt}[2]{\mathsf{Expt}_{#1}^{\mathsf{#2}}}
\newcommand{\expa}[3]{\mathsf{Exp}_{#1}^{ \mathsf{#2} \mbox{-} \mathsf{#3}}}
\newcommand{\expb}[4]{\mathsf{Exp}_{#1}^{ \mathsf{#2} \mbox{-} \mathsf{#3} \mbox{-} \mathsf{#4}}}

\newcommand*{\sk}{\keys{sk}}
\newcommand*{\pk}{\keys{pk}}

\newcommand*{\pp}{\keys{pp}}

\newcommand*{\msg}{\keys{m}}

\newcommand{\prfkey}{\keys{K}}

\newcommand*{\keys}[1]{\mathsf{#1}}
\newcommand*{\qstate}[1]{\mathpzc{#1}}

\newcommand*{\algo}[1]{\ensuremath{\mathsf{#1}}}
\newcommand*{\qalgo}[1]{\ensuremath{\mathpzc{#1}}}

\newcommand*{\entity}[1]{\mathcal{#1}}

\newenvironment{boxfig}[2]{\begin{figure}[#1]\fbox{\begin{minipage}{0.97\linewidth}
                        \vspace{0.2em}
                        \makebox[0.025\linewidth]{}
                        \begin{minipage}{0.95\linewidth}
            {{
                        #2 }}
                        \end{minipage}
                        \vspace{0.2em}
                        \end{minipage}}}{\end{figure}}

\newcommand{\bit}{\{0,1\}}

\newcommand{\mm}[1]{\boldsymbol{#1}}
\newcommand{\mv}[1]{\boldsymbol{#1}}

\newcommand{\prf}{\algo{F}}

\newcommand{\Setup}{\algo{Setup}}
\newcommand{\Gen}{\algo{Gen}}
\newcommand{\gen}{\algo{Gen}}

\newcommand{\vrfy}{\algo{Vrfy}}

\newcommand{\Prove}{\algo{Prove}}
\newcommand{\prove}{\algo{Prove}}

\newcommand{\NIZK}{\mathsf{NIZK}}

\newcommand{\PRF}{\algo{PRF}}

\newcommand{\Eval}{\algo{Eval}}
\newcommand{\pEval}{\algo{pEval}}
\newcommand{\Puncture}{\algo{Puncture}}

\newcommand{\WM}{\mathsf{WM}}

\newcommand{\Mark}{\mathsf{Mark}}

\newcommand{\Extract}{\mathsf{Extract}}

\newcommand{\unmarked}{\mathsf{unmarked}}

\newcommand{\setup}{\mathsf{Setup}}

\newcommand{\negl}{{\mathsf{negl}}}

\newcommand{\zo}[1]{\{0,1\}^{#1}}

\DeclareMathOperator*{\Exp}{\mathbb{E}}

\newcommand{\xor}{\oplus}

\newcommand{\class}[1]{\mathsf{#1}}

\newcommand{\NP}{\class{NP}}

\newcommand{\Hmin}{\mathrm{H}_{\infty}}

\newcommand{\B}{\mathcal{B}}

\newcommand{\calK}{\mathcal{K}}

\newcommand{\calU}{\mathcal{U}}

\newcommand{\Sim}{\algo{Sim}}

\newcommand{\PuncPRF}{\algo{PPRF}}

\newcommand{\state}{\mathsf{st}}

\newcommand{\Dowf}{D_{\mathsf{owf}}}
\newcommand{\Rowf}{R_{\mathsf{owf}}}
\newcommand{\Fow}{\cF_{\mathsf{ow}}}

\newcommand{\owf}{f}

\newcommand{\cnc}[2]{\mathbf{C}\{#1,#2\}}
\newcommand{\Ccnc}{\cC_{\mathsf{cnc}}^{\inplen,\outlen}}
\newcommand{\Dcnc}[1]{\cD_{#1\textrm{-}\mathsf{cnc}}}

\newcommand{\prfkeyspace}{\calK}
\newcommand{\Dprf}{\zo{n}}
\newcommand{\Rprf}{\zo{m}}

\newcommand{\fksetup}{\mathsf{FkSetup}}

\newcommand{\MAC}{\mathsf{MAC}}
\newcommand{\Macgen}{\MAC.\mathsf{Gen}}
\newcommand{\Mactag}{\MAC.\mathsf{Tag}}
\newcommand{\Macvrfy}{\MAC.\mathsf{Vrfy}}
\newcommand{\mackey}{\mathsf{s}}
\newcommand{\Dmac}{D_{\mathsf{mac}}}
\newcommand{\mac}{\mathsf{tag}}

\newcommand{\event}[1]{\mathtt{#1}}
\newcommand{\qAnizk}{\qB}
\newcommand{\inplen}{n}
\newcommand{\outlen}{m}
\newcommand{\msglen}{k}
\newcommand{\first}{{(1)}}
\newcommand{\second}{{(2)}}

\newcommand{\bolt}{\qstate{bolt}}
\newcommand{\snum}{\keys{snum}}
\newcommand{\semivrfy}{\qalgo{SemiVrfy}}
\newcommand{\fullvrfy}{\qalgo{FullVrfy}}
\newcommand{\boltgen}{\qalgo{BoltGen}}

\newcommand{\boltcert}{\qalgo{BoltCert}}
\newcommand{\cert}{\keys{cert}}
\newcommand{\ttQL}{\mathsf{ttQL}}
\newcommand{\certvrfy}{\algo{CertVrfy}}

\begin{document}

\newcount\authorcounter
\newcommand{\provideauthors}{%
		\ifnum\authorcounter<\theauthorcount
			\csname\the\authorcounter name\endcsname
			\expandafter\ifx\csname\the\authorcounter thanks\endcsname\empty
			\else
				\thanks{\csname\the\authorcounter thanks\endcsname}
			\fi%
			\inst{\csname\the\authorcounter institute\endcsname} 
			\and 
			\global\advance\authorcounter by 1 
			\provideauthors
		\else
			\csname\the\authorcounter name\endcsname 
			\expandafter\ifx\csname\the\authorcounter thanks\endcsname\empty 
			\else
				\thanks{\csname\the\authorcounter thanks\endcsname} 
			\fi%
			\inst{\csname\the\authorcounter institute\endcsname} 
		\fi
}

\def\atleastoneauthorplaced{0}
\newcommand{\providerunning}{%
	\ifnum\authorcounter<\theauthorcount%
		\expandafter\ifx\csname\the\authorcounter running\endcsname\empty
		\else
			\ifnum\authorcounter>1
				\ifnum\atleastoneauthorplaced=1
					\and%
				\fi
			\fi
			\csname\the\authorcounter running\endcsname
			\def\atleastoneauthorplaced{1}
		\fi
		\global\advance\authorcounter by 1
		\providerunning%
	\else%
		\expandafter\ifx\csname\the\authorcounter running\endcsname\empty
		\else
			\ifnum\authorcounter>1
				\ifnum\atleastoneauthorplaced=1
					\and%
				\fi
			\fi
			\csname\the\authorcounter running\endcsname
		\fi
	\fi
}

\newcount\institutecounter

\newcommand{\provideinstitutes}{%
	\ifnum\institutecounter<\theinstitutecount%
		\ifnum\llncs=0
			$^{\csname\the\institutecounter number\endcsname}$
		\fi
		\csname\the\institutecounter instname\endcsname
		
		\email{
			\ifx\contactmail\empty
				\csname\the\institutecounter mail\endcsname
			\else
				\href{mailto:\contactmail}{\csname\the\institutecounter mail\endcsname}
			\fi
		}
		
		\and%
			\global\advance\institutecounter by 1
		\provideinstitutes%
	\else%
		\ifnum\llncs=0
			\ifcsname 1name\endcsname
				$^{\csname\the\institutecounter number\endcsname}$
			\fi
		\fi
		\csname\the\institutecounter instname\endcsname
		
		\email{
			\ifx\contactmail\empty
				\csname\the\institutecounter mail\endcsname
			\else
				\href{mailto:\contactmail}{\csname\the\institutecounter mail\endcsname}
			\fi
		}
	\fi
}

\title{
	\ifnum\stuffedtitlepage=1
		\ifnum\llncs=1
			\vspace*{-7ex}
		\else
		\vspace*{-3ex}
		\fi
		\textbf{\titletext}
		\ifnum\llncs=1
			\vspace*{-2ex}
		\else
			\vspace*{-1ex}
		\fi
	\else
		\textbf{\titletext}
	\fi
}
\ifnum\anonymous=1
	\author{}
\else
	\ifnum\llncs=0
		\newcommand{\inst}[1]{{
			\ifcsname 1name\endcsname
				$^{#1}$
			\fi
			}}
	\fi
	\ifcsname 1name\endcsname
		\author{
			\global\authorcounter 1
			\provideauthors
		}
	\fi
\fi

\ifnum\llncs=1
	\titlerunning{\runningtitle}
	\ifnum\anonymous=1
		\institute{}
		\authorrunning{}
	\else
		\ifcsname 1instname\endcsname{
			\institute{
				\global\institutecounter 1
				\provideinstitutes
			}
		\fi
		\ifcsname 1name\endcsname{
			\authorrunning{
				\global \authorcounter 1
				\providerunning
			}
		\fi
	\fi
\fi
\maketitle
\ifnum\stuffedtitlepage=1
	\ifnum\llncs=0
		\vspace{-4ex}
	\fi
\fi

\ifnum\llncs=0
	\ifnum\anonymous=0
		\newcommand{\email}[1]{
			\texttt{
				\ifx\contactmail\empty
					#1
				\else
					\href{mailto:\contactmail}{#1}
				\fi
			}
		}
		\newcommand{\and}{}
		\ifnum\stuffedtitlepage=1
			\ifnum\llncs=0
				\vspace{-2ex}
			\fi
		\fi
		\begin{small}
			\begin{center}
				\global \institutecounter 1
				\provideinstitutes
			\end{center}
		\end{small}
	\fi
\fi

\ifnum\stuffedtitlepage=1
	\ifnum\llncs=1
		\vspace*{-4ex}
	\else
		\vspace*{-2ex}
	\fi
\fi

\begin{abstract}
	

Secure software leasing (SSL) is a quantum cryptographic primitive that enables an authority to lease software to a user by encoding it into a quantum state. SSL prevents users from generating authenticated pirated copies of leased software, where authenticated copies indicate those run on legitimate platforms. Although SSL is a relaxed variant of quantum copy protection that prevents users from generating any copy of leased softwares, it is still meaningful and attractive. Recently, Ananth and La Placa proposed the first SSL scheme. It satisfies a strong security notion called infinite-term security. On the other hand, it has a drawback that it is based on public key quantum money, which is not instantiated with standard cryptographic assumptions so far. Moreover, their scheme only supports a subclass of evasive functions.

In this work, we present SSL schemes that satisfy a security notion called finite-term security based on the learning with errors assumption (LWE). Finite-term security is weaker than infinite-term security, but it still provides a reasonable security guarantee. Specifically, our contributions consist of the following.
\begin{itemize}
 \item We construct a finite-term secure SSL scheme for pseudorandom functions from the LWE assumption against quantum adversaries.
 \item We construct a finite-term secure SSL scheme for a subclass of evasive functions from the LWE assumption against sub-exponential quantum adversaries.
 \item We construct finite-term secure SSL schemes for the functionalities above with classical communication from the LWE assumption against (sub-exponential) quantum adversaries.
 \end{itemize}
SSL with classical communication means that entities exchange only classical information though they run quantum computation locally.

Our crucial tool is two-tier quantum lightning, which is introduced in this work and a relaxed version of quantum lighting. In two-tier quantum lightning schemes, we have a public verification algorithm called semi-verification and a \emph{private} verification algorithm called full-verification. An adversary cannot generate possibly entangled two quantum states whose serial numbers are the same such that one passes the semi-verification, and the other also passes the full-verification. We show that we can construct a two-tier quantum lightning scheme from the LWE assumption.

	\vspace{1ex}
\ifnum\llncs=1
\else

	\textbf{Keywords\ifnum\llncs=1{.}\else{:}\fi}\keywords
\fi
\end{abstract}
\ifnum\stuffedtitlepage=1
	\ifnum\llncs=1
		\vspace*{-2ex}
	\fi
\fi

\ifnum\llncs=0
	\vspace{1ex}
\fi

\ifnum\choosepubinfo=1
\def\pubinfo{
	\noindent An extended abstract of this paper will appear at
	\ifx\pubinfoCONFERENCE\empty \textcolor{red}{conference missing}\else \pubinfoCONFERENCE\fi.
}
\fi

\ifnum\choosepubinfo=2
	\def\pubinfo{
		\noindent \copyright\ IACR 
		\ifx\pubinfoYEAR\empty \textcolor{red}{year missing}\else \pubinfoYEAR\fi.
		This article is the final version submitted by the author(s) to the IACR and to Springer-Verlag on
		\ifx\pubinfoSUBMISSIONDATE\empty \textcolor{red}{submission date missing}\else \pubinfoSUBMISSIONDATE\fi.
		The version published by Springer-Verlag is available at
		\ifx\pubinfoDOI\empty \textcolor{red}{DOI missing}\else \pubinfoDOI\fi.
	}
\fi

\ifnum\choosepubinfo=3
	\def\pubinfo{
		\noindent \copyright\ IACR
		\ifx\pubinfoYEAR\empty \textcolor{red}{year missing}\else \pubinfoYEAR\fi.
		This article is a minor revision of the version published by Springer-Verlag available at
		\ifx\pubinfoDOI\empty \textcolor{red}{DOI missing}\else \pubinfoDOI\fi.
	}
\fi

\ifnum\choosepubinfo=4
	\def\pubinfo{
		\noindent This article is based on an earlier article:
		\ifx\pubinfoBIBDATA\empty \textcolor{red}{bibliographic data missing}\else \pubinfoBIBDATA\fi,
		\copyright\ IACR
		\ifx\pubinfoYEAR\empty \textcolor{red}{year missing}\else \pubinfoYEAR\fi,
		\ifx\pubinfoDOI\empty \textcolor{red}{DOI missing}\else \pubinfoDOI\fi.
	}
\fi

\ifnum\choosepubinfo=5
		\def\pubinfo{
			\noindent \pubinfoindividual
		}
	\fi

\textblockorigin{0.5\paperwidth}{0.9\paperheight}
\setlength{\TPHorizModule}{\textwidth}

\newlength{\pubinfolength}
\ifnum\choosepubinfo=0
\else
	\settowidth{\pubinfolength}{\pubinfo}
	\begin{textblock}{1}[0.5,0](0,.25)
		 \ifnum\pubinfolength<\textwidth
			\centering
		\fi
		\pubinfo
	\end{textblock}
\fi
\thispagestyle{empty}
		\ifnum\confvers=1
\else
\ifnum\submission=1
\else
\newpage
  \setcounter{tocdepth}{2}      
  \setcounter{secnumdepth}{2}   
  \setcounter{page}{0}          
  \tableofcontents
  \thispagestyle{empty}
\newpage
\fi
\fi

\section{Introduction}\label{sec:intro}

\subsection{Background}\label{sec:background}
Secure software leasing (SSL) introduced by Ananth and La Placa~\cite{EC:AnaLaP21} is a quantum cryptographic primitive that enables an authority (the lessor) to lease software\footnote{Software is modeled as (Boolean) circuits or functions.} to a user (the lessee) by encoding it into a quantum state. SSL prevents users from generating authenticated pirated copies of leased software, where authenticated copies indicate those run on the legitimate platforms.

More specifically, an SSL is the following protocol between the lessor and lessee.
The lessor generates a secret key $\sk$ used to create a leased version of a circuit $C$. The leased version is a quantum state and denoted by $\sft_C$. The lessor leases the functionality of $C$ to the lessee by providing $\sft_C$. The lessee can compute $C(x)$ for any input $x$ by using $\sft_C$. That is, there exists a quantum algorithm $\run$ and it holds that $\run(\sft_C,x) = C(x)$ for any $x$. The lessor can validate the states returned from the user by using the secret key. That is, there exists a quantum algorithm $\sslcheck$ and $\sslcheck(\sk,\sft_C)$ outputs whether $\sft_C$ is a valid leased state or not. Since users can create as many copies of classical information as they want, we need the power of quantum computing to achieve SSL.

Ananth and La Placa introduced two security notions for SSL, that is, infinite-term security and finite term security.
Infinite-term security guarantees that given a single leased state of a circuit $C$, adversaries cannot generate possibly entangled bipartite states $\sft_0^\ast$ and $\sft_1^\ast$ both of which can be used to compute $C$ with $\run$.
Finite-term security guarantees that adversaries cannot generate possibly entangled bipartite states $\sft_0^\ast$ and $\sft_1^\ast$ such that $\sslcheck(\sk,\sft_0^\ast)=\top$ (returning a valid leased state) and $\run(\sft_1^\ast,x)=C(x)$ (adversary still can compute $C$ by using $\sft_1^\ast$) in an SSL scheme.
Roughly speaking, finite-term security guarantees that adversaries cannot compute $C(x)$ via $\run$ after they return the valid leased state to the lessor.

\paragraph{SSL and copy-protection.}
Quantum software copy-protection~\cite{CCC:Aaronson09} is a closely related notion to SSL.    
Quantum copy-protection guarantees the following. When adversaries are given a copy-protected circuit for computing $C$, they cannot create two (possibly entangled) quantum states, both of which can be used to compute $C$. Here, adversaries are not required to output a quantum state that follows an honest evaluation algorithm $\run$ (they can use an arbitrary evaluation algorithm $\run^\prime$).
Software copy-protection can be crucial technology to prevent software piracy since users lose software if they re-distribute it. Quantum copy-protection for some circuits class is also known to yield public-key quantum money~\cite{CoRR:AarLiuZha20}.

Although SSL is weaker than copy-protection, SSL (with even finite-term security) has useful applications such as limited-time use software, recalling buggy software, preventing drain of propriety software from malicious employees~\cite{EC:AnaLaP21}.
SSL makes software distribution more controllable. In addition, achieving SSL could be a crucial stepping stone to achieve quantum software copy-protection.

\ryo{I added the following paragraph from the EC21 rebuttal.}
One motivative example of (finite-term secure) SSL is a video game platform. A user can borrow a video game title from a company and enjoy it on an appropriate platform (like Xbox of Microsoft). After the user returned the title, s/he cannot enjoy it on the appropriate platform. The title is not guaranteed to work on another (irregular) platform. Thus, SSL is a useful tool in this use case.

\paragraph{(Im)possibility of SSL and copy-protection.}
Although SSL and software copy-protection have many useful applications, there are few positive results on them.
Aaronson observed that learnable functions could not be copy-protected~\cite{CCC:Aaronson09}. He also constructed a copy-protection scheme for arbitrary unlearnable Boolean functions relative to a quantum oracle
and two \emph{heuristic} copy-protection schemes for point functions in the standard model~\cite{CCC:Aaronson09}.
Aaronson, Liu, and Zhang constructed a quantum copy-protection scheme for unlearnable functions relative to classical oracles~\cite{CoRR:AarLiuZha20}.
There is no secure quantum copy-protection scheme with a reduction-based proof \emph{without classical/quantum oracles}. We do not know how to implement such oracles under cryptographic assumptions in the previous works.

Ananth and La Placa constructed an infinite-term secure SSL scheme for a sub-class of evasive functions in the common reference string (CRS) model by using public-key quantum money~\cite{STOC:AarChr12,JC:Zhandry21} and the learning with errors (LWE) assumption~\cite{EC:AnaLaP21}. Evasive functions is a class of functions such that it is hard to find an accepting input (a function outputs $1$ for this input) only given black-box access to a function. They also prove that there exists an unlearnable function class such that it is impossible to achieve an SSL scheme for that function class even in the CRS model. The SSL scheme by Ananth and La Placa is the only one positive result without classical/quantum oracles on this topic before our work.\footnote{We will refer to a few concurrent works in~\cref{sec:concurrent_work}.}
\fuyuki{Should be delete the last sentence?}\ryo{I added a footnote.}

\paragraph{Motivation.} 
There are many fascinating questions about SSL/copy-protection. We focus on the following three questions in this study.

The first one is whether we can achieve SSL/copy-protection from standard assumptions. Avoiding strong assumptions is desirable in cryptography. It is not known whether public-key quantum money is possible under standard assumptions. Zhandry proves that post-quantum indistinguishability obfuscation (IO)~\cite{JACM:BGIRSVY12} implies public-key quantum money~\cite{JC:Zhandry21}. Several works~\cite{TCC:CHVW19,EC:AgrPel20,TCC:BGMZ18,mySTOC:GayPas21,EPRINT:BDGM20b,EC:WeeWic21} presented candidate constructions of post-quantum secure IO by using lattices.\footnote{Their constructions need heuristic assumptions related to randomness leakage and circular security~\cite{EPRINT:BDGM20b,mySTOC:GayPas21}, a heuristic construction of oblivious LWE sampling~\cite{EC:WeeWic21}, a heuristic construction of noisy linear functional encryption~\cite{EC:AgrPel20}, or an idealized model~\cite{TCC:BGMZ18,TCC:CHVW19}. Some heuristic assumptions~\cite{mySTOC:GayPas21,EC:WeeWic21,EPRINT:BDGM20b} were found to be false~\cite{C:HopJaiLin21}.} There are several other candidate constructions of public key quantum money ~\cite{ITCS:FGHLS12,JC:Zhandry21}. However, none of them has a  reduction to standard assumptions.

The second question is whether we can achieve SSL/copy-protection only with classical communication and local quantum computing as in the case of quantum money~\cite{AFT:RadSat19,STOC:AGKZ20}.   
Even if quantum computers are available, communicating only classical data is much easier than communicating quantum data over quantum channels. Communication infrastructure might not be updated to support quantum data soon, even after practical quantum computers are commonly used.

The third question is whether we can achieve SSL/copy-protection beyond for evasive functions. The function class is quite limited. For practical software protection, it is crucial to push the function class's boundaries where we can achieve SSL/copy-protection.

\subsection{Our Results}\label{sec:our_result}
We constructed finite-term secure SSL schemes from standard assumptions in this study. We prove the following theorems. 
\begin{theorem}[informal]\label{thm:informal_main_one}
Assuming the hardness of the LWE problem against polynomial time quantum adversaries, there is a finite-term secure SSL scheme and SSL scheme with classical communication for pseudorandom functions (PRFs) in the CRS model.
\end{theorem}

\begin{theorem}[informal]\label{thm:informal_main_two}
Assuming the hardness of the LWE problem against sub-exponential time quantum adversaries, there is a finite-term secure SSL scheme and SSL scheme with classical communication for a subclass of evasive functions in the CRS model.
\end{theorem}

The notable features of our SSL schemes are the following.
\begin{itemize}
\item Constructed via a clean and unified framework.
\item Secure under standard assumptions (the LWE assumption).
\item Can be achieved only with classical communication.
\item Supporting functions other than a sub-class of evasive functions.
\end{itemize}

The crucial tools in our framework are two-tier quantum lighting, which we introduce in this study, and (a relaxed version of) software watermarking~\cite{JACM:BGIRSVY12,SIAMCOMP:CHNVW18}. Two-tier quantum lighting is a weaker variant of quantum lighting~\cite{JC:Zhandry21}. Interestingly, two-tier quantum lightning can be instantiated with standard assumptions, while quantum lightning is not so far. Another exciting feature is that software watermarking can be a building block of SSL. Our study gives a new application of software watermarking. By using these tools, our SSL constructions are modular, and we obtain a clean perspective to achieve SSL.
Our abstracted construction ensures that a relaxed watermarking scheme for any circuit class can be converted to SSL for the same class assuming the existence of two-tier QL.
As a bonus, our schemes are based on standard assumptions (i.e., do not rely on public-key quantum money). However, our schemes are \emph{finite-term} secure while the scheme by Ananth and La Placa~\cite{EC:AnaLaP21} is \emph{infinite-term} secure.\ryo{I added this sentence.} See~\cref{sec:overview} for an overview of our technique, (two-tier) quantum lightning, and software watermarking.

We can achieve SSL schemes with classical communication, where entities send only classical information to other entities (though they generate quantum states for their local computation). Our schemes are the first SSL schemes with classical communication.

We present the first SSL schemes for function classes other than evasive functions. Our schemes open the possibilities of software copy-protection for broader functionalities in the standard model.

\subsection{Related Work}\label{sec:related_work}
 Amos, Georgiou, Kiayias, and Zhandry presented many hybrid quantum cryptographic protocols, where we exchange only classical information and local quantum operation can yield advantages~\cite{STOC:AGKZ20}. Their constructions are secure relative to classical oracles. Radian and Sattath presented the notion of semi-quantum money, where both minting and verification protocols are interactive with classical communication~\cite{AFT:RadSat19}.
 Georgiou and Zhandry presented the notion of unclonable decryption keys~\cite{myEPRINT:GeoZha20}, which can be seen as quantum copy-protection for specific cryptographic tasks.

 \subsection{Concurrent Work}\label{sec:concurrent_work}
Aaronson et al.~\cite{CoRR:AarLiuZha20} significantly revised their paper in October 2020 and added new results in the revised version with additional authors~\cite{C:ALLZZ21}. They use a similar idea to ours to achieve their additional results. They achieved software copy-detection, which is a version of finite-term secure SSL, from public key quantum money and watermarking. They defined their copy detection so that it can provide natural security guarantee even if we consider leasing decryption or signing functionalities of cryptographic primitives.
As previously discussed in the context of watermarking~\cite{C:GKMWW19}, when considering those functionalities, we need to take a wider class of adversaries into consideration than considering just functions including PRF.
In fact, the reason why we focus only on PRF functionalities among cryptographic functionalities is that there was no definition of SSL that can handle decryption or signing functionalities.
We believe that by combining the work by Aaronson et al.~\cite{C:ALLZZ21} and our work, we can realize finite-term secure SSL for decryption and signing functionalities \ryo{for decryption and signing?}\fuyuki{Done.} based on the LWE assumption under a reasonable definition.

Coladangelo, Majenz, and Poremba~\cite{CoRR:ColMajPor20} realized finite-term secure SSL for the same sub-class of evasive functions as Ananth and La Placa~\cite{EC:AnaLaP21} using the quantum random oracle.
Based on their work, Broadbent, Jeffery, Lord, Podder, and Sundaram~\cite{CoRR:BJLPS21} showed that finite-term secure SSL for the class can be realized without any assumption.
We note that the definition of SSL used in these two works is different from the definition by Ananth and La Placa that we basically follow in this work.
Their definition has a nice property that their security notion captures any form of pirated copies rather than just authorized copies.
On the other hand, in their definition, not only the security notion, but also the correctness notion is parameterized by distributions on inputs to functions.
The security and correctness of the SSL schemes proposed in those works hold with respect to a specific distribution.

The advantage of our results over the above concurrent results is that we achieve SSL for functions beyond evasive functions, that is, PRF under standard lattice assumptions.
Moreover, our work is the first one that considers classical communication in the context of SSL.


\subsection{Technical Overview}\label{sec:overview}

\paragraph{Definition of SSL}
We review the definition of SSL given in \cite{EC:AnaLaP21}.
In this paper, we use a calligraphic font to represent quantum algorithms and calligraphic font or bracket notation to represent
quantum states following the notation of~\cite{STOC:AGKZ20}.
 
Formally, an SSL for a function class $\cC$ consists of the following algorithms. 
\begin{description}
\item[$\setup(1^\secp)\rightarrow \crs$:] This is a setup algorithm that generates a common reference string.  
\item[$\sslgen(\crs)\rightarrow \ssl.\sk$:] This is an algorithm supposed to be run by the lessor that generates lessor's secret key $\ssl.\sk$. The key is used to generate a leased software and verify the validity of a software returned by the lessee.
\item[$\lessor(\ssl.\sk,C)\rightarrow \sft_C$:] This is an algorithm supposed to be run by the lessor that generates a leased software $\sft_C$ that computes a circuit $C$.
\item[$\run(\crs,\sft_C,x)\rightarrow C(x)$:] This is an algorithm supposed to be run by the lessee to evaluate the software. As correctness, we require that the output should be equal to $C(x)$ with overwhelming probability if $\sft_C$ is honestly generated.\footnote{In the actual syntax, it also outputs a software, which is negligibly close to a software given as input.}
\item[$\sslcheck(\ssl.\sk,\sft_C)\rightarrow \top/\bot$:]
This is an algorithm supposed to be run by the lessor to check the validity of the software $\sft_C$ returned by the lessee.
As correctness, we require that this algorithm returns $\top$ (i.e., it accepts) with overwhelming probability if $\sft_C$ is an honestly generated one.
\end{description}

In this work, we focus on finite-term secure SSL.
Roughly speaking, the finite-term security of SSL requires that no quantum polynomial time (QPT) adversary given $\sft_C$ (for randomly chosen $C$ according to a certain distribution) can generate (possibly entangled) quantum states $\sft_0$ and $\sft_1$ such that 
$\sslcheck(\ssl.\sk,\sft_0)\rightarrow \top$ and $\run(\crs,\sft_1,\cdot)$ computes $C$ with non-negligible probability.
Thus, intuitively, the finite-term security ensures that finite-term security guarantees that adversaries cannot compute $C(x)$ via $\run$ after they return the valid leased state to the lessor.

\paragraph{Construction of SSL in \cite{EC:AnaLaP21}}
We review the construction of SSL in \cite{EC:AnaLaP21}.
Their construction is based on the following three building blocks:
\begin{itemize}
\item[\textit{Publicly verifiable unclonable state generator.}]
This enables us to generate a pair $(\pk,\sk)$ of public and secret keys in such a way that the following conditions are satisfied:
\begin{enumerate}
\item Given $\sk$, we can efficiently generate a quantum state $\ket{\psi_\pk}$.
\item Given $\pk$, we can efficiently implement a projective measurement $\{\ket{\psi_\pk}\bra{\psi_\pk},\allowbreak I-\ket{\psi_\pk}\bra{\psi_\pk}\}$.
\item Given $\pk$ and $\ket{\psi_\pk}$, no QPT algorithm can generate $\ket{\psi_\pk}^{\otimes 2}$ with non-negligible probability. 
\end{enumerate}
Aaronson and Christiano \cite{STOC:AarChr12} constructed a publicly verifiable unclonable state generator (under the name ``quantum money mini-scheme'') relative to a classical oracle, and Zhandry \cite{JC:Zhandry21} gave an instantiation in the standard model assuming post-quantum IO.
\item[\textit{Input-hiding obfuscator.}]
  This converts a circuit $C\in \cC$ (that is taken from a certain distribution) 
  to a functionally equivalent obfuscated circuit $\tlC$ in such a way that no QPT algorithm given $\tlC$ can find accepting point i.e., $x$ such that $C(x)=1$. 
  
 Ananth and La Placa \cite{EC:AnaLaP21} constructed an input-hiding obfuscator for a function class called compute-and-compare circuits under the LWE assumption.\footnote{A compute-and-compare circuit is specified by a circuit $C$ and a target value $\alpha$ and outputs $1$ on input $x$ if and only if $C(x)=\alpha$.}
 
  \item[\textit{Simulation-extractable non-interactive zero-knowledge.}]

A non-interactive zero-knowledge (NIZK) enables a prover to non-interactively prove an NP statement without revealing anything beyond the truth of the statement assuming a common reference string (CRS) generated by a trusted third party.
A simulation-extractable NIZK (seNIZK) additionally enables us to extract a witness from an adversary that is given arbitrarily many proofs generated by a zero-knowledge simulator and generates a new valid proof.
This property especially ensures that an seNIZK is an \emph{argument of knowledge} where a prover can prove not only truth of a statement but also that it knows a witness for the statement.

  Ananth and La Placa \cite{EC:AnaLaP21} showed that an seNIZK can be constructed from any (non-simulation-extractable) NIZK and CCA secure PKE, which can be instantiated under the LWE assumption~\cite{C:PeiShi19,SIAMCOMP:PeiWat11}.
\end{itemize}

Then their construction of SSL for $\cC$ is described as follows:
\begin{description}
\item[$\setup(1^\secp)$:] This just generates and outputs  a CRS $\crs$ of seNIZK.
\item[$\sslgen(\crs)$:] This generates a pair $(\pk,\sk)$ of  public and secret keys of the publicly verifiable unclonable state generator and outputs $\ssl.\sk:=(\pk,\sk)$.
\item[$\lessor(\ssl.\sk=(\pk,\sk),C)$:]  This obfuscates $C$ to generate an obfuscated circuit $\tlC$ by the input-hiding obfuscator and generates an seNIZK proof $\pi$ for a statement $(\pk,\tlC)$ that  it knows an accepting input $x$ of $\tlC$.\footnote{In the original construction in \cite{EC:AnaLaP21}, seNIZK also proves that $\pk$ and $\tlC$ was honestly generated. However, we found that this is redundant, and essentially the same security proof works even if it only proves the knowledge of an accepting input of $\tlC$. We note that it is important to include $\pk$ in the statement to bind a proof to $\pk$ even though the knowledge proven by the seNIZK has nothing to do with $\pk$.
In fact, this observation is essential to give our simplified construction of SSL.}
Then it outputs a leased software $\sft_C:=(\ket{\psi_\pk},\pk,\tlC,\pi)$.
We call $\ket{\psi_\pk}$ and $(\pk,\tlC,\pi)$ as quantum and classical parts of $\sft_C$, respectively.
\item[$\run(\crs,\sft_C,x)$:] This immediately returns $\bot$ if $\pi$ does not pass the verification of seNIZK.
It performs a projective measurement   $\{\ket{\psi_\pk}\bra{\psi_\pk},I-\ket{\psi_\pk}\bra{\psi_\pk}\}$ on the quantum part of $\sft_C$ by using $\pk$ and if the latter projection was applied, then it returns $\bot$.
Otherwise, it outputs $\tlC(x)$.
\item[$\sslcheck(\ssl.\sk,\sft_C)$:]
It performs a projective measurement   $\{\ket{\psi_\pk}\bra{\psi_\pk},I-\ket{\psi_\pk}\bra{\psi_\pk}\}$ on the quantum part  of $\sft_C$ and returns $\top$ if the former projection was applied and $\bot$ otherwise.
\end{description}

Intuitively, the finite-term security of the above SSL can be proven as follows.\footnote{
Note that Ananth and La Placa proved that the construction in fact satisfies infinite-term security that is stronger than finite-term security.
For ease of exposition of our ideas, we explain why the construction satisfies finite-term security.
}
Suppose that there exists an adversary that is given $\sft_C=(\ket{\psi_\pk},\pk,\tlC,\pi)$ and generates $\sft_0=(\qstate{psi}_0,\pk_0,\tlC_0,\pi_0)$ and $\sft_1=(\qstate{psi}_0,\pk_1,\tlC_1,\pi_1)$ such that $\sslcheck(\ssl.\sk,\sft_0)\rightarrow \top$ and $\run(\crs,\sft_1,\cdot)$ computes $C$ with non-negligible probability.
Then we  consider the following two cases:
\begin{enumerate}
\item[Case 1. $\pk_1=\pk$:] In this case, if $\run(\crs,\sft_1,\cdot)$ correctly computes $C$ (and especially outputs a non-$\bot$ value), then the quantum part of $\sft_1$ after the execution should be $\ket{\psi_\pk}$ by the construction of $\run$. 
On the other hand, if we have $\sslcheck(\ssl.\sk,\sft_0)\rightarrow \top$, then the quantum part of $\sft_0$ after the verification should also be $\ket{\psi_\pk}$ by the definition of the verification. 
Therefore, they can happen simultaneously only with a negligible probability due to the unclonability of  $\ket{\psi_\pk}$.
\item[Case 2. $\pk_1\neq\pk$:]   In this case, if $\run(\crs,\sft_1,\cdot)$ correctly computes $C$, then  $\pi_1$ is a valid proof for a statement $(\pk_1,\tlC_1)$ and $\tlC_1$ is functionally equivalent to $C$.
Since we have $(\pk_1,\tlC_1)\neq (\pk,\tlC)$,  by the simulation extractability of seNIZK, even if we replace $\pi$ with a simulated proof, we can extract a witness for $(\pk_1,\tlC_1)$, which contains an accepting input for $C$.
Since simulation of $\pi$ can be done only from the statement $(\pk,\tlC)$, this contradicts security of the input-hiding obfuscator, and thus this happens with a negligible probability.
\end{enumerate}
In summary, an adversary cannot win with a non-negligible probability in either case, which means that the SSL is finite-term secure.

\paragraph{Our idea for weakening assumptions.}
Unfortunately, their construction is based on a very strong assumption of post-quantum IO, which is needed to construct a publicly verifiable unclonable state generator.
Indeed, a publicly verifiable unclonable state generator implies public key quantum money by combining it with digital signatures \cite{STOC:AarChr12}.
Therefore, constructing a publicly verifiable unclonable state generator is as difficult as constructing a public key quantum money scheme, which is not known to exist under standard assumptions.

Our main observation is that we actually do not need the full power of public key quantum money for the above construction of SSL if we require only finite-term security since $\sslcheck$ can take a secret key, and thus it can run a private verification algorithm.
Then, does secret key quantum money suffice? Unfortunately, the answer is no. The reason is that 
even though $\sslcheck$ can take a secret key, 
$\run$ cannot since the secret key should be hidden from the lessee.
Based on this observation, we can see that what we actually need is something between public key quantum money and secret key quantum money.
We formalize this as \textit{two-tier quantum lightning}, which is a significant relaxation of quantum lightning introduced by Zhandry  \cite{JC:Zhandry21}.

\paragraph{Two-tier quantum lightning.}
Roughly speaking, quantum lightning (QL) is a special type of public key quantum money where anyone can generate a money state.
In QL, a public key $\pk$ is published by a setup algorithm and given $\pk$, anyone can efficiently generate a serial number $\snum$ along with a corresponding quantum state called bolt, which we denote by $\bolt$.   
We call this a \textit{bolt generation} algorithm.
As correctness, we require that given $\pk$, $\snum$, and any quantum state $\bolt$, anyone can verify if $\bolt$ is a valid state corresponding to the serial number $\snum$. Especially, if $\bolt$ is an honestly generated bolt, then the verification accepts with overwhelming probability. 
On the other hand, as security, we require that no QPT algorithm given $\pk$ can generate two (possibly entangled) quantum states $\bolt_0$ and $\bolt_1$ and a serial number $\snum$ such that both states pass the verification w.r.t. the serial number $\snum$ with non-negligible probability.

We introduce a weaker variant of QL which we call \textit{two-tier QL}.
In two-tier QL, a setup algorithm generates both a public key $\pk$ and a secret key $\sk$, 
and given $\pk$, anyone can efficiently generate a serial number $\snum$ along with a corresponding quantum state $\bolt$ similarly to the original quantum lightning.
The main difference from the original QL is that it has two types of verification: \textit{full-verification} and \textit{semi-verification}.
Full-verification uses a secret key $\sk$ while semi-verification only uses a public key $\pk$.
As correctness, we require that an honestly generated bolt passes both verifications with overwhelming probability.
On the other hand, as security, we require that no QPT algorithm given $\pk$ can generate two (possibly entangled) quantum states $\bolt_0$ and $\bolt_1$ and a serial number $\snum$ such that $\bolt_0$ passes the \textit{full-verification} w.r.t. the serial number $\snum$ and $\bolt_1$ passes the \textit{semi-verification} w.r.t. the serial number $\snum$ with non-negligible probability.
We note that this does not prevent an adversary from generating two states that pass semi-verification. Thus, we cannot use the semi-verification algorithm as a verification algorithm of the original QL.

We show that this two-tier verification mechanism is a perfect fit for finite-term secure SSL.
Specifically, based on the observation that $\sslcheck$ can take a secret key whereas $\run$ cannot as explained in the previous paragraph, we can use two-tier QL instead of publicly verifiable quantum state generators. This replacement is a slight adaptation of the construction in \cite{EC:AnaLaP21} by implementing verification by $\sslcheck$ and $\run$ with full- and semi-verification of two-tier QL, respectively.
We omit the detailed construction since that is mostly the same as that in  \cite{EC:AnaLaP21} except that we use two-tier QL.

\paragraph{Constructions of two-tier quantum lightning.}
Although no known construction of the original QL is based on a standard assumption, we give two two-tier QL schemes based on standard assumptions.

The first construction is based on the SIS assumption inspired by the recent work by Roberts and Zhandry \cite{RobZha21}.
The SIS assumption requires that no QPT algorithm given a matrix $\mm{A}\gets \Z_q^{n\times m}$ can find a short $\mv{s}\in \Z^m$ such that $\mm{A}\mv{s}=0 \mod q$.
Using this assumption, a natural approach to construct QL is as follows:\footnote{This approach was also discussed in the introduction of \cite{JC:Zhandry21}.}
Given a public key $\mm{A}$, a bolt generation algorithm generates a bolt of the form $\sum_{\mv{x}:\mm{A}\mv{x}=\mv{y} \text{~and~} \mv{x}\text{~is~``short''}}\alpha_{\mv{x}}\ket{\mv{x}}$ and a corresponding serial number $\mv{y}$.
This can be done by generating a superposition of short vectors in $\Z^m$, multiplying by $\mm{A}$ in superposition to write the result in an additional register, and measuring it.
The SIS assumption ensures that no QPT algorithm can generate two copies of a well-formed bolt for the same serial number with non-negligible probability. If it is possible, one can break the SIS assumption by measuring both bolts and returns the difference between them as a solution.
However, the fundamental problem is that we do not know how to publicly verify that a given state is a well-formed bolt for a given serial number.
Roughly speaking, Roberts and Zhandry showed that such verification is possible given a trapdoor behind the matrix $\mm{A}$, which yields a secretly verifiable version of QL (which is formalized as \textit{franchised quantum money} in \cite{RobZha21}).
We use this verification as the full-verification of our two-tier QL.
On the other hand, we define a semi-verification algorithm as an algorithm that just checks that a given state is a superposition of short preimages of $\snum=\mv{y}$ regardless of whether it is a well-formed superposition or not. This can be done by multiplying $\mm{A}$ in superposition, and especially can be done publicly.
Though a state that passes the semi-verification may collapse to a classical state, a state that passes the full-verification should not. 
Therefore, if we measure states that pass full- and semi- verification w.r.t. the same serial number, then the measurement outcomes are different with non-negligible probability. Thus the difference between them gives a solution to the SIS problem.
This implies that this construction of two-tier QL satisfies the security assuming the SIS assumption.

The second construction is based on the LWE assumption. The design strategy is based on a similar idea to the proof of quantumness by Brakerski et al. \cite{FOCS:BCMVV18}. We especially use a family of noisy trapdoor claw-free permutations constructed based on the LWE assumption in \cite{FOCS:BCMVV18}.
For simplicity, we describe the construction based on a family of clean  (non-noisy) trapdoor claw-free permutations in this overview.
A family of trapdoor claw-free permutations enables us to generate a function $f: \bit\times \bit^n \rightarrow \bit^n$ such that both $f(0,\cdot)$ and $f(1,\cdot)$ are permutations along with a trapdoor. 
As claw-free property, we require that no QPT algorithm given a description of $f$ can generate $x_0,x_1 \in \bit^n$  such that $f(0,x_0)=f(1,x_1)$ with non-negligible probability.
On the other hand, if one is given a trapdoor, then one can efficiently computes $x_0,x_1$ such that  $f(0,x_0)=f(1,x_1)=y$ for any $y\in \bit^n$.
Based on this, we construct two-tier QL as follows:
The setup algorithm generates $f$ and its trapdoor $\td$,  and sets a public key as the function $f$ and secret key as the trapdoor $\td$.
A bolt generation algorithm first prepares a uniform superposition $\sum_{b\in \bit,x\in \bit^n}\ket{b}\ket{x}$, applies $f$ in superposition to generate $\sum_{b\in \bit,x\in \bit^n}\ket{b}\ket{x}\ket{f(b,x)}$, measures the third register to obtain $y\in \bit^n$ along with a collapsed state $\frac{1}{\sqrt{2}}\left(\ket{0}\ket{x_0}+\ket{1}\ket{x_1}\right)$ where $f(0,x_0)=f(1,x_1)=y$.
Then it outputs a serial number $\snum:=y$ and a bolt $\bolt:=\frac{1}{\sqrt{2}}\left(\ket{0}\ket{x_0}+\ket{1}\ket{x_1}\right)$.
The full-verification algorithm given a trapdoor $\td$, a serial number $\snum=y$, and a (possibly malformed) bolt $\bolt$, computes $x_0,x_1$ such that $f(0,x_0)=f(1,x_1)=y$ using the trapdoor, and checks if $\bolt$ is $\frac{1}{\sqrt{2}}\left(\ket{0}\ket{x_0}+\ket{1}\ket{x_1}\right)$. More formally, it performs a projective measurement $\{\Pi, I-\Pi\}$ where $\Pi:=\frac{1}{2}\left(\ket{0}\ket{x_0}+\ket{1}\ket{x_1}\right)\left(\bra{0}\bra{x_0}+\bra{1}\bra{x_1}\right)$ 
and accepts if $\Pi$ is applied.
The semi-verification algorithm given $f$, $\snum=y$ and a (possibly malformed) bolt $\bolt$ just checks that $\bolt$ is a (not necessarily uniform) superposition of $(0,x_0)$ and $(1,x_1)$ by applying $f$ in superposition. 
Suppose that we are given states $\bolt_0$ and $\bolt_1$ that pass the full- and semi-verification respectively w.r.t. the same serial number $\snum=y$ with probability $1$.
Then after these verifications accept, if we measure $\bolt_0$, then we get $x_0$ or $x_1$ with equal probability and if we measure $\bolt_1$, we get either of $x_0$ and $x_1$. Therefore, with probability $1/2$, we obtain both $x_0$ and $x_1$, which contradicts the claw-free property.
This argument can be extended to show that the probability that $\bolt_0$ and $\bolt_1$ pass the full- and semi-verification respectively is at most $1/2 + \negl(\secp)$.  
By parallel repeating it many times, we can obtain two-tier QL. 

\paragraph{Abstracted construction of SSL via watermarking.}
Besides weakening the required assumption, we also give a slightly more abstracted SSL construction through the lens of watermarking. 
In general, a watermarking scheme enables us to embed a mark into a program so that the mark cannot be removed or modified without significantly changing the functionality.
 We observe that the classical part $(\pk,\tlC,\pi)$ of a leased software of \cite{EC:AnaLaP21} can be seen as a watermarked program of $C$ where $\pk$ is regarded as a mark.
 In this context, we only need to ensure that one cannot remove or modify the mark as long as one does not change the program's functionality \textit{when it is run on a legitimate evaluation algorithm} similarly to the security requirement for SSL. 
 We call a watermarking with such a weaker security guarantee a \textit{relaxed watermarking}.
With this abstraction along with the observation that two-tier QL suffices as already explained, we give a generic construction of SSL for $\cC$ based on two-tier QL and relaxed watermarking for $\cC$. This construction is in our eyes simpler than that in \cite{EC:AnaLaP21}.\footnote{Strictly speaking, our construction additionally uses message authentication code (MAC).}
 From this point of view, we can see that \cite{EC:AnaLaP21} essentially constructed a relaxed watermarking for compute-and-compare circuits based on seNIZK and input-hiding obfuscator for compute-and-compare circuits.
 We observe that an input-hiding obfuscator for compute-and-compare circuits can be instantiated from any injective one-way function, which yields a simpler construction of relaxed watermarking for compute-and-compare circuits without explicitly using input-hiding obfuscators.

\paragraph{SSL for PRF.}
Our abstracted construction ensures that a relaxed watermarking scheme for any circuit class can be converted to SSL for the same class assuming the existence of two-tier QL.
Here, we sketch our construction of a relaxed watermarking scheme for PRF.
Let  $F_{K}$ be a function that evaluates a PRF with a key $K$.
We assume that the PRF is a puncturable PRF. That is, one can generate a punctured key $K_{x^*}$ for any input $x^*$ that can be used to evaluate $F_{K}$ on all inputs except for $x^*$ but $F_{K}(x^*)$ remains pseudorandom even given $K_{x^*}$.
For generating a watermarked version of $F_{K}$ with a mark $\msg$, we generate $(K_{x^*},y^*:= F_{K}(x^*))$ for any fixed input $x^*$ and  an seNIZK proof $\pi$ for a statement $(\msg,K_{x^*},y^*)$ that  it knows $K$.
Then a watermarked program is set to be $(\msg,K_{x^*},y^*,\pi)$. A legitimate evaluation algorithm first checks if $\pi$ is a valid proof, and if so evaluates $F_K$ by using $K_{x^*}$ and $y^*$, and returns $\bot$ otherwise.
Roughly speaking, this construction satisfies the security of relaxed watermarking since if an adversary is given $(\msg,K_{x^*},y^*,\pi)$ can generate a program with a mark $\msg'\neq \msg$ that correctly computes $F_K$ on the legitimate evaluation algorithm. The program should contain a new valid proof of seNIZK that is different from $\pi$.
By the simulation extractability, we can extract  $K$ by using such an adversary.
Especially, this enables us to compute $K$ from  $(K_{x^*},y^*)$, which contradicts security of the puncturable PRF.
\ifnum\cameraready=1
\footnote{Strictly speaking, we need to assume the key-injectiveness for the PRF. See the full version of this paper~\cite{EPRINT:KitNisYam20} for the definition.}
\else
\footnote{Strictly speaking, we need to assume the key-injectiveness for the PRF as defined in \cref{def:key_injective}.}
\fi

By plugging the above relaxed watermarking for PRF into our generic construction, we obtain SSL for PRF. This would be impossible through the abstraction of \cite{EC:AnaLaP21}  since input-hiding obfuscator can exist only for evasive functions, whereas PRF is not evasive.

\paragraph{SSL with classical communication.}
As a final contribution, we give a construction of finite-term secure SSL where communication between the lessor and lessee is entirely classical. 
At a high level, the only quantum component of our SSL is two-tier QL, which can be seen as a type of quantum money. Thus we rely on techniques used for constructing semi-quantum money \cite{AFT:RadSat19}, which is a (secret key) quantum money with classical communication. More details are explained below.

In the usage scenario of finite-term secure SSL, there are two parts where the lessor and lessee communicate through a quantum channel. The first is when the lessor sends a software to the lessee. The second is when the lessee returns the software to the lessor.

For removing the first quantum communication, we observe that the only quantum part of a software is a bolt of two-tier QL in our construction, which can be generated publicly.
Then, our idea is to let the lessee generate the bolt by himself and only send the corresponding serial number to ask the lessor to generate a classical part of a software while keeping the bolt on lessee's side.
This removes the quantum communication at the cost of introducing an interaction. 
Though we let the lessor generate a bolt and a serial number by himself, 
 the security of SSL  is not affected because the security of two-tier QL ensures that an adversary cannot clone a bolt even if it is generated by himself.

For removing the second quantum communication, we assume an additional property for two-tier QL called \textit{bolt-to-certificate capability}, which was originally considered for (original) QL \cite{CoRR:ColSat20}.
Intuitively, this property enables us to convert a bolt to a classical certificate that certifies that the bolt was broken. Moreover, it certifies that one cannot generate any state that passes the semi-verification.
With this property, when returning the software, instead of sending the software itself, it can convert the bolt to a corresponding certificate and then send the classical certificate.
Security is still maintained with this modification since 
if the verification of the certification passes, then this ensures that the lessee no longer possesses a state that passes the semi-verification, and thus $\run$ always returns $\bot$.

Finally, we show that our LWE-based two-tier QL can be modified to have the bolt-to-certificate capability based on ideas taken from~\cite{FOCS:BCMVV18,AFT:RadSat19}.
Recall that in the LWE-based construction, a bolt is of the form $\frac{1}{\sqrt{2}}\left(\ket{0}\ket{x_0}+\ket{1}\ket{x_1}\right)$.
If we apply a Hadamard transform to the state and then measures both registers in the standard basis, then we obtain $(m,d)$ such that $m=d\cdot (x_0 \oplus x_1)$ as shown in \cite{FOCS:BCMVV18}. 
Moreover, Brakerski et al. \cite{FOCS:BCMVV18} showed that the LWE-based trapdoor claw-free permutation satisfies a nice property called \textit{adaptive hardcore property}, which roughly means that no QPT algorithm can output $(m,d,x',y)$ such that $d\neq 0$, $m=d\cdot (x_0 \oplus x_1)$ and $x'\in \{x_0,x_1\}$ with probability larger than $1/2+\negl(\secp)$
 where $x_0$ and $x_1$ are the unique values such that $f(0,x_0)=f(1,x_1)=y$.\footnote{More precisely, they prove an analogous property for a family of noisy trapdoor claw-free permutations.}
Since a quantum state that passes the semi-verification w.r.t. a serial number $y$ is a (not necessarily uniform) superposition of $x_0$ and $x_1$, we can see that $(m,d)$ works as a certificate with a weaker security guarantee that if one keeps a quantum state that passes the semi-verification, then one can generate $(m,d)$ that passes verification of $m=d\cdot (x_0 \oplus x_1)$ with probability at most $1/2+\negl(\secp)$.
 But this still does not suffice for our purpose since one can generate a certificate that passes the verification without discarding the original bolt with probability $1/2$ by just randomly guessing $(m,d)$. 
 To reduce this probability to negligible, we rely on an amplification theorem in~\cite{AFT:RadSat19} (which in turn is based on \cite{TCC:CanHalSte05}).
As a result, we can show that a parallel repetition to the above construction yields a two-tier QL with the bolt-to-certificate capability.

\ifnum\cameraready=1
\subsection{Organization}
In \cref{sec:prelim}, we provide definitions of cryptographic primitives used in this work.
In \cref{sec:abstraction_tool}, we introduce the notion of two-tier quantum lightning and provide concrete constructions of it.
In \cref{sec:relaxed_watermarking}, we define relaxed watermarking and provide concrete constructions of it.
In \cref{sec:SSL_ttQL}, we finally show how to construct SSL by combining two-tier quantum lightning and relaxed watermarking.
Due to the space limitation, some contents are omitted from this paper.
Especially, we omit the definition and construction of SSL with classical communication.
See the full version of this paper~\cite{EPRINT:KitNisYam20} for ommited contents.
\else\fi


\newcommand{\qIHO}{\algo{qIHO}}
\newcommand{\qSENIZK}{\algo{qSENIZK}}
\newcommand{\NTCF}{\algo{NTCF}}
\newcommand{\inv}{\algo{Inv}}
\newcommand{\support}{\algo{Supp}}
\newcommand{\samp}{\qalgo{Samp}}
\newcommand{\CHK}{\algo{Chk}}
\newcommand{\fk}{\keys{k}}
\newcommand{\Heldist}{\mathsf{H}^2}
\newcommand{\reg}{\mathsf{R}}
\newcommand{\qA}{\qalgo{A}}
\newcommand{\qB}{\qalgo{B}}

\section{Preliminaries}\label{sec:prelim}

\ifnum\cameraready=0
We review notations and definitions of cryptographic tools used in this paper.
\else
Due to the space limitation, some standard notations and definitions of cryptographic tools are omitted here and provided in the full version of this paper~\cite{EPRINT:KitNisYam20}.
\fi

\ifnum\submission=0

\ifnum\submission=0
\else
\section{Omitted Definitions}\label{sec:omitted-defs}
\fi

\subsection{Notations}\label{sec:notation}


In this paper, standard math or sans serif font stands for classical algorithms (e.g., $C$ or $\algo{Gen}$) and classical variables (e.g., $x$ or $\keys{pk}$).
Calligraphic font stands for quantum algorithms (e.g., $\qalgo{Gen}$) and calligraphic font and/or the bracket notation for (mixed) quantum states (e.g., $\qstate{sk}$ or $\ket{\psi}$).

In this paper, for a finite set $X$ and a distribution $D$,  $x \chosen X$ denotes selecting an element from $X$ uniformly at random, $x \chosen D$ denotes sampling an element $x$ according to $D$, and Let $y \gets \algo{A}(x)$ and $y \gets \qalgo{A}(\qstate{x})$ denote assigning to $y$ the output of a probabilistic or deterministic algorithm $\algo{A}$ and a quantum algorithm $\qalgo{A}$ on an input $x$ and $\qstate{x}$, respectively. When we explicitly show that $\algo{A}$ uses randomness $r$, we write $y \gets \algo{A}(x;r)$.
Let $[\ell]$ denote the set of integers $\{1, \cdots, \ell \}$, $\secp$ denote a security parameter, and $y \seteq z$ denote that $y$ is set, defined, or substituted by $z$.
PPT and QPT algorithms stand for probabilistic polynomial time algorithms and polynomial time quantum algorithms, respectively.
Let $\negl$ denote a negligible function.

Let $X$ be a random variable over a set $S$.
The \emph{min-entropy} of $X$, denoted by $\Hmin(X)$, is defined by
$
\Hmin(X) \seteq - \log_2 \max_{x \in S} \Pr[X = x] \enspace.
$
The conditional min-entropy of $X$ conditioned on a correlated variable $Y$, denoted by $\Hmin(X|Y)$, is defined as
$
\Hmin(X|Y) \seteq - \log_2\left(\textrm{E}_{y\la Y}\left[ \max_{x \in S} \Pr[X = x|Y=y]\right]\right) \enspace.
$

Let $\cH$ denote a finite-dimensional Hilbert space. For an operator $X$ on $\cH$, let $\norm{X}$ denote the operator norm of $X$, and $\trdist{X} \seteq \frac{1}{2} \norm{X}_1 = \frac{1}{2}\sqrt{X\adjmat{X}}$ for the trace norm.

\subsection{Distributions and Distance}\label{sec:distribution}

\begin{itemize}
\item $D$: a distribution over a finite domain $X$.
\item $f$: density on $X$. That is, a function $f: X \ra [0,1]$ such that $\sum_{x\in X} f(x)=1$.
\item $\cD_X$: the set of all densities on $X$.
\item For any $f\in \cD_X$, $\support(f) \seteq \setbracket{x \in X \mid f(x) >0}$.
\item For two densities $f_0$ and $f_1$ over the same finite domain $X$, the Hellinger distance between $f_0$ and $f_1$ is
\[
\Heldist(f_0,f_1) \seteq 1 - \sum_{x\in X}\sqrt{f_0(x)f_1(x)}.
\]
\item For density matrices $\qstate{X},\qstate{Y}$, the trace distance $\trdist{\qstate{X} - \qstate{Y}}$ is equal to
\[
\frac{1}{2} \Trace(\sqrt{(\qstate{X} -\qstate{Y})^2}).
\]
\end{itemize}
The following lemma relates Hellinger distance and the trace distance of superpositions.
\begin{lemma}\label{lem:Hellinger_trace_distance}
Let $X$ be a finite set, $f_0,f_1 \in \cD_X$, and
\[
\ket{\psi_b} \seteq \sum_{x\in X}\sqrt{f_b (x)}\ket{x}
\]
for $b\in\zo{}$. It holds that
\[
\trdist{\ket{\psi_0}\bra{\psi_0} - \ket{\psi_1}\bra{\psi_1}}= \sqrt{1-(1-\Heldist(f_0,f_1))^2}.
\]
\end{lemma}

\subsection{Lattices}\label{sec:lattice}

\begin{definition}[Learning with Errors]\label{def:LWE}
Let $n,m,q \in \N$ be integer functions of the security parameter $\secp$. Let $\chi = \chi(\secp)$ be a error distribution over $\Z$.
The LWE problem $\LWE_{n,m,q,\chi}$ is to distinguish the following two distributions.
\[
D_0 \seteq \setbracket{(\mm{A},\mm{A}\mv{s}+\mv{e}) \mid \mm{A} \chosen \Zq^{n\times m}, \mv{s}\chosen \Zq^{n}, \mv{e}\chosen \chi^m}  \text{ and } D_1 \seteq \setbracket{(\mm{A},\mv{u}) \mid \mm{A} \chosen \Zq^{n\times m}, \mv{u}\chosen \Zq^m}.
\]

When we say we assume the quantum hardness of the LWE problem, we assume that for any QPT adversary $\qalgo{A}$, it holds that
\[
\abs{ \Pr[\qalgo{A} (D_0)\out  1 ] - \Pr[\qalgo{A} (D_1)\out 1]} \le \negl(\secp).
\]
\end{definition}

\begin{definition}[Short Integer Solution]\label{def:SIS}
Let $n,m,q \in \N$ be integer functions of the security parameter $\secp$.
The SIS problem $\SIS_{n,m,q,\beta}$ is as follows.
Given $\mm{A}\chosen \Zq^{n\times m}$ and a positive real $\beta$, find a non-zero vector $\mv{s}\in \Z^m$ such that $\mm{A}\mv{s} = \mv{0}\bmod{q}$ and $\norm{s}\le \beta$.

When we say we assume the quantum hardness of the SIS problem, we assume that for any QPT adversary $\qalgo{A}$, it holds that
\[
 \Pr[\mm{A}\mv{s} = \mv{0}\bmod{q} \land \norm{\mv{s}}\le \beta \mid \mm{A}\chosen \Zq^{n\times m},\mv{s}\gets \qalgo{A} (\mm{A},\beta) ]  \le \negl(\secp).
\]
\end{definition}

\subsection{One-Way Functions}


We introduce the definition of a family of one-way functions (OWF) for high min-entropy sources.
\begin{definition}[OWF for High Min-Entropy Sources] \label{def-owf}
Let $\Fow=\{\owf:\Dowf \ra \Rowf\}$ be a family of efficiently computable deterministic functions.
Let $\gamma(\secp)$ be a function and $\cD$ a distribution $\cD=\{\cD_\secp\}_{\secp}$, where $(x,\aux)\la\cD_\secp$ outputs $x\in\Dowf$ and some auxiliary information $\aux$ such that $\Hmin(x|\aux)\ge\alpha(\lambda)$.
We say that $\Fow$ is a family of OWF for for $\alpha$-sources if for all QPT adversaries $\qA$, we have
\begin{align*}
\Pr\left[
f(x') = y~\middle|
\begin{array}{cc}
 &\owf\chosen\Fow \\
 &(x,\aux) \chosen \cD_\secp\\
 &x' \la \qA(1^\lambda, \owf,\aux,\owf(x))
\end{array}
\right] \leq \negl(\secp).
\end{align*}
\end{definition}

Alwen, Krenn, Pietrzak, and Wichs~\cite{C:AKPW13} prove that we can achieve deterministic encryption secure for any $\secp^\eta$ min-entropy source for any $\eta>0$ under the LWE assumption. Such deterministic encryption implies a family of injective OWF for $\secp^\eta$-sources.
This is the case when we consider QPT adversaries.
Formally, we have the following theorem.
\begin{theorem}\label{thm:inj_OWF_LWE}
Let $\eta>0$ be any constant.
Assuming the quantum hardness of the LWE problem, there exists a post-quantum injective OWF family for $\lambda^\eta$-sources.
\end{theorem}


\subsection{Pseudorandom Functions and Related Notions}

We introduce the definitions of pseudorandom functions (PRF) and puncturable PRF.

\begin{definition}[Pseudorandom Functions]\label{def:prf}
For sets $\prfkeyspace$, $\Dprf$, and $\Rprf$, let $\{\prf_\prfkey(\cdot): \Dprf \ra \Rprf \mid \prfkey \in \prfkeyspace \}$ be a family of polynomially computable functions.
We say that $\prf$ is pseudorandom if for any QPT adversary $\qA$, it holds that
\begin{eqnarray*}
\adva{\prf, \qA}{prf}(\lambda)
&=&|\Pr[\qA^{\prf_\prfkey(\cdot)}(1^\lambda)=1:\prfkey \chosen \prfkeyspace]\\
& &~~~~~~~~~~-\Pr[\qA^{\algo{R}(\cdot)}(1^\lambda)=1:\algo{R} \chosen \calU]
|=\negl(\lambda)
\enspace,
\end{eqnarray*}
where $\mathcal{U}$ is the set of all functions from $\Dprf$ to $\Rprf$.

\end{definition}


\begin{definition}[Puncturable PRF]\label{def:pprf}
For sets $\Dprf$ and $\Rprf$, a puncturable PRF $\PuncPRF$ whose key space is $\prfkeyspace$ consists of a tuple of algorithms $(\PRF.\Eval,\allowbreak\Puncture,\PRF.\pEval)$ that satisfies the following two conditions.
\begin{description}
\item[Functionality preserving under puncturing:] For all polynomial size subset $\{x_i\}_{i\in[k]}$ of $\Dprf$, all $x\in \Dprf \setminus \{x_i\}_{i\in[k]}$, and all $\prfkey\in\prfkeyspace$, we have $\PRF.\Eval(\prfkey,x) = \PRF.\pEval(\prfkey^*,x)$, where $\prfkey^* \la \Puncture(\prfkey,\{x_i\}_{i\in[k]})$.

\item[Pseudorandomness at punctured points:]
For all polynomial size subset $\{x_i\}_{i\in[k]}$ of $\Dprf$,
and any QPT adversary $\qA$, it holds that
\[
\abs{\Pr[\qA(\prfkey^*,\{\PRF.\Eval(\prfkey,x_i)\}_{i\in[k]}) = 1] -\Pr[\qA(\prfkey^*, \cU^{k}) = 1]} = \negl(\lambda)\enspace,
\]
where $\prfkey \chosen \prfkeyspace$, $\prfkeyspace^* \la \Puncture(\prfkey,\{x_i\}_{i\in[k]})$, and $\cU$ denotes the uniform distribution over $\Rprf$.
\end{description}
\end{definition}


We recall the notion of key-injectiveness for puncturable PRF~\cite{SIAMCOMP:CHNVW18}.

\begin{definition}[Key-Injectiveness]\label{def:key_injective}
We say that a puncturable PRF $\PuncPRF$ is key-injective if we have
\[
\Pr_{\prfkey\chosen\prfkeyspace}[\exists x\in\Dprf,\prfkey'\in\prfkeyspace\textrm{~s.t.~}\prfkey\neq\prfkey'\land\PRF.\Eval(\prfkey,x)=\PRF.\Eval(\prfkey',x)]\leq\negl(\secp).
\]
\end{definition}

We can realize puncturable PRF based on any one-way function.
Moreover, even if we require key-injectiveness, we can realize it under the LWE assumption, as shown by Cohen et al.~\cite{SIAMCOMP:CHNVW18}.
This is the case when we consider QPT adversaries.
Formally, we have the following theorem.
\begin{theorem}\label{thm:keyinj_pPRF_LWE}
There exists a key-injective puncturable PRF secure against QPT adversaries assuming the quantum hardness of the LWE problem.
\end{theorem}

\subsection{One-Time Message Authentication Code}

We introduce the definition of one-time message authentication code (OT-MAC).
\begin{definition}[OT-MAC]\label{def-mac}
An OT-MAC $\MAC$ is a three tuple $(\Macgen, \Mactag,\allowbreak\Macvrfy)$ of PPT algorithms. 
Below, let $\Dmac$ be the domain of $\MAC$.
\begin{itemize}

\item $\Macgen(1^\secp):$ Given a security parameter $1^\secp$, outputs a key $\mackey$.

\item $\Mactag(\mackey,m):$ Given a key $\mackey$ and a message $m \in \Dmac$, outputs $\mac$.

\item $\Macvrfy(\mackey,m,\mac):$ Given a key $\mackey$, message $m\in\Dmac$, and $\mac$, outputs $\top$ or $\bot$.

\end{itemize}

We require the following properties.
\begin{description}
\item[Correctness:]
For every $m\in\Dmac$ and $\mackey \la \Macgen(1^\lambda)$, we have $\Macvrfy(\mackey,m,\Mactag(\mackey,m))=\top$.

\item[Security:]
For any QPT adversary $\qA$, it holds that
\[
\Pr\left[ \begin{array}{rl}
&\Macvrfy(\mackey,m,\mac)=\top \land\\
&m\ne m_1
\end{array} \ \middle |
\begin{array}{rl}
 &\mackey \gets \Macgen(1^\secp) \\
 &(m,\mac) \gets \qA(1^\secp)^{\Mactag(\mackey,\cdot)}
\end{array}
\right]\le \negl(\secp)
\]
where $\qA$ can access to the oracle only once and $m_1$ is the query from $\qA$.
\end{description}


\end{definition}

We have the following theorem.
\begin{theorem}\label{thm:MAC_information_theoretic}
There exists an information-theoretically secure OT-MAC.
\end{theorem}

\subsection{Non-interactive Zero-Knowledge Systems}

We introduce the definition of a non-interactive zero-knowledge (NIZK) system and true-simulation extractability for it.

\begin{definition}[NIZK]\label{def:nizk}
Let $L$ be an NP language associated with the corresponding NP relation $R$.
A NIZK system for $L$ is a tuple of algorithms $(\NIZK.\setup,\allowbreak\NIZK.\prove,\allowbreak\NIZK.\vrfy)$.
\begin{itemize}
\item $\NIZK.\setup(1^\secp)$: The setup algorithm takes as input the security parameter $1^\secp$ and outputs a common reference string $\crs$.
\item $\NIZK.\prove(\crs,x,w)$: The prove algorithm takes as input common reference string $\crs$, NP instance $x$, and witness $w$, and outputs a proof $\pi$.
\item $\NIZK.\vrfy(\crs,x,\pi)$: The verification algorithm takes as input common reference string $\crs$, NP instance $x$, and proof $\pi$, and outputs $\top$ or $\bot$.
\end{itemize}
\end{definition}

\begin{definition}[Completeness]
A NIZK system for NP is said to be complete if we have $\NIZK.\vrfy(\crs, x,\allowbreak \NIZK.\prove(\crs, x,w)) = \top$ for all common reference string $\crs$ output by $\NIZK.\setup(1^\lambda)$ and all valid statement/witness pairs $(x,w) \in R$.
\end{definition}

\begin{definition}[True-Simulation Extractability]\label{def:tse_nizk}
Let $\NIZK$ be a NIZK system and $\qA$ a QPT adversary.
Let $\Sim=(\fksetup,\Sim_1,\Sim_2)$ be a tuple of PPT algorithms.
We define the following experiment $\expt{\qA,\NIZK}{se\textrm{-}real}$.
\begin{enumerate}
\item The challenger first generates $\crs \gets \NIZK.\setup(1^\secp)$ and sends $\crs$ to $\qA$.
\item $\qA$ sends $q$ statement/witness pairs $(x_i,w_i)_{i\in[q]}$ to the challenger.
The challenger responds with $\{\pi_i\}_{i\in[q]}$, where $\pi_i\gets\NIZK.\prove(\crs,x_i,w_i)$ for every $i\in[q]$.
\item Finally, $\qA$ outputs $(x',\pi')$.
The challenger outputs $1$ if $\NIZK.\vrfy(\crs,x',\pi')=\top$, $(x_i,w_i)\in R$ for every $i\in[q]$, and $x_i\neq x'$ for every $i\in[q]$ hold. 
Otherwise, the challenger outputs $0$.
\end{enumerate}
We also define the following experiment $\expt{\qA,\Sim,\NIZK}{se\textrm{-}sim}$.
\begin{enumerate}
\item The challenger first generates $(\crs,\td) \gets \fksetup(1^\secp)$ and sends $\crs$ to $\qA$.
\item $\qA$ sends $q$ statement/witness pairs $(x_i,w_i)_{i\in[q]}$ to the challenger.
The challenger computes $(\{\pi_i\}_{i\in[q]},\state_{\Sim})\gets\Sim_1(\crs,\td,\{x_i\}_{i\in[q]})$ and returns $\{\pi_i\}_{i\in[q]}$ to $\qA$.
\item Finally, $\qA$ outputs $(x',\pi')$.
The challenger computes $w'\gets\Sim_2(\state_{\Sim},x',\pi')$.
The challenger outputs $1$ if $\NIZK.\vrfy(\crs,x',\pi')=\top$, $(x_i,w_i)\in R$ for every $i\in[q]$, $(x',w')\in R$, and $x_i\neq x'$ for every $i\in[q]$ hold.
Otherwise, the challenger outputs $0$.
\end{enumerate}

%

A NIZK system is said to be true-simulation extractable if for any QPT adversary $\qA$, there exists a tuple of PPT algorithms $\Sim$ such that we have
\[\abs{\Pr[1\gets\expt{\qA,\NIZK}{se\textrm{-}real}]-\Pr[1\gets\expt{\qA,\Sim,\NIZK}{se\textrm{-}sim}]}\leq\negl(\secp).\]
\end{definition}

Ananth and La Placa~\cite{EC:AnaLaP21} showed the following theorem.
\begin{theorem}\label{thm:seNIZK_LWE}
There exists a true-simulation extractable NIZK system secure against polynomial (resp. sub-exponential) time quantum adversaries assuming the quantum hardness of the LWE problem against polynomial (resp. sub-exponential) time quantum adversaries.
\end{theorem}

\else\fi

\subsection{Noisy Trapdoor Claw-Free Hash Function}\label{sec:NTCF}

We recall the notion of noisy trapdoor claw-free (NTCF) hash function~\cite{FOCS:BCMVV18}.




\begin{definition}[NTCF Hash Function~\cite{FOCS:BCMVV18}]\label{def:NTCF}
Let $\cX$, $\cY$ be finite sets, $\cD_\cY$ the set of probability densities over $\cY$, and $\cK_\cF$ a finite set of keys.
A family of functions
\[\cF \seteq \setbracket{f_{\fk,b} : \cX \ra \cD_\cY}_{\fk\in \cK_\cF ,b\in\zo{}}\]
is a NTCF family if the following holds.
\begin{description}
\item [Efficient Function Generation:] There exists a PPT algorithm $\NTCF.\gen_\cF$ which generates a key $\fk\in \cK_\cF$ and a trapdoor $\td$.
\item [Trapdoor Injective Pair:] For all keys $\fk\in \cK_\cF$, the following holds.
\begin{enumerate}
\item \emph{Trapdoor:} For all $b\in \zo{}$ and $x\ne x' \in \cX$, $\support(f_{\fk,b}(x)) \cap \support(f_{\fk,b}(x')) = \emptyset$. In addition, there exists an efficient deterministic algorithm $\inv_\cF$ such that for all $b\in\zo{}, x\in\cX$ and $y\in \support(f_{\fk,b}(x))$, $\inv_\cF(\td,b,y)=x$.
\item \emph{Injective pair:} There exists a perfect matching relation $\cR_\fk \subseteq \cX \times \cX$ such that $f_{\fk,0}(x_0)=f_{\fk,1}(x_1)$ if and only if $(x_0,x_1) \in \cR_\fk$.
\end{enumerate}
\item [Efficient Range Superposition:] For all keys $\fk\in \cK_\cF$ and $b\in \zo{}$, there exists a function $f'_{\fk,b}: \cX \ra \cD_\cY$ such that the following holds.
\begin{enumerate}
\item For all $(x_0,x_1)\in \cR_\fk$ and $y \in \support(f'_{\fk,b}(x_b))$, $\inv_\cF (\td,b,y)=x_b$ and $\inv_\cF(\td,b\xor 1,y)=x_{b\xor 1}$.
\item There exists an efficient deterministic procedure $\CHK_\cF$ that takes as input $k,b\in\zo{},x\in\cX$ and $y\in\cY$ and outputs $1$ if $y\in \support(f'_{\fk,b}(x))$ and $0$ otherwise. This procedure does not need the trapdoor $\td$.
\item For all $\fk \in \cK$ and $b\in\zo{}$,
\[
\Exp_{x \chosen \cX}[\Heldist(f_{\fk,b}(x),f'_{\fk,b}(x))] \le \negl(\secp).
\]
Here $\Heldist$ is the Hellinger distance
\ifnum\cameraready=1
(See~\cite{EPRINT:KitNisYam20}).
\else
(See~\cref{sec:distribution}).
\fi
In addition, there exists a QPT algorithm $\samp_\cF$ that takes as input $\fk$ and $b\in\zo{}$ and prepare the quantum state
\begin{align}
\ket{\psi'} = \frac{1}{\sqrt{\abs{\cX}}}\sum_{x\in \cX, y\in\cY}\sqrt{(f'_{\fk,b}(x))(y)}\ket{x}\ket{y}. \label{eq:NTCF_superposition}
\end{align}
This property and
\ifnum\cameraready=1
a lemma about trace and Hellinger distances (See~\cite{EPRINT:KitNisYam20})
\else
~\cref{lem:Hellinger_trace_distance}
\fi
immediately imply that
\[
\trdist{ \ket{\psi}\bra{\psi} - \ket{\psi'}\bra{\psi'}} \leq \negl(\secp),
\]
where $\ket{\psi} = \frac{1}{\sqrt{\abs{\cX}}}\sum_{x\in \cX, y\in\cY}\sqrt{(f_{\fk,b}(x))(y)}\ket{x}\ket{y}$.
\end{enumerate}
\item [Adaptive Hardcore Bit:] For all keys $\fk\in\cK_\cF$, the following holds. For some integer $w$ that is a polynomially bounded function of $\secp$,
\begin{enumerate}
\item For all $b\in\zo{}$ and $x\in\cX$, there exists a set $G_{\fk,b,x}\subseteq \zo{w}$ such that $\Pr_{d\chosen \zo{w}}[d \notin G_{\fk,b,x}] \le \negl(\secp)$. In addition, there exists a PPT algorithm that checks for membership in $G_{\fk,b,x}$ given $\fk,b,x$, and $\td$.
\item  There is an efficiently computable injection $J: \cX \ra \zo{w}$ such that $J$ can be inverted efficiently on its range, and such that the following holds. Let
\begin{align*}
 H_\fk  & \seteq \setbracket{(b,x_b,d, d\cdot(J(x_0)\xor J(x_1))) \mid b\in\zo{},(x_0,x_1)\in \cR_\fk, d \in G_{\fk,0,x_0} \cap G_{\fk,1,x_1}},\\
 \overline{H}_\fk & \seteq \setbracket{(b,x_b,d, c)\mid (b,x,d,c\xor 1)\in H_k},
\end{align*}
then for any QPT $\A$, it holds that
\begin{align}
\abs{\Pr_{(\fk,\td)\gets \NTCF.\gen_\cF (1^\secp)}[\A(\fk)\in H_\fk] - \Pr_{(\fk,\td)\gets \NTCF.\gen_\cF (1^\secp)}[\A(\fk)\in \overline{H}_\fk]} \le \negl(\secp).
\end{align}
\end{enumerate}
\end{description}
\end{definition}

Brakerski et al. showed the following theorem.

\begin{theorem}[\cite{FOCS:BCMVV18}]\label{thm:NTCF_LWE}
If we assume the quantum hardness of the LWE problem, then there exists an NTCF family.
\end{theorem}

\subsection{Secure Software Leasing}

We introduce the notion of secure software leasing (SSL) defined by Ananth and La Placa~\cite{EC:AnaLaP21}.

\newcommand{\runout}{\run_{\mathsf{out}}}

\begin{definition}[SSL with Setup~\cite{EC:AnaLaP21}]\label{def:ssl_crs}
Let $\cC=\setbracket{\cC_{\secp}}_{\secp}$ be a circuit class such that $\cC_{\secp}$ contains circuits of input length $n$ and output length $m$.
A secure software lease scheme with setup for $\cC$ is a tuple of algorithms $(\setup,\sslgen,\lessor,\allowbreak\run,\sslcheck)$.
\begin{itemize}
\item $\setup(1^\secp)$: The setup algorithm takes as input the security parameter $1^\secp$ and outputs a classical string $\crs$.
\item $\sslgen(\crs)$: The key generation algorithm takes as input $\crs$ and outputs a secret key $\sk$.
\item $\lessor(\sk,C)$: The lease algorithm takes as input $\sk$ and a polynomial-sized classical circuit $C\in \cC_\secp$ and outputs a quantum state $\sft_C$.
\item $\run(\crs,\sft_C,x)$: The run algorithm takes as input $\crs$, $\sft_C$, and an input $x \in \zo{n}$ for $C$, and outputs $y \in \zo{m}$ and some state $\sft'$.
We use the notation $\runout(\crs,\sft_C,x)=y$ to denote that $\run(\crs,\sft_C,x)$ results in an output of the form $(\sft',y)$ for some state $\sft'$.
\item $\sslcheck(\sk,\sft_{C}^{\ast})$: The check algorithm takes as input $\sk$ and $\sft_{C}^{\ast}$, and outputs $\top$ or $\bot$.
\end{itemize}
\end{definition}

\begin{definition}[Correctness for SSL]\label{def:correctness_SSL}
An SSL scheme $(\setup,\sslgen,\lessor,\run,\allowbreak\sslcheck)$ for $\cC = \setbracket{\cC_\secp}_\secp$ is correct if for all $C\in \cC_\secp$, the following two properties hold:
\begin{itemize}
\item Correctness of $\run$:
\ifnum\submission=0
\[
\Pr\left[ \forall x\ \Pr[\runout(\crs,\sft_C,x) = C(x)]\ge 1-\negl(\secp) \ \middle |
\begin{array}{rl}
 &\crs \gets \setup(1^\secp) \\
 &\sk \gets \sslgen(\crs) \\
 &\sft_C \gets \lessor(\sk,C)
\end{array}
\right]\ge 1-\negl(\secp).
\]
\else
\begin{align*}
&\Pr\left[ \forall x\ \Pr[\runout(\crs,\sft_C,x) = C(x)]\ge 1-\negl(\secp) \ \middle |
\begin{array}{rl}
 &\crs \gets \setup(1^\secp) \\
 &\sk \gets \sslgen(\crs) \\
 &\sft_C \gets \lessor(\sk,C)
\end{array}
\right]\\
&\ge1-\negl(\secp).
\end{align*}
\fi
\item Correctness of $\sslcheck$:
\[
\Pr\left[ \sslcheck (\sk,\sft_C) = \top \ \middle |
\begin{array}{rl}
 &\crs \gets \setup(1^\secp) \\
 &\sk \gets \sslgen(\crs) \\
 &\sft_C \gets \lessor(\sk,C)
\end{array}
\right]\ge 1-\negl(\secp).
\]
\end{itemize}
\end{definition}

\begin{definition}[Reusability for SSL]\label{def:reusability_SSL}
An SSL scheme $(\setup,\sslgen,\lessor,\run,\allowbreak\sslcheck)$ for $\cC = \setbracket{\cC_\secp}_\secp$ is reusable if for all $C\in \cC_\secp$ and for all $x\in\zo{n}$, it holds that
\[\trdist{\sft'_{C,x} - \sft_C} \le \negl(\secp),\]
\end{definition}
where $\sft'_{C,x}$ is the quantum state output by $\run(\crs,\sft_C,x)$.

\begin{lemma}[\cite{EC:AnaLaP21}]\label{lem:reusability_SSL}
If an SSL scheme $(\setup,\sslgen,\lessor,\run,\sslcheck)$ for $\cC = \setbracket{\cC_\secp}_\secp$ is correct, then there exists a QPT algorithm $\run'$ such that $(\setup,\sslgen,\allowbreak\lessor,\run',\sslcheck)$ is a reusable SSL scheme for $\cC = \setbracket{\cC_\secp}_\secp$.
\end{lemma}

Below, we introduce a security notion called finite-term lessor security for SSL.
We can also consider a stronger security notion called infinite-term lessor security for SSL.
For the definition of infinite-term lessor security, see the paper by Ananth and La Placa~\cite{EC:AnaLaP21}.

In the security experiment of SSL, an adversary outputs a bipartite state $\sft^\ast$ on the first and second registers. Let $\sft_0^\ast \seteq \Trace_2[\sft^\ast]$ and $\sft_0^\ast$ is verified by $\sslcheck$.\footnote{$\Trace_i[\qstate{X}]$ is the partial trace of $\qstate{X}$ where the $i$-th register is traced out.} In addition, $P_2(\sk,\sft^\ast)$ denotes the resulting post-measurement state on the second register (after the check on the first register).
We write
\[
P_2(\sk,\sft^\ast) \propto \Trace_1[\Pi_1[(\sslcheck(\sk,\sft^\ast)_1 \tensor I_2)(\sft^\ast)]]
\]
for the state that $\qA$ keeps after the first register has been returned and verified. Here, $\Pi_1$ denotes projecting the output of $\sslcheck$ onto $\top$, and where $(\sslcheck(\sk,\sft^\ast)_1 \tensor I_2)(\sft^\ast)$ denotes applying $\sslcheck$ on to the first register, and the identity on the second register of $\sft^\ast$.
\begin{definition}[Perfect Finite-Term Lessor Security]\label{def:perfect_lessor_security}
Let $\beta$ be any inverse polynomial of $\lambda$ and $\cD_{\cC}$ a distribution on $\cC$.
We define the $(\beta,\cD_{\cC})$-perfect finite-term lessor security game $\expt{\qA,\cD_\cC}{pft\textrm{-}lessor}(\secp,\beta)$ between the challenger and adversary $\qA$ as follows.
\begin{enumerate}
\item The challenger generates $C \gets \cD_{\cC}$, $\crs \gets \setup(1^\secp)$, $\sk\gets \sslgen(\crs)$, and $\sft_C \gets \lessor(\sk,C)$, and sends $(\crs,\sft_C)$ to $\qA$.
\item $\qA$ outputs a bipartite state $\sft^\ast$. Below, we let $\sft_0^\ast \seteq \Trace_2 [\sft^\ast]$.
\item If $\sslcheck(\sk,\sft_0^\ast) = \top$ and $\forall x\ \Pr[\runout(\crs,P_2(\sk,\sft^\ast),x) = C(x)]\ge \beta$ hold, where the probability is taken over the choice of the randomness of $\run$, then the challenger outputs $1$. Otherwise, the challenger outputs $0$.
\end{enumerate}

We say that an SSL scheme $(\setup,\sslgen,\lessor,\run,\sslcheck)$ is $(\beta,\cD_{\cC})$-perfect finite-term lessor secure,
if for any QPT $\qA$ that outputs a bipartite (possibly entangled) quantum state on the first and second registers, the following holds.
\begin{align*}
\Pr[\expt{\qA,\cD_\cC}{pft\textrm{-}lessor}(\secp,\beta) \out 1] \le \negl(\secp).
\end{align*}
\end{definition}

In addition to the above perfect finite-term lessor security, we also introduce a new security notion \emph{average-case finite-term lessor security}.
For an SSL scheme for a family of PRF, we consider average-case finite-term lessor security.
This is because when we consider cryptographic functionalities, the winning condition ``$\forall x\ \Pr[\runout(\crs,P_2(\sk,\sigma^\ast),x) = C(x)]\ge \beta$'' posed to the adversary in the definition of perfect finite-term lessor security seems to be too strong.
In fact, for those functionalities, adversaries who can generate a bipartite state $\sft^*$ such that $\runout(\crs,P_2(\sk,\sft^\ast),x) = C(x)$ holds for some fraction of inputs $x$ should be regarded as successful adversaries.
Average-case finite-term lessor security considers those adversaries.

\begin{definition}[Average-Case Finite-Term Lessor Security]\label{def:average_lessor_security}
Let $\epsilon$ be any inverse polynomial of $\lambda$ and $\cD_{\cC}$ a distribution on $\cC$.
We define the $(\epsilon,\cD_{\cC})$-average-case finite-term lessor security game $\expt{\qA,\cD_\cC}{aft\textrm{-}lessor}(\secp,\epsilon)$ between the challenger and adversary by replacing the third stage of $\expt{\qA,\cD_\cC}{pft\textrm{-}lessor}(\secp,\beta)$ with the following.
\begin{enumerate}
\setcounter{enumi}{2}
\item If $\sslcheck(\sk,\sft_0^\ast) = \top$ and $\Pr[\runout(\crs,P_2(\sk,\sft^\ast),x) = C(x)]\ge \epsilon$ hold, where the probability is taken over the choice of $x\la\zo{n}$ and the random coin of $\run$, then the challenger outputs $1$. Otherwise, the challenger outputs $0$.
\end{enumerate}

We say that an SSL scheme $(\setup,\sslgen,\lessor,\run,\sslcheck)$ is $(\epsilon,\cD_{\cC})$-average-case finite-term lessor secure,
if for any QPT $\qA$ that outputs a bipartite (possibly entangled) quantum state on the first and second registers, the following holds.
\begin{align*}
\Pr[\expt{\qA,\cD_\cC}{aft\textrm{-}lessor}(\secp,\epsilon) \out 1] \le \negl(\secp).
\end{align*}
\end{definition}

\section{Two-Tier Quantum Lightning}\label{sec:abstraction_tool}

In this section, we present definitions of our new tools and their instantiations.

\subsection{Two-Tier Quantum Lightning}\label{sec:tt_QL}
We define two-tier QL, which is a weaker variant of QL~\cite{JC:Zhandry21}.
A big difference from QL is that we have two types of verification called semi-verification and full-verification.
We need a secret key for full-verification while we use a public key for semi-verification.

\begin{definition}[Two-Tier Quantum Lightning (syntax)]\label{def:tt_QL}
A two-tier quantum lightning scheme is a tuple of algorithms $(\setup,\boltgen,\semivrfy,\fullvrfy)$.
\begin{itemize}
\item $\setup(1^\secp)$: The setup algorithm takes as input the security parameter $1^\secp$ and outputs a key pair $(\pk,\sk)$.
\item $\boltgen(\pk)$: The bolt generation algorithm takes as input $\pk$ and outputs a classical string $\snum$ (called a serial number) and a quantum state $\bolt$ (called a bolt for the serial number).
\item $\semivrfy(\pk,\snum,\bolt)$: The semi-verification algorithm takes as input $\pk$, $\snum$, and $\bolt$ and outputs $(\top,\bolt')$ or $\bot$.
\item $\fullvrfy(\sk,\snum,\bolt)$: The full-verification algorithm takes as input $\sk$, $\snum$, and $\bolt$ and outputs $\top$ or $\bot$.
\end{itemize}
\end{definition}

\begin{definition}[Correctness for Two-Tier Quantum Lightning]\label{def:ttQL_correctness}

There are two verification processes. We say that a two-tier quantum lightning with classical verification is correct if it satisfies the following two properties.
\begin{description}
	\item[Semi-verification correctness:]
	\begin{align*}
\Pr\left[
(\top,\bolt') \gets \semivrfy(\pk,\snum,\bolt)
\ \middle |
\begin{array}{rl}
 &(\pk,\sk) \gets \setup(1^\secp) \\
 &(\snum,\bolt) \gets \boltgen(\pk)
\end{array}
\right] > 1-\negl(\secp).
\end{align*}

	\item[Full-verification correctness:]
\begin{align*}
\Pr\left[
\top\gets \fullvrfy(\sk,\snum,\bolt)\ \middle |
\begin{array}{rl}
 &(\pk,\sk) \gets \setup(1^\secp) \\
 &(\snum,\bolt) \gets \boltgen(\pk)
\end{array}
\right] > 1-\negl(\secp).
\end{align*}
\end{description}
\end{definition}

\begin{definition}[Reusability for Two-Tier Quantum Lightning]\label{def:ttQL_reusability}
A two-tier quantum lightning scheme $(\setup,\boltgen,\semivrfy,\fullvrfy)$ is reusable if for all $(\pk,\sk)\gets \setup(1^\secp)$, $(\snum,\bolt)\gets \boltgen(\pk)$, and $(\bolt',\top)\gets \semivrfy(\pk,\allowbreak\snum,\bolt)$, it holds that
\[\trdist{\bolt' - \bolt} \le \negl(\secp).\]
\end{definition}
\begin{remark}\label{rem:reusability}
We can show that any two-tier QL scheme that satisfies semi-verification correctness can be transformed into one that satisfies reusability by using the Almost As Good As New Lemma \cite{Aar05} similarly to an analogous statement for SSL shown in \cite{EC:AnaLaP21}.   
Therefore, we focus on correctness.
\end{remark}

\begin{definition}[Two-Tier Unclonability]\label{def:tt_unclonability}
We define the two-tier unclonability game between a challenger and an adversary $\A$ as follows.
\begin{enumerate}
\item The challenger generate $(\pk,\sk)\gets \setup(1^\secp)$ and sends $\pk$ to $\qA$.
\item $\qA$ outputs possibly entangled quantum states $\qstate{L_0}$ and $\qstate{L_1}$ and a classical string $\snum^\ast$, and sends them to the challenger.
\item The challenger runs $\fullvrfy(\sk,\snum^\ast,\qstate{L_0})$ and $\semivrfy(\pk,\snum^\ast,\qstate{L_1})$. If both the outputs are $\top$, then this experiments outputs $1$. Otherwise, it outputs $0$.
\end{enumerate}
This game is denoted by $\expa{\A,\Sigma}{tt}{unclone}(1^\secp)$. A two-tier quantum lightning scheme is two-tier unclonable if for any QPT adversary $\qA$, it holds that
\[
\Pr[\expa{\qA,\Sigma}{tt}{unclone}(1^\sep)=1] \leq \negl(\secp).
\]
\end{definition}

\begin{definition}[Secure Two-Tier Quantum Lightning]\label{def:secure_ttQL}
A two-tier quantum lightning scheme is secure if it satisfies~\cref{def:tt_QL,def:ttQL_correctness,def:ttQL_reusability,def:tt_unclonability}.
\end{definition}

\ifnum\submission=0

\newcommand{\FQM}{\mathsf{FQM}}
\newcommand{\fqmver}{\qalgo{FQMVrfy}}
\newcommand{\fqmgen}{\qalgo{FQMGen}}

\ifnum\submission=0
\subsection{Two-Tier Quantum Lightning from SIS}\label{sec:ttQL_SIS}
\else
\section{Two-Tier Quantum Lightning from SIS}\label{sec:ttQL_SIS}
\fi
We show how to construct a two-tier quantum lightning scheme from the SIS assumption.
The construction is based on the franchised quantum money scheme by Roberts and Zhandry~\cite{RobZha21}.
They (implicitly) proved the following lemma holds by appropriately setting parameters $n,m,q,\beta$ in such a way that $\SIS_{n,m,q,\beta}$ is believed to be hard:
\begin{lemma}[{\cite{RobZha21}}]\label{lem:fqm}
There exist PPT algorithm $\trapgen$ and QPT algorithms $(\fqmgen,\fqmver)$ that work as follows:\footnote{$\trapgen$ is by now a standard  algorithm to sample a matrix with its trapdoor \cite{STOC:GenPeiVai08,EC:MicPei12}.}
\begin{description}
\item[$\trapgen(1^\secp)$:]
This algorithm generates a matrix $\mm{A}\in \Zq^{n\times m}$ and its trapdoor $\td$.
\item[$\fqmgen(\mm{A})$:]
Given a matrix $\mm{A}\in \Zq^{n\times m}$, it outputs a vector $\mv{y}\in \Zq^n$ along with a quantum state 
\begin{align*}
\ket{\Sigma} & = \sum_{\mv{x} \in \Zq^m : \mm{A}\mv{x} = \mv{y} \bmod{q}} \sqrt{p(\mv{x})} \ket{\mv{x}}.
\end{align*}
for a certain probability density function $p$ over $\{\mv{x} \in \Zq^m : \mm{A}\mv{x} = \mv{y} \bmod{q}\}$
such that if we take $\mv{x}$ according to $p$, we have $\Pr[\|\mv{x}\|> \beta/2]=\negl(\secp)$.
\footnote{Specifically, $p$ is proportional to discrete Gaussian.}
\item[$\fqmver(\td,\mv{y},\ket{\Sigma})$:]
It outputs $\top$ or $\bot$.
\end{description}
Moreover, the following is satisfied:
\begin{enumerate}
\item \label{item:correctness}
If we generate $(\mm{A},\td)\gets \trapgen(1^\secp)$ and $(\mv{y},\ket{\Sigma})\gets \fqmgen(\mm{A})$, we have 
\[
\Pr[\fqmver(\td,\mv{y},\ket{\Sigma})=\bot]=\negl(\secp).
\]  
\item \label{item:security}
For any $(\mm{A},\td)\gets \trapgen(1^\secp)$, $\mv{y}\in \Zq^n$ and (possibly malformed) quantum state $\qstate{sigma}$ such that $\Pr[\fqmver(\td,\mv{y},\qstate{sigma})=\top]$ is non-negligible,  
if we measure $\qstate{sigma}$, then the outcome $\mv{x}$ satisfies $\mm{A}\mv{x} = \mv{y} \bmod{q}$ and $\|\mv{x}\|\leq \beta/2$ with a non-negligible probability, and no value for $\mv{x} \in \Zq^m$ has overwhelming probability of being measured conditioned on that the above holds. 


\end{enumerate}
\end{lemma}
\begin{proof}(sketch.)
$\fqmver(\td,\mv{y},\qstate{sigma})$ first checks if the value $\mv{x}$ in the register of $\qstate{sigma}$ satisfies $\mm{A}\mv{x}=\mv{y}$ and $\|\mv{x}\|\leq \beta/2$ in superposition by writing the result into another register and measuring it.  If that is not satisfied, it immediately outputs $\bot$ and halts. 
Otherwise, it applies the quantum Fourier transform on $\qstate{sigma}$ and measures the state to get a vector $\mv{z}\in \Zq^m$. 
If it is an LWE instance, i.e., $\mv{z}^{T}=\mv{s}^{T}\mm{A}+\mv{e}^{T}$ for some $\mv{s}\in \Zq^n$ and a ``small'' error $\mv{e}\in \Zq^m$, it outputs $\top$, and otherwise outputs $\bot$. 
Note that it can check that because it knows the trapfoor $\td$ for $\mm{A}$. 

\cref{item:correctness} follows from the fact that the quanutm Fourier transform of the honestly generated $\ket{\Sigma}$ results in a superposition of LWE instances (e.g., see \cite[Proposition 20]{CLZ21}). 
\cref{item:security} holds because if $\qstate{sigma}$ (almost) collapses to a single $\mv{x}$ after the first check of $\fqmver$, the masurement outcome of its quantum Fourier transform is (almost) uniformly distributed over $\Zq^m$, which is an LWE instance with a negligible probability (under an appropriate parameter setting).
Therefore, if it has a non-negligible chance of being accepted, it should not have an overwheling amplitude on a single $\mv{x}$. This means that \cref{item:security} holds.   
\end{proof}

\begin{construction}\label{const:ttQL_SIS}
Our two-tier quantum lightning scheme is described as follows.
\begin{itemize}
\item $\setup(1^\secp)$: Run $(\mm{A},\td)\gets \trapgen(1^\secp)$ and outputs $\pk\seteq \mm{A}$ and $\sk\seteq \td$.
\item $\boltgen(\pk=\mm{A})$: 
Run $(\mv{y},\ket{\Sigma})\gets \fqmgen(\mm{A})$ and 
 outputs $(\snum,\bolt) \seteq (\mv{y},\ket{\Sigma})$.
\item $\fullvrfy(\sk=\td,\snum=\mv{y},\bolt)$: 
This is exactly the same algorithm as $\fqmver(\td,\mv{y},\bolt)$.
 \item $\semivrfy(\pk=\mm{A},\snum=\mv{y},\bolt)$: 
 This algorithm checks if the value $\mv{x}$ in the register of $\bolt$ satisfies $\mm{A}\mv{x}=\mv{y}$ and $\|\mv{x}\|\leq \beta/2$ in superposition by writing the result into another register and measuring it.  
 If that is satisfied, then it outputs $\top$ along with a resulting state $\bolt'$ in the register that stored $\bolt$. Otherwise, it outputs $\bot$. 
\end{itemize}
\end{construction}

Full- and semi-verification correctness directly follows from~\cref{lem:fqm}.
Security is stated as follows:
\begin{theorem}\label{thm:ttQL_SIS}
If we assume the quantum hardness of the $\SIS_{n,m,q,\beta}$, then the above two-tier quantum lightning satisfies two-tier unclonability.
\end{theorem}
\begin{proof}
We show that if two-tier unclonability of \cref{const:ttQL_SIS} is broken, then the SIS problem is also broken.
We construct a QPT adversary $\qB$ for SIS by using a QPT adversary $\qA$ against two-tier QL.
$\qB$ is given a matrix $\mm{A}$ and sends $\pk \seteq \mm{A}$ to $\qA$.
When $\qA$ outputs $(\snum^\ast, \qstate{L}_0,\qstate{L}_1)$, 
 $\qB$ measures $\qstate{L}_0$ and $\qstate{L}_1$. Let the results of the measurement $\mv{x}_0$ and $\mv{x}_1$, respectively. 
Then $\qB$ outputs $\mv{x}_0-\mv{x}_1$.

 Since $\qA$ breaks the security of two-tier QL, $\qstate{L}_0$ and $\qstate{L}_1$ pass $\semivrfy$ and $\fullvrfy$ respectively with non-negligible probability. 
 Thus, by  definitions of $\fullvrfy$ and $\semivrfy$  and \cref{lem:fqm},  we have $\mm{A}\mv{x}_b=\snum^\ast$ and $\|\mv{x}_b\|\leq \beta/2$ for both $b\in \bit$ with non-negligible probability. 
 Therefore, $\B$ succeeds in solving the SIS problem as long as $\mv{x}_0\neq\mv{x}_1$.
 Again, \cref{lem:fqm} ensures that we have $\mv{x}_0\neq\mv{x}_1$ with non-negligible probability.
This completes the proof.  
 
\end{proof}
\else
We can construct a two-tier quantum lightning scheme from the SIS assumption.
The construction is based on the franchised quantum money scheme by Roberts and Zhandry~\cite{RobZha21}.
We provide it in
\ifnum\cameraready=1
the full version~\cite{EPRINT:KitNisYam20}.
\else
\cref{sec:ttQL_SIS}.
\fi
\fi

\subsection{Two-Tier Quantum Lightning with Classical Verification}\label{sec:ttQL_Classical_Vrfy}
We extend two-tier QL to have an algorithm that converts a bolt into a classical certificate which certifies that the bolt was collapsed.
This bolt-to-certificate capability was introduced by Coladangelo and Sattath~\cite{CoRR:ColSat20} for the original QL notion.
We can consider a similar notion for two-tier QL.

\begin{definition}[Two-tier Quantum Lightning with Classical Verification (syntax)]\label{def:ttQL_classical_vrfy}
A two-tier quantum lightning scheme with classical semi-verification is a tuple of algorithms $(\setup,\boltgen,\boltcert,\allowbreak \semivrfy,\certvrfy)$.
\begin{itemize}
\item $\setup(1^\secp)$: The setup algorithm takes as input the security parameter $1^\secp$ and outputs a key pair $(\pk,\sk)$.
\item $\boltgen(\pk)$: The bolt generation algorithm takes as input $\pk$ and outputs a classical string $\snum$ (called a serial number) and a quantum state $\bolt$ (called a bolt for the serial number).
\item $\semivrfy(\pk,\snum,\bolt)$: The semi-verification algorithm takes as input $\pk$, $\snum$, and $\bolt$ and outputs $(\top,\bolt')$ or $\bot$.
\item $\boltcert(\bolt)$: The bolt certification algorithm takes as input $\bolt$ and outputs a classical string $\cert$ (called a certification for collapsing a bolt).
\item $\certvrfy(\sk,\snum,\cert)$: The certification-verification algorithm takes as input $\sk$ and $\cert$ and outputs $\top$ or $\bot$.
\end{itemize}
\end{definition}

\begin{definition}[Correctness for Two-Tier Quantum Lighting with Classical Verification]\label{def:cvttQL_correctness}
There are two verification processes. We say that a two-tier quantum lightning with classical verification is correct if it satisfies the following two properties.
\begin{description}
\item[Semi-verification correctness:] It holds that
\begin{align*}
\Pr\left[
\begin{array}{rl}
(\top,\bolt') \gets \semivrfy(\pk,\snum,\bolt)
\end{array}
\ \middle |
\begin{array}{rl}
 &(\pk,\sk) \gets \setup(1^\secp) \\
 &(\snum,\bolt) \gets \boltgen(\pk)
\end{array}
\right] > 1-\negl(\secp).
\end{align*}

\item[Certification-verification correctness:] It holds that
\begin{align*}
\Pr\left[
\top\gets \certvrfy(\sk,\snum,\cert)\ \middle |
\begin{array}{rl}
 &(\pk,\sk) \gets \setup(1^\secp) \\
 &(\snum,\bolt) \gets \boltgen(\pk)\\
 &\cert \gets \boltcert(\bolt)
\end{array}
\right] > 1-\negl(\secp).
\end{align*}
\end{description}
\end{definition}

\begin{definition}[Reusability for Two-Tier Quantum Lighting with Classical Verification]\label{def:cvttQL_reusability}
A two-tier quantum lightning scheme with classical verification $(\setup,\boltgen,\allowbreak\semivrfy,\boltcert,\certvrfy)$ is reusable if for all $(\pk,\sk)\gets \setup(1^\secp)$, $(\snum,\bolt)\gets \boltgen(\pk)$, and $(\bolt',\top)\gets \semivrfy(\pk,\snum,\bolt)$, it holds that
\[\trdist{\bolt - \bolt'}\le \negl(\secp).\]
\end{definition}
\begin{remark}
Similarly to \cref{rem:reusability}, any two-tier QL scheme with classical verification that satisfies semi-verification correctness can be transformed into one that satisfies reusability. 
Therefore, we focus on correctness.
\end{remark}

\begin{definition}[Two-Tier Unclonability with Classical Verification]\label{def:classical_vrfy_tt_unclonability}
We define the two-tier unclonability 
game between a challenger and an adversary $\A$ in the classical verification setting as follows.

\begin{enumerate}
\item The challenger generates $(\pk,\sk)\gets \setup(1^\secp)$ and $(\snum,\bolt)\gets \boltgen(\pk)$ and sends $\pk$ to $\A$.
\item $\A$ outputs a classical string $\snum$, a quantum state $\qstate{L}$, and a classical string $\keys{CL}$ and sends them to the challenger.
\item The challenger runs $\certvrfy(\sk,\snum,\keys{CL})$ and $\semivrfy(\pk,\snum,\qstate{L})$. If both the outputs are $\top$, then this experiments outputs $1$. Otherwise, it outputs $0$.
\end{enumerate}
This game is denoted by $\expb{\A,\Sigma}{tt}{unclone}{cv}(1^\secp)$.

We say that $\Sigma=(\setup,\boltgen,\semivrfy,\boltcert,\certvrfy)$ is two-tier unclonable if the following holds.
For any QPT adversary $\A$, it holds that
\[
\Pr[\expb{\cA,\Sigma}{tt}{unclone}{cv}(1^\sep)=1] \leq \negl(\secp).
\]
\end{definition}

\begin{definition}[Secure Two-Tier Quantum Lightning with Classical Verification]\label{def:secure_ttQLcv}
A two-tier quantum lightning with classical verification is secure if it satisfies~\cref{def:ttQL_classical_vrfy,def:cvttQL_correctness,def:cvttQL_reusability,def:classical_vrfy_tt_unclonability}.
\end{definition}

Note that a two-tier quantum lightning scheme with classical verification can be easily transformed into an ordinary two-tier quantum lightning scheme.
This is done by setting the latter's full-verification algorithm as the combination of the bolt certification algorithm and the certification-verification algorithm of the former.
Namely, we have the following theorem.

\begin{theorem}\label{thm:ttQLcv_implies_ttQL}
If there exists two-tier quantum lightning with classical verification, then there also exists ordinary two-tier quantum lightning.
\end{theorem}

\subsection{Two-Tier Quantum Lightning with Classical Verification from LWE}\label{sec:cvttQL_LWE}



In this section, we show how to construct a two-tier QL scheme with classical verification from the LWE assumption. 
First, we define an amplified version of the adaptive hardcore bit property of an NTCF family. 

\begin{definition}[Amplified Adaptive Hardcore Property]\label{def:amplified_adaptive_hardcore}
 We say that a NTCF family $\cF$ (defined in \cref{def:NTCF}) satisfies the amplified adaptive hardcore property if  for any QPT $\A$ and $n=\omega(\log \secp)$, it holds that
 \begin{align*}
\Pr\left[
 \begin{array}{ll}
 \forall i\in[n]~
 x_i=x_{i,b_i},\\
d_i\in  G_{\fk,0,x_{i,0}} \cap G_{\fk,1,x_{i,1}},\\
m_i=d_i\cdot (J(x_{i,0})\oplus J(x_{i,1}))
 \end{array}
 \middle | 
 \begin{array}{ll}
 (\fk_i,\td_i)\gets \NTCF.\gen_\cF(1^\secp) \text{~for~} i\in[n] \\
 (\{(b_i,x_i,y_i,d_i,m_i)\}_{i\in[n]})\gets \A(\{\fk_i\}_{i\in[n]})\\
 x_{i,\beta} \gets \inv_\cF (\td_i,\beta,y_i)\text{~for~}(i,\beta)\in[n]\times \bit
 \end{array}
 \right]=\negl(\secp).
 \end{align*}
\end{definition}
As implicitly shown in \cite{AFT:RadSat19}, any NTCF family satisfies the amplified adaptive hardcore property.\footnote{\cite{AFT:RadSat19} proved essentially the same lemma through an abstraction which they call \emph{1-of-2 puzzle}.}
\begin{lemma}[Implicit in \cite{AFT:RadSat19}]\label{lem:amplified_adaptive_hardcore}
Any NTCF family  satisfies the amplified adaptive hardcore property. 
\end{lemma}
\begin{proof}(sketch.)
This proof sketch is a summary of the proof in \cite{AFT:RadSat19}. 
Canetti et al. \cite{TCC:CanHalSte05} proved that a parallel repetition exponentially decreases hardness of \emph{weakly verifiable puzzle}, which is roughly a computational problem whose solution can be verified by a secret verification key generated along with the problem.
Though Canetti et al. only considered hardness against classical algorithms, Radian and Sattath \cite{AFT:RadSat19} observed that a similar result holds even for quantum algorithms. 
Then we consider a weakly verifiable puzzle described below: 
\begin{enumerate}
\item A puzzle generation algorithm runs  $(\fk,\td)\sample  \NTCF.\gen_\cF(1^\secp)$  and publishes  $\fk$ as a puzzle while keeping $\td$ as a secret verification key.
\item We say that $(b,x,y,d,m)$ is a valid solution to the puzzle $\fk$ if  
it holds that  $x=x_{b}$, 
 $d\in  G_{\fk,0,x_0} \cap G_{\fk,1,x_1}$, and
$m=d \cdot (J(x_{0})\oplus J(x_{1}))$ where  $x_{\beta} \gets \inv_\cF (\td,\beta,y)$ for $\beta\in \bit$. 
\end{enumerate}
We can see that the adaptive hardcore property implies that a QPT algorithm can find a valid solution of the above weakly verifiable puzzle  with probability at most $\frac{1}{2}+\negl(\secp)$. 
By applying the amplification theorem of \cite{TCC:CanHalSte05,AFT:RadSat19}  as explained above, $n=\omega(\log(\secp))$-parallel repetition version of the above protocol is hard for any QPT algorithm to solve with non-negligible probability.
This is just a rephrasing of amplified adaptive hardcore property. 
\end{proof}

\paragraph{Two-Tier Quantum Lightning from NTCF.}
We show how to construct a two-tier QL scheme with classical verification from an NTCF family.

\begin{construction}\label{const:cvttQL_aNTCF}
Let $n=\omega(\log \secp)$. 
Our two-tier QL with classical verification scheme is described as follows.
\begin{itemize}
\item $\setup(1^\secp)$: 
Generate $(\fk_i,\td_i)\gets \NTCF.\gen_\cF (1^\secp)$ for $i\in[n]$ and set $(\pk,\sk)\seteq (\{\fk_i\}_{i\in[n]},\{\td_i\}_{i\in[n]})$.
\item $\boltgen(\pk)$: Parse $\pk = \{\fk_i\}_{i\in[n]}$.
For each $i\in[n]$, generate a quantum state
\[
\ket{\psi'_{i}} = \frac{1}{\sqrt{\abs{\cX}}}\sum_{x\in \cX, y\in\cY, b\in\zo{}}\sqrt{(f'_{\fk_i,b}(x))(y)}\ket{b,x}\ket{y}
\]
by using $\samp_\cF$, measure the last register to obtain $y_i\in\cY$, and let $\ket{\phi'_i}$ be the post-measurement state where the measured register is discarded.
Output $(\snum,\bolt)\seteq (\{y_i\}_{i\in[n]},\left\{\ket{{\phi'}_{i}}\right\}_{i\in[n]})$.
\item $\semivrfy(\pk,\snum,\bolt)$: 
Parse $\pk = \{\fk_i\}_{i\in[n]}$, $\snum=\{y_i\}_{i\in[n]}$, $\bolt=\{\bolt_i\}_{i\in[n]}$.
For each $i\in[n]$,  check if the value $(b_i,x_i)$ in the register of $\bolt_i$ satisfies $y\in\support(f'_{\fk_i,b_i}(x_i))$ in superposition by writing the result to another register and measuring it. We note that this procedure can be done efficiently without using $\td_i$ since $y\in\support(f'_{\fk_i,b_i}(x_i))$ can be publicly checked by using $\CHK_\cF$  as defined in \cref{def:NTCF}.   
If the above verification passes for all $i\in [n]$, then output $\top$ and the post-measurement state (discarding measured registers).   
Otherwise, output $\bot$. 
\item $\boltcert(\bolt)$: 
Parse $\bolt=\{\bolt_i\}_{i\in[n]}$.
For each $i\in [n]$, do the following:
Evaluate the function $J$ on the second register of $\bolt_i$. That is, apply a unitary that maps $\ket{b,x}$ to $\ket{b,J(x)}$ to $\bolt_i$. (Note that this can be done efficiently since $J$ is injective and efficiently invertible.) Then, apply Hadamard transform and  measure both registers to obtain $(m_i,d_i)$. Output $\cert\seteq \{(d_i,m_i)\}_{i\in[n]}$.
\item $\certvrfy(\sk,\snum,\cert)$: Parse $\sk = \{\td_i\}_{i\in[n]}$, $\snum =\{y_i\}_{i\in[n]}$, and $\cert =\{(d_i,m_i)\}_{i\in[n]}$.
For each $i\in[n]$ and $\beta\in\bit$, compute $x_{i,\beta}\gets \inv_\cF(\td_i,\beta,y_i)$. 
Output $\top$ if and only if it holds that  
$d_i\in  G_{\fk,0,x_{i,0}} \cap G_{\fk,1,x_{i,1}}$ and 
$m_i=d_i\cdot (J(x_{i,0})\oplus J(x_{i,1}))$ for all $i\in[n]$.  
\end{itemize}
\end{construction}

\begin{theorem}\label{thm:cvttQL_aNTCF}
If there exists an NTCF family, there exists a two-tier QL with classical verification.
\end{theorem}

\begin{proof}[Proof of~\cref{thm:cvttQL_aNTCF}]
We prove correctness and two-tier unclonability below:
\paragraph{Correctness of certification-verification.}
We need to prove that if 
$\cert$ is generated by $\boltcert(\bolt)$ for an honestly generated $\bolt$ corresponding a serial number $\snum$, $\certvrfy(\sk,\snum,\cert)$ returns $\top$ with overwhelming probability. 

For each $i\in[n]$, if we define a quantum state 
\[
\ket{\psi_{i}} = \frac{1}{\sqrt{\abs{\cX}}}\sum_{x\in \cX, y\in\cY, b\in\zo{}}\sqrt{(f_{\fk_i,b}(x))(y)}\ket{b,x}\ket{y},
\]
then we have 
\[
\trdist{ \ket{\psi_i}\bra{\psi_i} - \ket{\psi'_i}\bra{\psi'_i}} \leq \negl(\secp),
\]
\ifnum\cameraready=1
as observed in \cref{def:NTCF}.
\else
as observed in \cref{def:NTCF} (where we used~\cref{lem:Hellinger_trace_distance}).
\fi
Therefore, even if we replace $\ket{\psi'_i}$ with $\ket{\psi_i}$ for each $i\in[n]$ in the execution of  $\boltgen(\pk)$ to generate $\bolt$, the probability that $\certvrfy(\sk,\snum,\cert)$ returns $\top$ only negligibly changes.
Therefore, it suffices to prove that $\certvrfy(\sk,\snum,\cert)$ returns $\top$ with overwhelming probability in a modified experiment where $\ket{\psi'_i}$ is replaced with $\ket{\psi_i}$ for each $i\in[n]$.\footnote{Of course, such a replacement cannot be done efficiently. We consider such an experiment only as a proof tool.}
In this experiment, if we let $\bolt_i$ be  the $i$-th component of $\bolt$, then we have 
\[
\bolt_i=\frac{1}{\sqrt{2}}(\ket{0,x_{i,0}}+\ket{1,x_{i,1}})
\]
for each $i\in[n]$ 
where  $x_{i,\beta} \gets \inv_\cF (\td_i,\beta,y_i)$ for $\beta \in \bit$ 
by the injective property of $\cF$.   
If we apply $J$ to the second register of $\bolt_i$ and then apply Hadamard transform for both registers as in $\boltcert$, then the resulting state can be written as 
\begin{align*}
&2^{- \frac{w+2}{2}}\sum_{d,b,m} (-1)^{d\cdot J(x_{i,b})\xor m b }\ket{m}\ket{d}\\
& = 2^{-\frac{w}{2}}\sum_{d \in \zo{w}}(-1)^{d\cdot J(x_{i,0})}\ket{d\cdot (J(x_{i,0})\xor J(x_{i,1}))}\ket{d}.
\end{align*}
Therefore, the measurement result is $(m_i,d_i)$ such that 
$m_i=d_i\cdot (J(x_{i,0})\xor J(x_{i,1}))$ 
for a uniform $d_i\sample \bit^w$. 
By the adaptive hardcore bit property (the first item) in~\cref{def:NTCF}, it holds that $d_i \in G_{\fk_i,0,x_{i,0}}\cap G_{\fk_i,1,x_{i,1}}$ except negligible probability.
Therefore, the certificate $\cert=\{(d_i,m_i)\}_{i\in[n]}$ passes the verification by  $\certvrfy$ with overwhelming probability.

\paragraph{Correctness of semi-verification.}
Let $\bolt=\{\phi'_{i}\}_{i\in[n]}$ be an honestly generated bolt. 
By the definition of $\boltgen$, $\ket{\phi_i}$ is a superposition of $(b,x)$ such that $y\in \support(f'_{\fk_i,b}(x))$. 
This clearly passes the verification by $\semivrfy$.  
\paragraph{Two-tier unclonability.}
As shown in \cref{lem:amplified_adaptive_hardcore}, any NTCF family satisfies the amplified adaptive hardcore property.
We show that if there exists a QPT adversary $\qA$ that breaks the two-tier unclonability with classical verification of~\cref{const:cvttQL_aNTCF} with probability $\epsilon$, we can construct a QPT adversary $\qB$ that breaks the amplified adaptive hardcore property the NTCF with probability $\epsilon$.

$\qB$ is given $\{\fk_i\}_{i\in[n]}$ and sends $\pk \seteq \{\fk_i\}_{i\in[n]}$ to $\qA$ this implicitly sets $\sk \seteq \{\td_i\}_{i\in[n]}$). 
When $\qA$ outputs $(\snum,\qstate{L},\cert)$, $\qB$ parses 
$\snum=\{y_i\}_{i\in[n]}$, 
$\qstate{L}=\{\qstate{L}_i\}_{i\in[n]}$, and $\cert=\{(d_i,m_i)\}_{i\in[n]}$, 
measures $\qstate{L}_i$ to obtain $(b_i,x_i)$ for each $i\in[n]$, and outputs $\{(b_i,x_i,y_i,d_i,m_i)\}_{i\in[n]}$.

By assumption on $\qA$, it holds that $\semivrfy(\pk,\snum,\qstate{L}) = \top$ and $\certvrfy(\sk,\allowbreak\snum,\cert)=\top$ with probability $\epsilon$.
If  $\semivrfy(\pk,\snum,\qstate{L}) = \top$ holds, we have $y_i\in\support(f'_{\fk_i,b_i}(x_{i}))$ for each $i\in[n]$ by the construction of $\semivrfy$.
We note that  $y_i\in\support(f'_{\fk_i,b_i}(x_{i}))$  implies $x_{i}=x_{i,b_i}$  by the efficient range superposition property of \cref{def:NTCF} where $x_{i,\beta}\gets \inv_\cF (\td_i,\beta,y_i)$ for $\beta\in\bit$.
If $\certvrfy(\sk,\snum,\cert)=\top$ we have $d_i\in  G_{\fk,0,x_{i,0}} \cap G_{\fk,1,x_{i,1}}$ and 
$m_i=d_i\cdot (J(x_{i,0})\oplus J(x_{i,1}))$ for all $i\in[n]$.
Clearly, $\qB$ wins the amplified adaptive hardcore game when both of them happen, which happens with probability $\epsilon$ by the assumption. 
This completes the proof.
\end{proof}

By combining~\cref{thm:NTCF_LWE,thm:cvttQL_aNTCF}, the following corollary immediately follows.
\begin{corollary}\label{thm:ttQLcv_LWE}
If we assume the quantum hardness of the LWE problem, there exists a secure two-tier QL with classical verification.
\end{corollary}


\section{Relaxed Watermarking}\label{sec:relaxed_watermarking}

In this section, we introduce the notion of relaxed watermarking and concrete constructions of relaxed watermarking.
\ifnum\cameraready=0
\else
Due to the space limitation, we only provide the construction for PRF.
For the construction for compute-and-compare circuits, see \cite{EPRINT:KitNisYam20}.
\fi

\subsection{Definition of Relaxed Watermarking}\label{sec:definition_relaxed_watermarking}

We introduce the definition of relaxed watermarking.
The following definition captures publicly markable and extractable watermarking schemes.
After the definition, we state the difference between relaxed watermarking and classical cryptographic watermarking~\cite{SIAMCOMP:CHNVW18}.

\begin{definition}[Relaxed Watermarking Syntax]\label{def:relaxed_watermarking_syntax}
Let $\cC=\setbracket{\cC_{\secp}}_{\secp}$ be a circuit class such that $\cC_{\secp}$ contains circuits of input length is $n$ and output length $m$.
A relaxed watermarking scheme for the circuit class
$\cC$ and a message space $\M = \{\M_\secp\}_{\secp}$ consists of four PPT algorithms $(\Gen,\Mark,\Extract,\Eval)$.

\begin{description}
 \item[Key Generation:] $\Gen(1^\secp)$ takes as input the security parameter and outputs a public parameter $\pp$. 
\item[Mark:] $\Mark (\pp,C, \msg)$ takes as input a public parameter, an arbitrary circuit $C\in\cC_{\secp}$ and a message $\msg \in \M_\secp$ and outputs a marked circuit $\tlC$.
\item[Extract:] $\msg' \gets \Extract (\pp,C^{\prime})$ takes as input a public parameter and an arbitrary circuit $C^\prime$, and outputs a message $\msg'$, where $\msg' \in \M_\secp \cup \{\unmarked\}$.
\item[Honest Evaluation:] $\Eval(\pp,C^\prime,x)$ takes as input a public parameter, an arbitrary circuit $C^\prime$, and an input $x$, and outputs $y$.
\end{description}
\end{definition}

We define the required correctness and security properties of a watermarking scheme.
\begin{definition}[Relaxed Watermarking Property]\label{def:relaxed_watermarking_property}
A watermarking scheme  $(\Gen,\allowbreak\Mark,\Extract,\Eval)$ for circuit family $\setbracket{\cC_{\secp}}_{\secp}$ and with message space $\M = \{\M_\secp\}_{\lambda}$ is required to satisfy the following properties.
\begin{description}
 \item[Statistical Correctness:] For any circuit $C \in \cC_{\secp}$, any message $\msg \in \M_{\secp}$, it holds that
\[
 \Pr\left[\forall x\ \Eval(\pp,\tlC,x)=C(x)~\left|~ 
 \begin{array}{c}
 \pp \gets \Gen(1^\secp)\\
  \tlC \gets \Mark    (\pp,C,\msg)
  \end{array}\right.\right] \geq 1 - \negl(\secp).
\]
\item[Extraction Correctness:]  For every $C \in \cC_{\secp}$, $\msg \in \M_\secp$ and $\pp \gets \Gen(1^\secp)$:
\[
 \Pr[\msg' \ne \msg ~\left|~ \msg' \gets \Extract (\pp,\Mark (\pp,C,\msg)) \right.] \leq \negl(\secp) .
\]
\item[Relaxed $(\epsilon,\cD_{\cC})$-Unremovability:] For every QPT $\qA$, we have
\[
 \Pr[\expa{\qA,\cD_{\cC}}{r}{urmv}(\secp,\epsilon) \out 1] \leq  \negl(\secp)
\]
where $\epsilon$ is a parameter of the scheme called the \emph{approximation factor}, $\cD_{\cC}$ is a distribution over $\cC_{\secp}$, and $\expa{\qA,\cD_{\cC}}{r}{urmv}(\secp,\epsilon)$
	   is the game defined next.
\end{description}
We say a watermarking scheme is relaxed $(\epsilon,\cD_{\cC})$-secure if it satisfies these properties.
\end{definition}
\begin{definition}[Relaxed $(\epsilon,\cD_{\cC})$-Unremovability Game]\label{def:relaxed_unremovability}
The game $\expa{\qA,\cD_{\cC}}{r}{urmv}(\secp,\epsilon)$  is
 defined as follows.
\begin{enumerate}
\item The challenger generates $\pp \gets \Gen(1^\secp)$ and gives $\pp$ to the adversary $\qA$.
\item At some point, $\qA$ sends a message $\msg \in \M_\secp$ to the challenger.
The challenger samples a circuit $C \chosen \cD_{\cC}$ and responds with $\tlC  \gets \Mark (\pp,C,\msg)$.
\item Finally, the adversary outputs a circuit $C^{\ast}$. If it holds that 
\[
\Pr_{x\la\zo{n}}[\Eval(\pp,C^{\ast},x) = C(x)]\ge\epsilon
\] and $\Extract (\pp, C^{\ast}) \ne \msg$, then the challenger outputs $1$, otherwise $0$.
\end{enumerate}
\end{definition}

Differently from the definition by Cohen et al.~\cite{SIAMCOMP:CHNVW18}, the above definition requires a watermarking scheme has an honest evaluation algorithm for running programs. In the unremovability game above, adversaries must output a circuit whose behavior is close to the original circuit when it is executed using the honest evaluation algorithm.

Relaxed watermarking is clearly weaker than classical watermarking.
However, in this work, watermarking is just an intermediate primitive, and relaxed watermarking is sufficient for our goal of constructing SSL schemes.
Moreover, this relaxation allows us to achieve a public extractable watermarking scheme for a PRF family under the LWE assumption, as we will see in \cref{sec:const_relaxed_watermarking_prf}.
For classical watermarking, we currently need IO to achieve such a scheme~\cite{SIAMCOMP:CHNVW18}.

\subsection{Relaxed Watermarking for PRF}\label{sec:const_relaxed_watermarking_prf}
We construct a relaxed watermarking scheme for PRFs from puncturable PRFs and true-simulation extractable NIZK.

\begin{construction}[Relaxed Watermarking for PRF]\label{const:relaxed_watermarking_prf}
Let $\PuncPRF=(\PRF.\Eval,\Puncture,\allowbreak\PRF.\pEval)$ be a puncturable PRF whose key space, domain, and range are $\prfkeyspace$, $\Dprf$, and $\Rprf$, respectively.
Also, let $\NIZK=(\NIZK.\setup,\NIZK.\prove\allowbreak,\NIZK.\vrfy)$ be a NIZK system for $\NP$.
Using these building blocks, we construct a relaxed watermarking scheme for the PRF family $\{\prf_\prfkey(\cdot)=\PRF.\Eval(\prfkey,\cdot) \mid \prfkey\in\prfkeyspace\}$ as follows.
Its message space is $\bit^\msglen$ for some polynomial $\msglen$ of $\secp$.
In the construction, $\mv{0}$ is some fixed point in $\Dprf$.
\begin{description}
\item[$\gen(1^\secp)$:] Compute $\crs \gets \NIZK.\setup(1^\secp)$ and Output $\pp \seteq \crs$. 
\item[$\Mark(\pp,\prf_\prfkey,\msg)$:] Compute $y_{\mv{0}} \gets \PRF.\Eval(\prfkey,\mv{0})$ and $\prfkey_{\setbracket{\mv{0}}} \gets \Puncture(\prfkey,\setbracket{\mv{0}})$.
Let an NP relation $\cR_L$ be as follows.
\[
\cR_L \seteq \left\{\left((\msg,y_{\mv{0}},\prfkey_{\setbracket{\mv{0}}}),\prfkey \right) \mid
y_{\mv{0}} = \PRF.\Eval(\prfkey,\mv{0}),
\prfkey_{\setbracket{\mv{0}}} = \Puncture(\prfkey,\setbracket{\mv{0}}),\textrm{~and~}\prfkey\in\prfkeyspace
\right\}.
\]
Compute $\pi \gets \NIZK.\Prove(\crs,(\msg,y_{\mv{0}},\prfkey_{\setbracket{\mv{0}}}),\prfkey)$.
Output $\tlC \seteq (\msg,y_{\mv{0}},\prfkey_{\setbracket{\mv{0}}},\pi)$.
\item[$\Extract(\pp,C^\prime)$:] Parse $C^\prime = (\msg^\prime,y^\prime,K^\prime,\pi^\prime)$ and output $\msg^\prime$.
\item[$\Eval(\pp,C^\prime,x)$:] Parse $C^\prime = (\msg^\prime,y^\prime,K^\prime,\pi^\prime)$ and run $\NIZK.\vrfy(\crs,(\msg^\prime,y',K^\prime),\pi)$. If the output is $\bot$, output $\bot$. Otherwise, output $\PRF.\pEval(K^\prime,x)$ for $x \ne \mv{0}$ and $y^\prime$ for $x=\mv{0}$.
\end{description}
\end{construction}

\begin{theorem}\label{thm:relaxed_watermarking_prf}
Let $\epsilon$ be any inverse polynomial of $\secp$ and $\cU_\prfkeyspace$ the uniform distribution over $\prfkeyspace$.
If $\PuncPRF$ is a puncturable PRF with key-injectiveness and $\NIZK$ is a true-simulation extractable NIZK system for $\NP$, then \cref{const:relaxed_watermarking_prf} is a relaxed $(\epsilon,\cU_\prfkeyspace)$-secure watermarking scheme for the PRF family $\{\prf_\prfkey(\cdot)=\PRF.\Eval(\prfkey,\cdot) \mid \prfkey\in\prfkeyspace\}$.
\end{theorem}

\ifnum\cameraready=0

\ifnum\submission=0
\begin{proof}[Proof of~\cref{thm:relaxed_watermarking_prf}]
\else
\section{Proof of \cref{thm:relaxed_watermarking_prf}}\label{sec:proof_relaxed_watermarking_prf}
\fi
The statistical correctness of \cref{const:relaxed_watermarking_prf} follows from the completeness of $\NIZK$ and the functionality preserving under puncturing of $\PuncPRF$.
Also, the extraction correctness of \cref{const:relaxed_watermarking_prf} immediately follows from the construction.
Below, we prove the relaxed $(\epsilon,\cU_\prfkeyspace)$-unremovability of \cref{const:relaxed_watermarking_prf}.

Let $\qA$ be a QPT adversary attacking relaxed $(\epsilon,\cU_\prfkeyspace)$-unremovability.
We prove this theorem using hybrid games.

\begin{description}
\item[Game $1$:]This is $\expa{\qA,\cU_\prfkeyspace}{r}{urmv}(\secp,\epsilon)$ for \cref{const:relaxed_watermarking_prf}.

\begin{enumerate}
\item The challenger generates $\crs \gets \NIZK.\setup(1^\secp)$ and gives $\pp \seteq \crs$ to the adversary $\qA$.
\item At some point, $\qA$ queries a message $\msg \in \bit^\msglen$ to the challenger.
The challenger first samples $\prfkey \gets\cU_\prfkeyspace$.
Next, the challenger computes $y_{\mv{0}} \gets \PRF.\Eval(\prfkey,\mv{0})$, $\prfkey_{\setbracket{\mv{0}}} \gets \Puncture(\prfkey,\setbracket{\mv{0}})$, and $\pi \gets \NIZK.\Prove(\crs,(\msg,y_{\mv{0}},\prfkey_{\setbracket{\mv{0}}}),\prfkey)$.
Then, the challenger returns $\tlC \seteq (\msg,y_{\mv{0}},\prfkey_{\setbracket{\mv{0}}},\pi)$ to $\qA$.

\item Finally, $\qA$ outputs a circuit $C^{\ast}=(\msg^*,y^*,\prfkey^*,\pi^*)$. If $\Pr_{x\gets \Dprf}[\Eval(\pp,C^{\ast},x) = \PRF.\Eval(\prfkey,x)]\ge\epsilon$ and $\Extract (\pp, C^{\ast})=\msg^* \ne \msg$ hold, then the challenger outputs $1$ as the output of this game. Otherwise, the challenger outputs $0$ as the output of this game.
\end{enumerate}
\end{description}

We define the following three conditions.

\begin{itemize}
\item[$(a)$] $\Pr_{x\gets \Dprf}[\Eval(\pp,C^{\ast},x) = \PRF.\Eval(\prfkey,x)]\ge\epsilon$.
\item[$(b)$] $\NIZK.\vrfy(\crs,(\msg^*,y^*,\prfkey^*),\pi^*)=\top$.
\item[$(c)$] $\msg^*\neq \msg$.
\end{itemize}

It is clear that if all of the above conditions are satisfied, the output of Game $1$ is $1$.
In the opposite direction, it is clear that the conditions $(a)$ and $(c)$ are satisfied whenever the output of Game $1$ is $1$ from the definition of Game $1$.
Also, we see that if the condition $(b)$ is not satisfied, $\Pr_{x\gets \Dprf}[\Eval(\pp,C^{\ast},x) = \PRF.\Eval(\prfkey,x)]=0$ holds and thus the output of Game $1$ is $0$.
Therefore, the conditions $(b)$ is satisfied whenever the output of Game $1$ is $1$.
Overall, the output of Game $1$ is $1$ if and only if the above three conditions hold in Game $1$.

We define $\event{S}$ as the event that the above conditions $(b)$ and $(c)$, and the following condition hold.
\begin{itemize}
\item[$(a')$]  Let $\omega=\secp/\epsilon$. $\Eval(\pp,C^{\ast},x_j)=\PRF.\Eval(\prfkey,x_j)$ holds for some $j\in[\omega]$, where $x_j$ is randomly chosen from $\Dprf$ for every $j\in[\omega]$. 
\end{itemize}
When the condition $(a)$ is satisfied, the probability that $(a')$ is not satisfied is bounded by $(1-\epsilon)^{\lambda/\epsilon}\le e^{-\lambda}=\negl(\secp)$.
Thus, we have $\Pr[\textrm{Output of Game~} 1 \textrm{~is~} 1]\leq \Pr[\event{S}]+\negl(\secp)$.

We next consider the following adversary $\qAnizk$ attacking the true-simulation extractability of $\NIZK$ using $\qA$.

\begin{enumerate}
\item Given $\crs$, $\qAnizk$ gives $\pp \seteq \crs$ to $\qA$.
\item When $\qA$ queries a message $\msg \in \bit^\msglen$, $\qAnizk$ first samples $\prfkey \gets \cU_\prfkeyspace$.
Next, $\qAnizk$ computes $y_{\mv{0}} \gets \PRF.\Eval(\prfkey,\mv{0})$ and $\prfkey_{\setbracket{\mv{0}}} \gets \Puncture(\prfkey,\setbracket{\mv{0}})$.
Then, $\qAnizk$ sends a statement/witness pair $((\msg,y_{\mv{0}},\prfkey_{\setbracket{\mv{0}}}),\prfkey)$ to the challenger.

\item Given $\pi$, $\qAnizk$ sends $\tlC \seteq (\msg,y_{\mv{0}},\prfkey_{\setbracket{\mv{0}}},\pi)$ to $\qA$.

\item When $\qA$ outputs $C^{\ast}=(\msg^*,y^*,\prfkey^*,\pi^*)$, $\qAnizk$ first randomly chooses $x_j$ from $\Dprf$ for every $j\in[\omega]$ and checks whether $\Eval(\pp,C^{\ast},x_j)=\PRF.\Eval(\prfkey,x_j)$ holds for some $j\in[\omega]$.
If so, $\qAnizk$ outputs a statement/proof pair $((\msg^*,y^*,\prfkey^*),\pi^*)$. Otherwise, $\qAnizk$ outputs $\bot$.
\end{enumerate}
When we execute $\expt{\qAnizk,\NIZK}{se\textrm{-}real}$, the output of it is $1$ if and only if the following conditions hold.
\begin{itemize}
\item  Let $\omega=\secp/\epsilon$. $\Eval(\pp,C^{\ast},x_j)=\PRF.\Eval(\prfkey,x_j)$ holds for some $j\in[\omega]$, where $x_j$ is randomly chosen from $\Dprf$ for every $j\in[\omega]$. 
\item $\NIZK.\vrfy(\crs,(\msg^*,y^*,\prfkey^*),\pi^*)=\top$.
\item $((\msg,y_{\mv{0}},\prfkey_{\setbracket{\mv{0}}}),\prfkey)\in \cR_L$.
\item $(\msg,y_{\mv{0}},\prfkey_{\setbracket{\mv{0}}})\neq(\msg^*,y^*,\prfkey^*)$.
\end{itemize}
$\qAnizk$ perfectly simulates Game $1$ until $\qA$ terminates.
We see that when the event $\event{S}$ occurs in the simulated Game 1, the output of $\expt{\qAnizk,\NIZK}{se\textrm{-}real}$ is $1$.
Namely, we have $\Pr[\event{S}]\leq\Pr[1\gets\expt{\qAnizk,\NIZK}{se\textrm{-}real}]$.

Since $\NIZK$ satisfies true-simulation extractability, there exists $\Sim=(\fksetup,\allowbreak\Sim_1,\Sim_2)$ such that we have
\[
\abs{\Pr[1\gets\expt{\qAnizk,\NIZK}{se\textrm{-}real}]-\Pr[1\gets\expt{\qAnizk,\Sim,\NIZK}{se\textrm{-}sim}]}\leq\negl(\secp).
\]
We then define the following Game~2.

\begin{description}
\item[Game $2$:]This game is the same as $\expt{\qAnizk,\Sim,\NIZK}{se\textrm{-}sim}$ except conceptual changes. Especially, this game is obtained by transforming $\expt{\qAnizk,\Sim,\NIZK}{se\textrm{-}sim}$ into a security game played between the challenger and $\qA$ so that the output distribution does not change.
\begin{enumerate}
\item The challenger generates $(\crs,\td)\gets\fksetup(1^\secp)$ and gives $\pp \seteq \crs$ to $\qA$.
\item When $\qA$ queries a message $\msg \in \bit^\msglen$, the challenger first samples $\prfkey \gets \cU_\prfkeyspace$.
Next, the challenger computes $y_{\mv{0}} \gets \PRF.\Eval(\prfkey,\mv{0})$ and $\prfkey_{\setbracket{\mv{0}}} \gets \Puncture(\prfkey,\setbracket{\mv{0}})$.
Then, the challenger computes $(\pi,\state_{\Sim})\gets\Sim_1(\crs,\td,(\msg,y_{\mv{0}},\prfkey_{\setbracket{\mv{0}}}))$ and sends $\tlC \seteq (\msg,y_{\mv{0}},\prfkey_{\setbracket{\mv{0}}},\pi)$ to $\qA$.

\item When $\qA$ outputs $C^{\ast}=(\msg^*,y^*,\prfkey^*,\pi^*)$, the challenger computes $\prfkey'\gets\Sim_2(\state_{\Sim},(\msg^*,y^*,\prfkey^*),\pi^*)$.
The challenger then outputs $1$ if all of the following conditions hold.
\begin{itemize}
\item For $\{x_j\}_{j\in[\omega]}$ randomly chosen from $\Dprf$, $\Eval(\pp,C^{\ast},x_j)=\PRF.\Eval(\prfkey,x_j)$ holds for some $j\in[\omega]$.
\item $\NIZK.\vrfy(\crs,((\msg^*,y^*,\prfkey^*),\pi^*))=\top$.
\item $((\msg,y_{\mv{0}},\prfkey_{\setbracket{\mv{0}}}),\prfkey)\in \cR_L$.
\item $((\msg^*,y^*,\prfkey^*),\prfkey')\in \cR_L$.
\item $(\msg,y_{\mv{0}},\prfkey_{\setbracket{0}})\neq (\msg^*,y^*,\prfkey^*)$.
\end{itemize}
Otherwise, the challenger outputs $0$.
\end{enumerate}
\end{description}
When the above first condition and fourth condition hold, we have $\PRF.\Eval(\prfkey,x_j)=\PRF.\Eval(\prfkey',x_j)$.
Then, from the key-injective property of $\PuncPRF$, we also have $\prfkey=\prfkey'$.
Therefore, from the security of $\PuncPRF$, we have $\Pr[\textrm{Output of Game~} 2 \textrm{~is~} 1]\leq\negl(\lambda)$.

From the discussions so far, we obtain $\Pr[\textrm{Output of Game~} 1 \textrm{~is~} 1]\leq\negl(\lambda)$.
This completes the proof.
\ifnum\submission=0
\end{proof}
\else
\fi
\else
Due to the space limitation, we provide the proof of \cref{thm:relaxed_watermarking_prf} in the full version~\cite{EPRINT:KitNisYam20}.
\fi

\medskip

\ifnum\cameraready=1
By using known results (see the full version~\cite{EPRINT:KitNisYam20} for the detail), we can instantiate \cref{const:relaxed_watermarking_prf} under the LWE assumption.
\else
From \cref{thm:keyinj_pPRF_LWE} and \cref{thm:seNIZK_LWE}, we can instantiate \cref{const:relaxed_watermarking_prf} under the LWE assumption.
\fi
Concretely, we obtain the following theorem.
\begin{theorem}\label{thm:relaxed_watermarking_prf_LWE}
Let $\epsilon$ be any inverse polynomial of $\secp$.
Assuming the quantum hardness of the LWE problem, there is a relaxed $(\epsilon,\cU_\prf)$-secure watermarking scheme for a family of PRF $\cF$, where $\cU_\prf$ is the uniform distribution over $\cF$.
\end{theorem}

\ifnum\submission=0

\ifnum\submission=0
\subsection{Relaxed Watermarking for Compute-and-Compare Circuits}
\else
\section{Relaxed Watermarking for Compute-and-Compare Circuits}
\fi
\label{sec:const_relaxed_watermarking_cnc}

We give a construction of relaxed watermarking for circuits called (searchable) compute-and-compare circuits.
The construction is essentially the classical part of the SSL construction by Ananth and La Placa~\cite{EC:AnaLaP21}.
Note that their construction uses a primitive called input-hiding obfuscation. However, our construction instead uses injective one-way functions that can be seen as a concrete instantiation of input-hiding obfuscation.

Below, we first define a family of compute-and-compare circuits and then provide the construction of a relaxed watermarking scheme for it.

\begin{definition}[Compute-and-Compare Circuits]
A compute-and-compare circuit $\cnc{C}{\alpha}$ is of the form
\[
\cnc{C}{\alpha}(x)\left\{
\begin{array}{ll}
1&(C(x)=\alpha)\\
0&(\text{otherwise})~,
\end{array}
\right.
\]
where $C$ is a circuit and $\alpha$ is a string called lock value.
We let $\Ccnc=\{\cnc{C}{\alpha}|C:\bit^n\rightarrow\bit^m,\alpha\in\bit^m\}$.
\begin{description}
\item[Searchability:]We say that a family of compute-and-compare circuits $\Ccnc=\{\cnc{C}{\alpha}|C:\bit^n\rightarrow\bit^m,\alpha\in\bit^m\}$ is searchable if there exists a PPT algorithm $\cS$ such that given any $\cnc{C}{\alpha}\in\Ccnc$, $\cS$ outputs $x\in\bit^n$ such that $\cnc{C}{\alpha}(x)=1$ (i.e., $C(x)=\alpha$).
\end{description}
\end{definition}

\paragraph{Distribution of interest.}
For a function $\gamma(\secp)$, we say that a distribution $\Dcnc{\gamma}$ over $\Ccnc$ has conditional min-entropy $\gamma$ if $\cnc{C}{\alpha}\la\Dcnc{\gamma}$ satisfies $\Hmin(\alpha|C)\ge\gamma(\lambda)$.

\begin{construction}[Relaxed Watermarking for Searchable Compute and Compare Circuits]\label{const:relaxed_watermarking_cnc}
Let $\inplen$, $\outlen$, $\ell$ be polynomials of $\secp$.
Let $\Fow=\{\owf:\bit^\outlen\rightarrow\bit^\ell\}$ be a family of injective one-way functions and let $\NIZK=(\NIZK.\setup,\NIZK.\prove,\allowbreak\NIZK.\vrfy)$ be a NIZK system for $\NP$.
Our relaxed watermarking scheme for searchable compute-and-compare circuits $\Ccnc$ is as follows.
Its message space is $\bit^\msglen$ for some polynomial $\msglen$ of $\secp$.
Below, let $\cS$ be the search algorithm for $\Ccnc$.
\begin{description}
\item[$\gen(1^\secp)$:] Generate $\crs \gets \NIZK.\setup(1^\secp)$ and $\owf\gets\Fow$. Output $\pp \seteq (\crs,\owf)$. 
\item[$\Mark(\pp,\cnc{C}{\alpha},\msg)$:] Compute $x \seteq \cS(\cnc{C}{\alpha})$. That is, $x$ is an accepting point of $\cnc{C}{\alpha}$. Compute $y\gets\owf(\alpha)$. An NP relation $\cR_L$ is defined as follows.
\[
\cR_L \seteq \left\{\left((\msg,\owf,y,C),x)\right) \mid
y=\owf(C(x))
\right\}.
\]
Compute $\pi \gets \NIZK.\prove(\crs,(\msg,\owf,y,C),x)$.
Output $\tlC \seteq (\msg,y,C,\pi)$.
\item[$\Extract(\pp,\tlC')$:] Parse $\tlC' = (\msg^\prime,y^\prime,C^\prime,\pi^\prime)$ and output $\msg^\prime$.
\item[$\Eval(\pp,C^\prime,x)$:] Parse $\tlC' = (\msg^\prime,y^\prime,C^\prime,\pi^\prime)$ and run $\NIZK.\vrfy(\crs,(\msg^\prime,\owf,y^\prime,C^\prime),\pi^\prime)$. If the output is $\bot$, output $\bot$. Otherwise, output $1$ if $y'=\owf(C'(x))$ and $0$ otherwise.
\end{description}
\end{construction}

\begin{theorem}\label{thm:relaxed_watermarking_cnc}
Let $\inplen,\outlen$, and $\gamma$ be functions of $\secp$.
Also, let $\Dcnc{\gamma}$ be any distribution over $\Ccnc$ that has conditional min-entropy $\gamma$.
If $\Fow$ is a family of injective OWF for $\gamma$-sources and $\NIZK$ is a true-simulation extractable NIZK system for $\NP$ secure against adversaries of running time $O(2^\inplen)$, then \cref{const:relaxed_watermarking_cnc} is a relaxed $(1,\Dcnc{\gamma})$-secure watermarking scheme for $\Ccnc$.
\end{theorem}


\begin{proof}[Proof of~\cref{thm:relaxed_watermarking_cnc}]
The statistical correctness of \cref{const:relaxed_watermarking_cnc} follows from the completeness of $\NIZK$ and the injective property of $\Fow$.
Also, the extraction correctness of \cref{const:relaxed_watermarking_cnc} immediately follows from the construction.
Below, we prove the relaxed $(\epsilon,\Dcnc{\gamma})$-unremovability of \cref{const:relaxed_watermarking_cnc}.

Let $\qA$ be a QPT adversary attacking relaxed $(1,\Dcnc{\gamma})$-unremovability.
We prove this theorem using hybrid games.

\begin{description}
\item[Game $1$:]This is $\expa{\qA,\Dcnc{\gamma}}{r}{urmv}(\secp,\epsilon)$ for \cref{const:relaxed_watermarking_cnc}.

\begin{enumerate}
\item The challenger generates $\crs \gets \NIZK.\setup(1^\secp)$ and $\owf\gets\Fow$, and gives $\pp \seteq (\crs,\owf)$ to the adversary $\qA$.
\item At some point, $\qA$ queries a message $\msg \in \bit^\msglen$ to the challenger.
The challenger first samples $\cnc{C}{\alpha}\gets\Dcnc{\gamma}$.
Next, the challenger computes $x \seteq \cS(\cnc{C}{\alpha})$ and $y\gets\owf(\alpha)$.
Then, the challenger computes $\pi \gets \NIZK.\Prove(\crs,(\msg,\owf,y,C),x)$.
Then, the challenger returns $\tlC \seteq (\msg,y,C,\pi)$ to $\qA$.

\item Finally, $\qA$ outputs $\tlC^{\ast}=(\msg^*,y^*,C^*,\pi^*)$. If $\Eval(\pp,\tlC^{\ast},\cdot)$ and $\cnc{C}{\alpha}(\cdot)$ are functionally equivalent, and $\Extract (\pp, \tlC^{\ast})=\msg^* \ne \msg$, then the challenger outputs $1$ as the output of this game. Otherwise, the challenger outputs $0$ as the output of this game.
\end{enumerate}
\end{description}

We define the following three conditions.

\begin{itemize}
\item[$(a)$]  $\Eval(\pp,\tlC^{\ast},\cdot)$ and $\cnc{C}{\alpha}(\cdot)$ are functionally equivalent.
\item[$(b)$] $\NIZK.\vrfy(\crs,(m^*,\owf,y^*,C^*),\pi^*)=\top$.
\item[$(c)$] $\msg^*\neq \msg$.
\end{itemize}

It is clear that if all of the above conditions are satisfied, the output of Game $1$ is $1$.
In the opposite direction, it is clear that the conditions $(a)$ and $(c)$ are satisfied whenever the output of Game $1$ is $1$ from the definition of Game $1$.
Also, we see that if the condition $(b)$ is not satisfied, $\Eval(\pp,\tlC^{\ast},\cdot)$ and $\cnc{C}{\alpha}(\cdot)$ are not functionally equivalent and thus the output of Game $1$ is $0$.
Therefore, the condition $(b)$ is satisfied whenever the output of Game $1$ is $1$.
Overall, the output of Game $1$ is $1$ if and only if the above three conditions hold in Game $1$.


We next consider the following adversary $\qAnizk$ attacking the true-simulation extractability of $\NIZK$ using $\qA$.

\begin{enumerate}
\item Given $\crs$, $\qAnizk$ generates $\owf\gets\Fow$, and gives $\pp \seteq (\crs,\owf)$ to $\qA$.
\item When $\qA$ queries a message $\msg \in \bit^\msglen$, $\qAnizk$ first samples $\cnc{C}{\alpha}\gets\Dcnc{\gamma}$.
Next, $\qAnizk$ computes $x \seteq \cS(\cnc{C}{\alpha})$ and $y\gets\owf(\alpha)$.
Then, $\qAnizk$ sends a statement/witness pair $((\msg,\owf,y,C),x)$ to the challenger.

\item Given $\pi$, $\qAnizk$ sends $\tlC \seteq (\msg,y,C,\pi)$ to $\qA$.

\item When $\qA$ outputs $\tlC^{\ast}=(\msg^*,y^*,C^*,\pi^*)$, $\qAnizk$ first checks whether $\Eval(\pp,C^{\ast},\cdot)$ and $\cnc{C}{\alpha}(\cdot)$ are functionally equivalent.
(Note that this check can be done in time $O(2^n)$.)
If so, $\qAnizk$ outputs a statement/proof pair $((m^*,\owf,y^*,C^*),\pi^*)$. Otherwise, $\qAnizk$ outputs $\bot$.
\end{enumerate}
When we execute $\expt{\qAnizk,\NIZK}{se\textrm{-}real}$, the output of it is $1$ if and only if the following conditions hold.
\begin{itemize}
\item $\Eval(\pp,\tlC^{\ast},\cdot)$ and $\cnc{C}{\alpha}(\cdot)$ are functionally equivalent.
\item $\NIZK.\vrfy(\crs,(m^*,\owf,y^*,C^*),\pi^*)=\top$.
\item $((\msg,\owf,y,C),x)\in \cR_L$.
\item $(\msg,\owf,y,C)\neq(m^*,\owf,y^*,C^*)$.
\end{itemize}
$\qAnizk$ perfectly simulates Game $1$ for $\qA$ until $\qA$ terminates.
We see that when the output of the simulated Game $1$ is 1, the output of $\expt{\qAnizk,\NIZK}{se\textrm{-}real}$ is also $1$.
Namely, we have $\Pr[\textrm{Output of Game~} 1 \textrm{~is~} 1]\leq\Pr[1\gets\expt{\qAnizk,\NIZK}{se\textrm{-}real}]$.

$\qAnizk$ runs in time $O(2^n)$.
Since $\NIZK$ satisfies true-simulation extractability against adversaries runs in time $O(2^n)$, there exists $\Sim=(\fksetup,\Sim_1,\Sim_2)$ such that we have
\[
\abs{\Pr[1\gets\expt{\qAnizk,\NIZK}{se\textrm{-}real}]-\Pr[1\gets\expt{\qAnizk,\Sim,\NIZK}{se\textrm{-}sim}]}\leq\negl(\secp).
\]
We then define the following Game 2.

\begin{description}
\item[Game $2$:]This game is the same as $\expt{\qAnizk,\Sim,\NIZK}{se\textrm{-}sim}$ except conceptual changes. Especially, this game is obtained by transforming $\expt{\qAnizk,\Sim,\NIZK}{se\textrm{-}sim}$ into a security game played between the challenger and $\qA$ so that the output distribution does not change.
\begin{enumerate}
\item The challenger generates $(\crs,\td)\gets\fksetup(1^\secp)$ and $\owf\gets\Fow$, and gives $\pp \seteq (\crs,\owf)$ to $\qA$.
\item When $\qA$ queries a message $\msg \in \bit^\msglen$, the challenger first samples $\cnc{C}{\alpha}\gets\Dcnc{\gamma}$.
Next, the challenger computes $x \seteq \cS(\cnc{C}{\alpha})$ and $y\gets\owf(\alpha)$.
Then, the challenger computes $(\pi,\state_{\Sim})\gets\Sim_1(\crs,\td,(\msg,\owf,y,C))$ and sends $\tlC \seteq (\msg,y,C,\pi)$ to $\qA$.

\item When $\qA$ outputs $\tlC^{\ast}=(\msg^*,y^*,\prfkey^*,\pi^*)$, the challenger computes $x^*\gets\Sim_2(\state_{\Sim},(\msg^*,\owf,y^*,C^*),\pi^*)$. The challenger then outputs $1$ if all of the following conditions hold.
\begin{itemize}
\item $\Eval(\pp,C^{\ast},\cdot)$ and $\cnc{C}{\alpha}$ are functionally equivalent.
\item $\NIZK.\vrfy(\crs,(m^*,\owf,y^*,C^*),\pi^*)=\top$.
\item $((\msg,\owf,y,C),x)\in \cR_L$.
\item $((m^*,\owf,y^*,C^*),x^*)\in \cR_L$.
\item $(\msg,\owf,y,C)\neq(m^*,\owf,y^*,C^*)$.
\end{itemize}
Otherwise, the challenger outputs $0$.
\end{enumerate}
\end{description}

If the above first item and fourth item hold, we have
\[
\cnc{C}{\alpha}(x^*)=1 \Leftrightarrow C(x^*)=\alpha,
\]
and thus $\owf(C(x^*))=y$. Therefore, we have $\Pr[\textrm{Output of Game~} 2 \textrm{~is~} 1]\leq\negl(\secp)$ from the security of $\Fow$.

From the discussions so far, we obtain $\Pr[\textrm{Output of Game~} 1 \textrm{~is~} 1]\leq\negl(\secp)$.
This completes the proof.
\end{proof}

\medskip

From \cref{thm:inj_OWF_LWE} and \cref{thm:seNIZK_LWE}, we can instantiate \cref{const:relaxed_watermarking_cnc} under the LWE assumption.
Concretely, we obtain the following theorem.
\begin{theorem}\label{thm:relaxed_watermarking_cnc_LWE}
Let $\eta>0$ be any constant.
Assuming the hardness of the LWE problem against sub-exponential time quantum adversaries, there exists a relaxed $(1,\Dcnc{\secp^\eta})$-secure watermarking scheme for the class of compute-and-compare circuits $\Ccnc$, where $\Dcnc{\secp^\eta}$ is any distribution over $\Ccnc$ that has conditional min-entropy $\secp^\eta$.
\end{theorem}

\else
\fi


\section{Secure Software Leasing from Two-Tier Quantum Lightning}\label{sec:SSL_ttQL}

This section shows how to construct a finite-term secure SSL scheme from two-tier quantum lightning and a relaxed watermarking.
Due to a technical reason, we additionally use an OT-MAC, which can be realized information theoretically.

\begin{construction}[SSL from Two-Tier Quantum Lightning]\label{const:SSL}
Let $\cC=\setbracket{\cC_{\secp}}_{\secp}$ be a circuit class such that $\cC_{\secp}$ contains circuit of input length is $n$ and output length $m$.
Our SSL scheme $(\setup,\sslgen,\lessor,\run,\sslcheck)$ for $\cC$ is based on a two-tier quantum lightning $\ttQL=(\ttQL.\setup, \boltgen,\semivrfy,\fullvrfy)$, a relaxed watermarking scheme $\WM=(\WM.\gen,\WM.\Mark,\WM.\Extract,\WM.\Eval)$ for $\cC$, and a OT-MAC $\MAC=(\Macgen,\Mactag,\Macvrfy)$.
\begin{itemize}
\item $\setup(1^\secp)$: Compute $\pp \gets \WM.\gen(1^{\secp})$ and output $\crs \seteq \pp$.
\item $\sslgen(\crs)$: Parse $\pp\gets\crs$. Compute $(\pk,\sk)\gets \ttQL.\setup(1^\secp)$ and $\mackey\gets\Macgen(1^\secp)$, and set $\ssl.\sk \seteq (\pp,\pk,\sk,\mackey)$.
\item $\lessor(\ssl.\sk,C)$: Do the following:
\begin{enumerate}
\item Parse $(\pp,\pk,\sk,\mackey)\gets\ssl.\sk$.
\item Compute $(\snum,\bolt)\gets \boltgen(\pk)$.
\item Compute $\tlC \gets \WM.\Mark(\pp,C,\pk\concat \snum)$.
\item Compute $\mac\gets\Mactag(\mackey,\snum)$.
\item Output $\sft_C \seteq (\bolt,\tlC,\mac)$.
\end{enumerate}
\item $\run(\crs,\sft_C,x)$: Do the following.
\begin{enumerate}
	\item Parse $\pp\gets\crs$ and $\sft_C = (\bolt,\tlC,\mac)$.
	\item Compute $\pk^\prime\concat \snum^\prime \gets \WM.\Extract(\pp,\tlC)$.
	\item Run $\semivrfy(\pk',\snum',\bolt)$ and obtain $(b,\bolt')$. If $b = \bot$, then output $\bot$. Otherwise, do the next step.
	\item Compute $y \gets \WM.\Eval(\pp,\tlC,x)$.
	\item Output $(\bolt',\tlC,\mac)$ and $y$.
\end{enumerate}
\item $\sslcheck(\ssl.\sk,\sft_C)$: Do the following.
\begin{enumerate}
\item Parse $(\pp,\pk,\sk,\mackey)\gets\ssl.\sk$ and $\sft_C = (\bolt,\tlC,\mac)$.
\item Compute $\pk^\prime\concat \snum^\prime \gets \WM.\Extract(\pp,\tlC)$.
\item If $\Macvrfy(\mackey,\snum^\prime,\mac)=\bot$, then output $\bot$. Otherwise, do the next step.
\item Output $d \gets \fullvrfy(\sk,\snum',\bolt)$.
\end{enumerate}
\end{itemize}
\end{construction}

We have the following theorems.

\begin{theorem}\label{thm:SSL_average}
Let $\epsilon$ be any inverse polynomial of $\secp$ and $\cD_{\cC}$ a distribution over $\cC$.
Assume $\ttQL$ is a two-tier quantum lightning scheme, $\WM$ is a $(\epsilon,\cD_{\cC})$-secure relaxed watermarking scheme for $\cC$, and $\MAC$ is an OT-MAC.
Then, \cref{const:SSL} is a $(\epsilon,\cD_{\cC})$-average-case finite-term lessor secure SSL scheme for $\cC$.
\end{theorem}

\begin{theorem}\label{thm:SSL_perfect}
Let $\beta$ be any inverse polynomial of $\secp$ and $\cD_{\cC}$ a distribution over $\cC$.
Assume $\ttQL$ is a two-tier quantum lightning scheme, $\WM$ is a $(1,\cD_{\cC})$-secure relaxed watermarking scheme for $\cC$, and $\MAC$ is an OT-MAC.
Then, \cref{const:SSL} is a $(\beta,\cD_{\cC})$-perfect finite-term lessor secure SSL scheme for $\cC$.
\end{theorem}

Since the proofs for the above two theorems are almost the same, we only provide the proof of \cref{thm:SSL_average} and omit the proof for \cref{thm:SSL_perfect}.


\begin{proof}[Proof of~\cref{thm:SSL_average}]
The correctness of $\run$ of \cref{const:SSL} follows from the statistical correctness and extraction correctness of $\WM$, and the semi-verification correctness of $\ttQL$.
Also, the correctness of $\sslcheck$ of \cref{const:SSL} follows from the extraction correctness of $\WM$, the correctness of $\MAC$, and the full-verification correctness of $\ttQL$.
Below, we prove the $(\epsilon,\cD_\cC)$-average-case finite-term lessor security of \cref{const:SSL}.

Let $\qA$ be a QPT adversary attacking $(\epsilon,\cD_\cC)$-average-case finite-term lessor security.
The detailed description of $\expt{\qA,\cD_\cC}{aft\textrm{-}lessor}(\secp,\epsilon)$ is as follows.
\begin{enumerate}
\item The challenger generates $\pp \gets \WM.\Gen(1^\secp)$, $(\pk,\sk)\gets\ttQL.\Setup(1^\secp)$, and $\mackey\gets\MAC.\Gen(1^\secp)$.
The challenger then generate $C\gets\cD_\cC$ and $(\snum,\bolt)\gets\boltgen(\pk)$.
The challenger also computes $\tlC\gets\WM.\Mark(\pp,C,\pk\|\snum)$ and $\mac\gets\Mactag(\mackey,\snum)$.
The challenger finally sends $\crs:=\pp$ and $\sft_C \seteq (\bolt,\tlC,\mac)$ to $\qA$.
Below, let $\ssl.\sk:=(\pp,\pk,\sk,\mackey)$.

\item $\qA$ outputs $(\tlC^{(1)},\mac^{(1)},\tlC^{(2)},\mac^{(2)},\qb^*)$.
$(\tlC^{(1)},\mac^{(1)})$ is the classical part of the first copy, and $(\tlC^{(2)},\mac^{(2)})$ is that of the second copy.
Moreover, $\qb^*$ is a density matrix associated with two registers $\reg_1$ and $\reg_2$, where the states in $\reg_1$ and $\reg_2$ are associated with the first and second copy, respectively.
Below, let $\sft^\first=(\Trace_2[\qb^*],\tlC^\first,\mac^\first)$ and $\sft^\second=(P_2(\ssl.\sk,\qb^*),\tlC^\second,\mac^\second)$.
Recall that $P_2(\ssl.\sk,\qb^*)$ denotes the resulting post-measurement state on $\reg_2$ after the check on $\reg_1$.

\item If it holds that $\sslcheck(\ssl.\sk,\sft^\first)=\top$ and $\Pr[\runout(\crs,\sft^\second,x) = C(x)]\geq\epsilon$, where the probability is taken over the choice of $x\la\zo{n}$ and the random coin of $\run$, then the challenger outputs $1$ as the output of this game. Otherwise, the challenger outputs $0$ as the output of this game.
\end{enumerate}

Below, we let $\pk^{(1)}\|\snum^\first\gets\WM.\Extract(\pp,\tlC^\first)$ and $\pk^\second\|\snum^\second\gets\WM.\Extract(\pp,\tlC^\second)$.
The output of $\expt{\qA,\cD_\cC}{aft\textrm{-}lessor}(\secp,\epsilon)$ is $1$ if and only if the following conditions hold.

\begin{itemize}
\item[(a)] $\Macvrfy(\mackey,\snum^\first,\mac^\first)=\top$.
\item[(b)] $\fullvrfy(\sk,\snum^\first,\Trace_2[\qb^*])=\top$.
\item[(c)] $\semivrfy(\pk^\second,\snum^\second,P_2(\ssl.\sk,\qb^*))=\top$.
\item[(d)] $\Pr_{x\la\zo{n}}[\WM.\Eval(\crs,\tlC^{(2)},x) = C(x)]\geq\epsilon$.
\end{itemize}

We can estimate the advantage of $\qA$ as
\ifnum\submission=0
\begin{align*}
\Pr[\expt{\qA,\cD_\cC}{aft\textrm{-}lessor}(\secp,\epsilon)\out 1]
&=\Pr[\expt{\qA,\cD_\cC}{aft\textrm{-}lessor}(\secp,\epsilon)\out 1\land\snum^\first=\snum\land\pk^\second\|\snum^\second=\pk\|\snum]\\
&~~~~+\Pr[\expt{\qA,\cD_\cC}{aft\textrm{-}lessor}(\secp,\epsilon)\out 1\land(\snum^\first\neq\snum\lor\pk^\second\|\snum^\second\neq\pk\|\snum)]\\
&\leq\Pr[\expt{\qA,\cD_\cC}{aft\textrm{-}lessor}(\secp,\epsilon)\out 1\land\snum^\first=\snum\land\pk^\second\|\snum^\second=\pk\|\snum]\\
&~~~~+\Pr[\expt{\qA,\cD_\cC}{aft\textrm{-}lessor}(\secp,\epsilon)\out 1\land\snum^\first\neq\snum]\\
&~~~~+\Pr[\expt{\qA,\cD_\cC}{aft\textrm{-}lessor}(\secp,\epsilon)\out 1\land\pk^\second\|\snum^\second\neq\pk\|\snum]
\end{align*}
\else
\begin{align*}
&\Pr[\expt{\qA,\cD_\cC}{aft\textrm{-}lessor}(\secp,\epsilon)\out 1]\\
&=\Pr[\expt{\qA,\cD_\cC}{aft\textrm{-}lessor}(\secp,\epsilon)\out 1\land\snum^\first=\snum\land\pk^\second\|\snum^\second=\pk\|\snum]\\
&~~~~+\Pr[\expt{\qA,\cD_\cC}{aft\textrm{-}lessor}(\secp,\epsilon)\out 1\land(\snum^\first\neq\snum\lor\pk^\second\|\snum^\second\neq\pk\|\snum)]\\
&\leq\Pr[\expt{\qA,\cD_\cC}{aft\textrm{-}lessor}(\secp,\epsilon)\out 1\land\snum^\first=\snum\land\pk^\second\|\snum^\second=\pk\|\snum]\\
&~~~~+\Pr[\expt{\qA,\cD_\cC}{aft\textrm{-}lessor}(\secp,\epsilon)\out 1\land\snum^\first\neq\snum]\\
&~~~~+\Pr[\expt{\qA,\cD_\cC}{aft\textrm{-}lessor}(\secp,\epsilon)\out 1\land\pk^\second\|\snum^\second\neq\pk\|\snum]
\end{align*}
\fi

We then have the following lemmas.
\begin{lemma}\label{lem:SSL_ttQL}
$\Pr[\expt{\qA,\cD_\cC}{aft\textrm{-}lessor}(\secp,\epsilon)\out 1\land\snum^\first=\snum\land\pk^\second\|\snum^\second=\pk\|\snum]\allowbreak =\negl(\secp)$ by the two-tier unclonability of $\ttQL$.
\end{lemma}
\begin{lemma}\label{lem:SSL_MAC}
$\Pr[\expt{\qA,\cD_\cC}{aft\textrm{-}lessor}(\secp,\epsilon)\out 1\land\snum^\first\neq\snum]=\negl(\secp)$ by the security of $\MAC$.
\end{lemma}
\begin{lemma}\label{lem:SSL_WM}
$\Pr[\expt{\qA,\cD_\cC}{aft\textrm{-}lessor}(\secp,\epsilon)\out 1\land\pk^\second\|\snum^\second\neq\pk\|\snum]=\negl(\secp)$ by the $(\epsilon,\cD_\cC)$-removability of $\WM$.
\end{lemma}

For \cref{lem:SSL_ttQL}, if the condition (b) and (c) above and $\snum^\first=\snum\land\pk^\second\|\snum^\second=\pk\|\snum$ hold at the same time with non-negligible probability, by using $\qA$, we can construct an adversary breaking the two-tier unclonability of $\ttQL$. Thus, we have \cref{lem:SSL_ttQL}.
Next, for \cref{lem:SSL_MAC}, if the condition (a) and $\snum^\first\neq\snum$ hold with non-negligible probability, also by using $\qA$, we can construct an adversary breaking the security of $\MAC$. Thus, we have \cref{lem:SSL_MAC}.
Finally, for \cref{lem:SSL_WM}, if the condition (d) and $\pk^\second\|\snum^\second\neq\pk\|\snum$ hold with non-negligible probability, by using $\qA$, we can construct an adversary breaking $(\epsilon,\cD_\cC)$-unremovability of $\WM$.
Thus, we have \cref{lem:SSL_WM}.

From the discussions so far, we obtain $\Pr[\expt{\qA,\cD_\cC}{aft\textrm{-}lessor}(\secp,\epsilon)\out 1]\leq\negl(\lambda)$.
This completes the proof.
\end{proof}

\ifnum\cameraready=1
\subsection*{Secure Software Leasing with Classical Communication}\label{sec:LNCS_only_ccSSL}
It is not difficult to extend the definition of SSL to that of SSL with classical communication.
By using two-tier QL with classical verification in~\cref{sec:ttQL_Classical_Vrfy} instead of two-tier QL, it is easy to extend the scheme in~\cref{sec:SSL_ttQL} to an SSL scheme with classical communication thanks to the bolt-to-certificate capability. Thus, we can achieve SSL with classical communication from the LWE assumption. Due to space limitations, we omit the details of the definition and construction. See the full version~\cite{EPRINT:KitNisYam20}.
\else
\fi

\ifnum\submission=0

\section{Secure Software Leasing with Classical Communication}\label{sec:ccSSL_csvttQL}

In this section, we extend our finite-term secure SSL scheme to one with classical communication by using two-tier quantum lightning with classical verification.

\subsection{Definition}\label{sec:ccSSL}
First, we formalize the notion of SSL with classical communication.

\begin{definition}[SSL with Setup and classical communication]\label{def:ssl_cc_crs}
Let $\cC=\setbracket{\cC_{\secp}}_{\secp}$ be a circuit class such that $\cC_{\secp}$ contains circuit of input length is $n$ and output length $m$.
A secure software lease scheme with setup and classical communication for $\cC$ is a tuple of algorithms $(\setup,\sslgen,\clessor,\lessee_1,\lessee_2,\run,\sslcert,\allowbreak\certvrfy)$.
\begin{itemize}
\item $\setup(1^\secp), \sslgen(\crs),\run(\crs,\sft_C,x)$: These are the same as the SSL in~\cref{def:ssl_crs}.
\item $\sslgen(\crs):$ The key generation algorithm takes as input $\crs$ and outputs a public key $\pk$ and secret key $\sk$.
\item $\lessee_1(\pk):$ The first stage lessee algorithm takes as input $\crs$ and outputs a classical string $\obligation$ and a quantum state $\statelessee$.
\item $\clessor(\sk,\obligation,C):$ The lessor algorithm takes as input $\sk$, $\obligation$, and a circuit $C\in\cC$, and outputs a classical string $\answer$.
\item $\lessee_2(\pk,\statelessee,\answer):$ The second stage lessee algorithm takes as input $\crs$, $\statelessee$, and $\answer$, and outputs a quantum state $\sft$.
\item $\sslcert(\crs,\sft^*):$ The certification algorithm takes as input $\crs$ and $\sft^*$ and outputs a classical string $\cert$.
\item $\certvrfy(\sk,\cert):$ The certification-verification algorithm takes as input $\sk$ and $\cert$ and outputs $\top$ or $\bot$.

\end{itemize}
\end{definition}

\begin{definition}[Correctness for SSL with classical verification]\label{def:correctness_SSL_cc}
An SSL scheme with classical communication $(\setup,\sslgen,\clessor,\allowbreak\lessee_1,\lessee_2,\run,\sslcheck)$ for $\cC = \setbracket{\cC_\secp}_\secp$ is correct if for all $C\in \cC_\secp$, the following two properties hold:
\begin{itemize}
\item Correctness of $\run$:
\ifnum\submission=0
\begin{align*}
&\Pr\left[ \forall x, \Pr[\runout(\crs,\sft_C,x)=C(x)]\ge 1-\negl(\secp) \ \middle |
\begin{array}{rl}
& \crs \gets \setup(1^\secp) \\
&(\pk,\sk) \gets \sslgen(\crs) \\
&(\obligation,\statelessee)\gets\lessee_1(\pk)\\
&\answer\gets\clessor(\sk,\obligation,C)\\
&\sft_C \gets\lessee_2(\pk,\statelessee,\answer)
\end{array}
\right]\\
& \ge  1-\negl(\secp).
\end{align*}
\else
\begin{align*}
&\Pr\left[ 
\begin{array}{ll}
\forall x, &\Pr[\runout(\crs,\sft_C,x)=C(x)]\\
 &\ge 1-\negl(\secp)
\end{array}
~\middle |
\begin{array}{rl}
& \crs \gets \setup(1^\secp) \\
&(\pk,\sk) \gets \sslgen(\crs) \\
&(\obligation,\statelessee)\gets\lessee_1(\pk)\\
&\answer\gets\clessor(\sk,\obligation,C)\\
&\sft_C \gets\lessee_2(\pk,\statelessee,\answer)
\end{array}
\right]\\
&\ge 1-\negl(\secp).
\end{align*}
\fi
\item Correctness of $\certvrfy$:
\[
\Pr\left[ \certvrfy(\sk,\cert) = \top \ \middle |
\begin{array}{rl}
&\crs \gets \setup(1^\secp) \\
&(\pk,\sk) \gets \sslgen(\crs) \\
&(\obligation,\statelessee)\gets\lessee_1(\pk)\\
&\answer\gets\clessor(\sk,\obligation,C)\\
&\sft_C \gets\lessee_2(\pk,\statelessee,\answer)\\
&\cert\gets\sslcert(\crs,\sft_C)\\
\end{array}
\right]\ge 1-\negl(\secp).
\]
\end{itemize}
\end{definition}

Similarly to the ordinary SSL, we consider the following two security notions perfect finite-term lessor security and average-case finite-term lessor security.

\begin{definition}[Perfect Finite-Term Lessor Security]\label{def:cc_perfect_lessor_security}
Let $\beta$ be any inverse polynomial of $\lambda$ and $\cD_{\cC}$ a distribution on $\cC$.
We define the $(\beta,\cD_{\cC})$-perfect finite-term lessor security game $\expt{\qA,\cD_\cC}{pft\textrm{-}lessor\textrm{-}cc}(\secp,\beta)$ between the challenger and adversary $\qA$ as follows.
\begin{enumerate}
\item The challenger generates $\crs \gets \setup(1^\secp)$ and $(\pk,\sk)\gets \sslgen(\crs)$, and sends $\crs$ and $\pk$ to $\qA$.
\item $\qA$ outputs $\obligation$. The challenger generates $C \gets \cD_{\cC}$, computes $\answer \gets \clessor(\sk,\obligation,C)$, and sends $\answer$ to $\qA$.
\item $\qA$ outputs a classical string $\cert^*$ and a quantum state $\sft^\ast$.
\item If $\certvrfy(\sk,\cert^*) = \top$ and $\forall x\ \Pr[\runout(\crs,\sft^*,x) = C(x)]\ge \beta$ hold, where the probability is taken over the choice of the random coin of $\run$, then the challenger outputs $1$. Otherwise, the challenger outputs $0$.
\end{enumerate}

We say that an SSL scheme with classical communication $(\setup,\sslgen,\clessor,\allowbreak\lessee_1,\allowbreak\lessee_2,\run,\allowbreak\sslcert,\certvrfy)$ is $(\beta,\cD_{\cC})$-perfect finite-term lessor secure,
if for any QPT $\qA$, the following holds.
\begin{align*}
\Pr[\expt{\qA,\cD_\cC}{pft\textrm{-}lessor\textrm{-}cc}(\secp,\beta) \out 1] \le \negl(\secp).
\end{align*}
\end{definition}

\begin{definition}[Average-Case Finite-Term Lessor Security]\label{def:cc_average_lessor_security}
Let $\epsilon$ be any inverse polynomial of $\lambda$ and $\cD_{\cC}$ a distribution on $\cC$.
We define the $(\epsilon,\cD_{\cC})$-average-case finite-term lessor security game $\expt{\qA,\cD_\cC}{aft\textrm{-}lessor\textrm{-}cc}(\secp,\epsilon)$ by replacing the fourth stage of $\expt{\qA,\cD_\cC}{pft\textrm{-}lessor\textrm{-}cc}(\secp,\beta)$ with the following.
\begin{enumerate}
\setcounter{enumi}{3}
\item If $\certvrfy(\sk,\cert^*) = \top$ and $\Pr[\runout(\crs,\sft^*,x) = C(x)]\ge \epsilon$ hold, where the probability is taken over the choice of $x\la\zo{n}$ and the random coin of $\run$, then the challenger outputs $1$. Otherwise, the challenger outputs $0$.
\end{enumerate}

We say that an SSL scheme with classical communication $(\setup,\sslgen,\clessor,\allowbreak\lessee_1,\allowbreak\lessee_2,\run,\allowbreak\sslcert,\certvrfy)$ is $(\epsilon,\cD_\cC)$-average-case finite-term lessor secure,
if for any QPT $\qA$, the following holds.
\begin{align*}
\Pr[\expt{\qA,\cD_\cC}{aft\textrm{-}lessor}(\secp,\epsilon) \out 1] \le \negl(\secp).
\end{align*}
\end{definition}

\subsection{Construction}\label{sec:construction_SSLcc}
We show how to construct a finite-term secure SSL scheme with classical communication from two-tier quantum lightning with classical verification, relaxed watermarking, and OT-MAC.

\begin{construction}[SSL from Two-Tier QL with classical verification]\label{const:SSLcc}
Let $\cC=\setbracket{\cC_{\secp}}_{\secp}$ be a circuit class such that $\cC_{\secp}$ contains circuit of input length is $n$ and output length $m$.
Our SSL scheme with classical communication $(\setup,\sslgen,\clessor,\allowbreak\lessee_1,\lessee_2,\run,\sslcert,\certvrfy)$ for $\cC$ is based on a two-tier QL with semi-classical verification $\ttQL=(\ttQL.\setup, \boltgen,\boltcert,\semivrfy,\ttQL.\certvrfy)$, relaxed watermarking scheme $\WM=(\WM.\setup,\WM.\Mark,\allowbreak\WM.\Extract,\WM.\Eval)$ for $\cC$, and OT-MAC $\MAC=(\Macgen,\allowbreak\Mactag,\Macvrfy)$.
\begin{itemize}
\item $\setup(1^\secp)$: Compute $\pp \gets \WM.\gen(1^{\secp})$ and output $\crs\seteq \pp$.
\item $\sslgen(\crs)$: Compute $(\pk,\sk)\gets \ttQL.\setup(1^\secp)$ and $\mackey\gets\Macgen(1^\secp)$, and output $\ssl.\pk\seteq \pk$ and $\ssl.\sk \seteq (\pp,\pk,\sk,\mackey)$.
\item $\lessee_1(\ssl.\pk)$:  Parse $\pk\gets\ssl.\pk$, generate $(\snum,\bolt)\gets\boltgen(\pk)$, and outputs $\obligation:=\snum$ and $\statelessee:=\bolt$.
\item $\clessor(\sk,\obligation,C)$:
\begin{enumerate}
\item Parse $(\pp,\pk,\sk,\mackey)\gets\ssl.\sk$ and $\snum\gets\obligation$.
\item Compute $\tlC \gets \WM.\Mark(\pp,C,\pk\concat \snum)$.
\item Compute $\mac\gets\Mactag(\mackey,\snum)$.
\item Output $\answer \seteq (\tlC,\mac)$.
\end{enumerate}
\item $\lessee_2(\ssl.\pk,\statelessee,\answer)$: Parse $\bolt\gets\statelessee$ and $(\tlC,\mac)\gets\answer$, and output $\sft:=(\bolt,\tlC,\mac)$.
\item $\run(\crs,\csft_C,x)$: Do the following.
\begin{enumerate}
	\item Parse $\pp\gets\crs$ and $(\bolt,\tlC,\mac)\gets\sft$.
	\item Compute $\pk^\prime\concat \snum^\prime \gets \WM.\Extract(\pp,\tlC)$.
	\item Run $(b,\bolt') \gets \semivrfy(\pk',\snum',\bolt)$. If $b=\bot$, then output $\bot$. Otherwise, do the next step.
	\item Compute $y \gets \WM.\Eval(\pp,\tlC,x)$.
	\item Output $(\bolt',\tlC,\mac)$ and $y$.
\end{enumerate}
\item $\sslcert(\crs,\sft)$: Parse $(\bolt,\tlC,\mac)\gets\sft$, runs $\ttQL.\cert \gets \boltcert(\bolt)$, and output $\cert \seteq (\ttQL.\cert,\tlC,\mac)$.
\item $\certvrfy(\ssl.\sk,\cert)$: Do the following.
\begin{enumerate}
\item Parse $(\pp,\pk,\sk,\mackey)\gets\ssl.\sk$ and $(\ttQL.\cert,\tlC,\mac)\gets\cert$.
\item Compute $\pk^\prime\concat \snum^\prime \gets \WM.\Extract(\pp,\tlC)$.
\item If $\Macvrfy(\mackey,\snum',\mac)=\bot$, then output $\bot$. Otherwise, do the next step.
\item Output $d \gets \ttQL.\certvrfy(\sk,\snum',\cert)$.
\end{enumerate}
\end{itemize}
\end{construction}

We have the following theorems.

\begin{theorem}\label{thm:SSLcc_average}
Let $\epsilon$ be any inverse polynomial of $\secp$ and $\cD_{\cC}$ a distribution over $\cC$.
Assume $\ttQL$ is a two-tier quantum lightning scheme with classical verification, $\WM$ is a $(\epsilon,\cD_{\cC})$-secure relaxed watermarking scheme for $\cC$, and $\MAC$ is an OT-MAC.
Then, \cref{const:SSLcc} is a $(\epsilon,\cD_{\cC})$-average-case finite-term lessor secure SSL scheme with classical communication for $\cC$.
\end{theorem}

\begin{theorem}\label{thm:SSLcc_perfect}
Let $\beta$ be any inverse polynomial of $\secp$ and $\cD_{\cC}$ a distribution over $\cC$.
Assume $\ttQL$ is a two-tier quantum lightning scheme with classical verification, $\WM$ is a $(1,\cD_{\cC})$-secure relaxed watermarking scheme for $\cC$, and $\MAC$ is an OT-MAC.
Then, \cref{const:SSLcc} is a $(\beta,\cD_{\cC})$-perfect finite-term lessor secure SSL scheme with classical communication for $\cC$.
\end{theorem}

Since the proofs for the above two theorems are almost the same, we provide the proof of only \cref{thm:SSLcc_average} and omit the proof of \cref{thm:SSLcc_perfect}.


\begin{proof}[Proof of~\cref{thm:SSLcc_average}]
The correctness of $\run$ of \cref{const:SSLcc} follows from the statistical correctness and extraction correctness of $\WM$, and the semi-verification correctness of $\ttQL$.
Also, the correctness of $\certvrfy$ of \cref{const:SSLcc} follows from the extraction correctness of $\WM$, the correctness of $\MAC$, and the certification-verification correctness of $\ttQL$.
Below, we prove the $(\epsilon,\cD_\cC)$-average-case finite-term lessor security of \cref{const:SSLcc}.

Let $\qA$ be a QPT adversary attacking $(\epsilon,\cD_\cC)$-average-case finite-term lessor security.
The detailed description of $\expt{\qA,\cD_\cC}{aft\textrm{-}lessor\textrm{-}cc}(\secp,\epsilon)$ is as follows.
\begin{enumerate}
\item The challenger generates $\pp \gets \WM.\Gen(1^\secp)$, $(\pk,\sk)\gets\ttQL.\Setup(1^\secp)$, and $\mackey\gets\MAC.\Gen(1^\secp)$.
The challenger sends $\crs:=\pp$ and $\ssl.\pk:=\pk$ to $\qA$.
Below, let $\ssl.\sk:=(\pp,\pk,\sk,\mackey)$.

\item $\qA$ sends $\obligation:=\snum^*$ to the challenger.
The challenger generates $C\gets\cD_\cC$.
The challenger also computes $\tlC\gets\WM.\Mark(\pp,C,\pk\|\snum^*)$ and $\mac\gets\Mactag(\mackey,\snum^*)$.
The challenger finally sends $\answer:=(\tlC,\mac)$ to $\qA$.

\item $\qA$ outputs $\cert^*=(\ttQL.\cert^*,\tlC^{(1)},\mac^{(1)})$ and $\sft^*=(\qb^*,\tlC^{(2)},\mac^{(2)})$, where $\qb^*$ is a single quantum state and others are classical strings.

\item If it holds that $\certvrfy(\ssl.\sk,\cert)=\top$ and $\Pr[\runout(\crs,\sft^*,x) = C(x)]\ge\epsilon$, where the probability is taken over the choice of $x\la\zo{n}$ and the random coin of $\run$, then the challenger outputs $1$ as the output of this game. Otherwise, the challenger outputs $0$ as the output of this game.
\end{enumerate}

Below, we let $\pk^{(1)}\|\snum^\first\gets\WM.\Extract(\pp,\tlC^\first)$ and $\pk^\second\|\snum^\second\gets\WM.\Extract(\pp,\tlC^\second)$.
The output of $\expt{\qA,\cD_\cC}{aft\textrm{-}lessor}(\secp,\epsilon)$ is $1$ if and only if the following conditions hold.

\begin{itemize}
\item[(a)] $\Macvrfy(\mackey,\snum^\first,\mac^\first)=\top$.
\item[(b)] $\ttQL.\certvrfy(\sk,\snum^\first,\ttQL.\cert^*)=\top$.
\item[(c)] $\semivrfy(\pk^\second,\snum^\second,\qb^*)=\top$.
\item[(d)] $\Pr_{x\la\zo{n}}[\WM.\Eval(\pp,\tlC^\second,x)= C(x)]\ge\epsilon$.
\end{itemize}

We can estimate the advantage of $\qA$ as
\ifnum\submission=0
\begin{align*}
\Pr[\expt{\qA,\cD_\cC}{aft\textrm{-}lessor}(\secp,\epsilon)\out 1]
&=\Pr[\expt{\qA,\cD_\cC}{aft\textrm{-}lessor}(\secp,\epsilon)\out 1\land\snum^\first=\snum^*\land\pk^\second\|\snum^\second=\pk\|\snum^*]\\
&~~~~+\Pr[\expt{\qA,\cD_\cC}{aft\textrm{-}lessor}(\secp,\epsilon)\out 1\land(\snum^\first\neq\snum^*\lor\pk^\second\|\snum^\second\neq\pk\|\snum^*)]\\
&\leq\Pr[\expt{\qA,\cD_\cC}{aft\textrm{-}lessor}(\secp,\epsilon)\out 1\land\snum^\first=\snum^*\land\pk^\second\|\snum^\second=\pk\|\snum^*]\\
&~~~~+\Pr[\expt{\qA,\cD_\cC}{aft\textrm{-}lessor}(\secp,\epsilon)\out 1\land\snum^\first\neq\snum^*]\\
&~~~~+\Pr[\expt{\qA,\cD_\cC}{aft\textrm{-}lessor}(\secp,\epsilon)\out 1\land\pk^\second\|\snum^\second\neq\pk\|\snum^*]
\end{align*}
\else
\begin{align*}
&\Pr[\expt{\qA,\cD_\cC}{aft\textrm{-}lessor}(\secp,\epsilon)\out 1]\\
&=\Pr[\expt{\qA,\cD_\cC}{aft\textrm{-}lessor}(\secp,\epsilon)\out 1\land\snum^\first=\snum^*\land\pk^\second\|\snum^\second=\pk\|\snum^*]\\
&~~~~+\Pr[\expt{\qA,\cD_\cC}{aft\textrm{-}lessor}(\secp,\epsilon)\out 1\land(\snum^\first\neq\snum^*\lor\pk^\second\|\snum^\second\neq\pk\|\snum^*)]\\
&\leq\Pr[\expt{\qA,\cD_\cC}{aft\textrm{-}lessor}(\secp,\epsilon)\out 1\land\snum^\first=\snum^*\land\pk^\second\|\snum^\second=\pk\|\snum^*]\\
&~~~~+\Pr[\expt{\qA,\cD_\cC}{aft\textrm{-}lessor}(\secp,\epsilon)\out 1\land\snum^\first\neq\snum^*]\\
&~~~~+\Pr[\expt{\qA,\cD_\cC}{aft\textrm{-}lessor}(\secp,\epsilon)\out 1\land\pk^\second\|\snum^\second\neq\pk\|\snum^*]
\end{align*}
\fi

We then have the following lemmas.
\begin{lemma}\label{lem:SSLcc_ttQL}
$\Pr[\expt{\qA,\cD_\cC}{aft\textrm{-}lessor}(\secp,\epsilon)\out 1\land\snum^\first=\snum^*\land\pk^\second\|\snum^\second=\pk\|\snum^*]=\negl(\secp)$ by the two-tier unclonability with classical verification of $\ttQL$.
\end{lemma}
\begin{lemma}\label{lem:SSLcc_MAC}
$\Pr[\expt{\qA,\cD_\cC}{aft\textrm{-}lessor}(\secp,\epsilon)\out 1\land\snum^\first\neq\snum^*]=\negl(\secp)$ by the security of $\MAC$.
\end{lemma}
\begin{lemma}\label{lem:SSLcc_WM}
$\Pr[\expt{\qA,\cD_\cC}{aft\textrm{-}lessor}(\secp,\epsilon)\out 1\land\pk^\second\|\snum^\second\neq\pk\|\snum^*]=\negl(\secp)$ by the $(\epsilon,\cD_\cC)$-unremovability of $\WM$.
\end{lemma}

For \cref{lem:SSLcc_ttQL}, if the condition (b) and (c) above and $\snum^\first=\snum^*\land\pk^\second\|\snum^\second=\pk\|\snum^*$ hold at the same time with non-negligible probability, by using $\qA$, we can construct an adversary breaking the two-tier unclonability of $\ttQL$. Thus, we have \cref{lem:SSLcc_ttQL}.
Next, for \cref{lem:SSLcc_MAC}, if the condition (a) and $\snum^\first\neq\snum^*$ hold with non-negligible probability, also by using $\qA$, we can construct an adversary breaking the security of $\MAC$. Thus, we have \cref{lem:SSLcc_MAC}.
Finally, for \cref{lem:SSLcc_WM}, if the condition (d) and $\pk^\second\|\snum^\second\neq\pk\|\snum^*$ hold with non-negligible probability, by using $\qA$, we can construct an adversary breaking $(\epsilon,\cD_\cC)$-unremovability of $\WM$.
Thus, we have \cref{lem:SSLcc_WM}.

From the discussions so far, we obtain $\Pr[\expt{\qA,\cD_\cC}{aft\textrm{-}lessor}(\secp,\epsilon)\out 1]\leq\negl(\lambda)$.
This completes the proof.
\end{proof}
%

\else
\fi
\ifnum\submission=0

\section{Putting It Altogether: SSL from LWE}\label{sec:SSL_LWE}

In this section, we summarize our results.

\paragraph{SSL for a family of PRF.}
By combining \cref{thm:SSL_average} with \cref{thm:ttQLcv_implies_ttQL}, \cref{thm:ttQLcv_LWE}, \cref{thm:relaxed_watermarking_prf_LWE}, and \cref{thm:MAC_information_theoretic}, we obtain the following theorem.
\begin{theorem}\label{thm:SSL_LWE_prf}
Let $\epsilon$ be any inverse polynomial of $\secp$.
Assuming the quantum hardness of the LWE problem, there exists a $(\epsilon,\cU_\prf)$-average-case finite-term lessor secure SSL scheme for a family of PRF $\cF$, where $\cU_\prf$ is the uniform distribution over $\cF$.
\end{theorem}
Also, by combing \cref{thm:SSLcc_average} with \cref{thm:ttQLcv_LWE}, \cref{thm:relaxed_watermarking_prf_LWE}, and \cref{thm:MAC_information_theoretic}, we obtain the following theorem.
\begin{theorem}\label{thm:SSLcc_LWE_prf}
Let $\epsilon$ be any inverse polynomial of $\secp$.
Assuming the quantum hardness of the LWE problem, there exists a $(\epsilon,\cU_\prf)$-average-case finite-term lessor secure SSL scheme with classical communication for a family of PRF $\cF$, where $\cU_\prf$ is the uniform distribution over $\cF$.
\end{theorem}

\paragraph{SSL for compute-and-compare circuits.}
By combining \cref{thm:SSL_perfect} with \cref{thm:ttQLcv_implies_ttQL}, \cref{thm:ttQLcv_LWE}, \cref{thm:relaxed_watermarking_cnc_LWE}, and \cref{thm:MAC_information_theoretic}, we obtain the following theorem.
\begin{theorem}\label{thm:SSL_LWE_cnc}
Let $\beta$ be any inverse polynomial of $\secp$ and $\eta>0$ any constant.
Assuming the hardness of the LWE problem against sub-exponential time quantum adversaries, there exists a $(\beta,\Dcnc{\secp^\eta})$-perfect finite-term lessor secure SSL scheme for the class of compute-and-compare circuits $\Ccnc$, where $\Dcnc{\secp^\eta}$ is any distribution over $\Ccnc$ that has conditional min-entropy $\secp^\eta$.
\end{theorem}
Also, by combing \cref{thm:SSLcc_perfect} with \cref{thm:ttQLcv_LWE}, \cref{thm:relaxed_watermarking_cnc_LWE}, and \cref{thm:MAC_information_theoretic}, we obtain the following theorem.
\begin{theorem}\label{thm:SSLcc_LWE_cnc}
Let $\beta$ be any inverse polynomial of $\secp$ and $\eta>0$ any constant.
Assuming the hardness of the LWE problem against sub-exponential time quantum adversaries, there exists a $(\beta,\Dcnc{\secp^\eta})$-perfect finite-term lessor secure SSL scheme with classical communication for the class of compute-and-compare circuits $\Ccnc$, where $\Dcnc{\secp^\eta}$ is any distribution over $\Ccnc$ that has conditional min-entropy $\secp^\eta$.
\end{theorem}

\else\fi

	\ifnum\anonymous=0
	\ifnum\acknowledgments=1
	\fi
	\fi
	\ifnum\llncs=1
	\bibliographystyle{auxiliary/myalpha}
	\bibliography{auxiliary/abbrev3,auxiliary/crypto,auxiliary/other-bib}
	\else
	\ifnum\choosebibstyle=0
	\else
	\ifnum\choosebibstyle=1
	\bibliographystyle{alpha}
	\else
	\bibliographystyle{abbrv}
	\fi
	\newcommand{\etalchar}[1]{$^{#1}$}

	\fi
	\fi

	
\ifnum\cameraready=0
	\ifnum\llncs=0
	\appendix

\else
	\newpage
	 	\appendix
	 	\setcounter{page}{1}
 	{
	\noindent
 	\begin{center}
	{\Large SUPPLEMENTAL MATERIALS}
	\end{center}
 	}
	\setcounter{tocdepth}{2}

	\setcounter{tocdepth}{1}

	\fi
\fi

\ifnum\submission=0
\else
\ifnum\cameraready=1
\else

\fi
\fi

\end{document}